\documentclass[11pt]{article}
\pdfoutput=1
\usepackage{amsmath,amsthm,amssymb,epsfig,amsfonts,mathrsfs}
\usepackage{graphicx}
\usepackage{subfig}
\usepackage[usenames, dvipsnames]{color}
\usepackage{cite}
\usepackage{multirow}
\usepackage{hyperref}
\usepackage{verbatim} 
\usepackage{enumitem}
\usepackage{rotating}
\usepackage{tikz-cd}
\usepackage{tikz}
\usetikzlibrary{decorations.pathmorphing}
\usepackage{xcolor}
\usepackage{multirow}
\usepackage{bm}
\usepackage[normalem]{ulem}
\usepackage{cleveref}
\usepackage{slashed}
\usepackage{dynkin-diagrams}

\usepackage{stackengine, array}

\usepackage{tikz} 
\usepackage{physics}

\makeatletter

\usepackage{tikz}
\usetikzlibrary{arrows}
\usetikzlibrary{arrows.meta}
\usetikzlibrary{shapes.geometric,calc,arrows, positioning,shapes.misc,decorations.markings}
\tikzset{
  big arrow/.style={
    decoration={markings,mark=at position 1 with {\arrow[scale=2,#1]{>}}},
    postaction={decorate},
    shorten >=0.4pt},
  big arrow/.default=black}
  
\pgfdeclarelayer{edgelayer}
\pgfdeclarelayer{nodelayer}
\pgfsetlayers{edgelayer,nodelayer,main} 
\tikzstyle{none}=[inner sep=0pt] 

\tikzstyle{NodeCross}=[draw, shape=circle, cross out, inner sep=0pt, minimum size=6pt,line width=0.25mm]
\tikzstyle{Circle}=[draw, shape=circle, black,  fill=black, inner sep=0pt, minimum size=6pt]
\tikzstyle{Star}=[draw, shape=star, fill=black, star points=8, inner sep=0pt, minimum size=8pt]

\tikzstyle{DashedLine}=[-, densely dashed, line width=0.25mm]
\tikzstyle{DottedLine}=[-, dotted, line width=0.25mm]
\tikzstyle{ThickLine}=[-, line width=0.25mm]
\tikzstyle{ArrowLineRight}=[-, -{Stealth[scale=1.75]}, line width=0.1mm, scale=5]
\tikzstyle{RedLine}=[-, draw={rgb,255: red,191; green,0; blue,0}, fill=none, line width=0.25mm]
\tikzstyle{DottedRed}=[-, dotted, draw={rgb,255: red,191; green,0; blue,0}, fill=none, line width=0.25mm]
\tikzstyle{DashedLineThin}=[-, densely dashed, line width=0.125mm, fill=none, draw=black]
\tikzstyle{ArrowLineRed}=[-, -{Stealth[scale=1.75]}, draw={rgb,255: red,191; green,0; blue,0}, line width=0.25mm, scale=5]

\newtheorem{definition}{Definition}[section]
\newtheorem{theorem}{Theorem}
\newtheorem{example}{Example}

\newcommand{\be}{\begin{equation}}
\newcommand{\ee}{\end{equation}}
\newcommand{\ba}{\begin{aligned}}
\newcommand{\ea}{\end{aligned}}

\newcommand{\cN}{\mathcal{N}}

\newcommand{\bea}{\begin{eqnarray}}
\newcommand{\eea}{\end{eqnarray}}

\newcommand{\U}{{\rm U}}

\newcommand{\Z}{{\mathbb Z}}

\def\ket#1{{|{#1}\rangle}}

\def\unit{{1\kern-.65ex {\rm l}}}
\def\1{{1\kern-.65ex {\rm l}}}

\def\CF{{\cal F}}
\def\CG{{\cal G}}
\def\CH{{\cal H}}

\def\CP{{\cal P}}

\def\CU{{\cal U}}

\def\bbR{{\mathbb{R}}}

\def\bbZ{{\mathbb{Z}}}

\newcount\hour \newcount\minute
\hour=\time \divide \hour by 60
\minute=\time
\def\now{%
\ifnum \hour<13
  \ifnum \hour=0 \advance \hour by 12 \number\hour:\else \number\hour:\fi%
     \ifnum \minute<10 0\fi%
     \number\minute%
\ A.M.%
\else \advance \hour by -12 \number\hour:%
  \ifnum \minute<10 0\fi%
  \number\minute%
  \ P.M.%
\fi%
}

\makeatother

\def\mb{\mathbb}
\def\mbf{\mathbf}
\def\mc{\mathcal}
\def\bp{\begin{pmatrix}}
\def\ep{\end{pmatrix}}

\def\ptl{\partial}

\def\Cech{${\rm \check{C}ech} \ $}
\def\tA{\Tilde{A}}
\def\dd{{\rm d}}
\newtheorem{prop}{Proposition}
\def\pmu{\partial_\mu}
\def\Fmn{F^{\mu\nu}}
\def\stack{\stackrel}

\begin{document}

\baselineskip=18pt  
\numberwithin{equation}{section}  
\allowdisplaybreaks  

\thispagestyle{empty}

\vspace*{0.8cm} 
\begin{center}
{\huge Lecture Notes on Generalized Symmetries and Applications}\\

 \vspace*{1.5cm}

 {\large Ran Luo$^1$, Qing-Rui Wang$^2$, Yi-Nan Wang$^{1,3}$}

\vspace*{1.0cm}

\smallskip
{\it $^1$ School of Physics,\\
Peking University, Beijing 100871, China\\ }

\smallskip

{\it  $^2$ Yau Mathematical Sciences Center, Tsinghua University,\\
Beijing 100084, China}

\smallskip

{\it  $^3$ Center for High Energy Physics, Peking University,\\
Beijing 100871, China}

\vspace*{0.8cm}
\end{center}
\vspace*{.5cm}

\noindent
In this lecture note, we give a basic introduction to the rapidly developing concepts of generalized symmetries, from the perspectives of both high energy physics and condensed matter physics. In particular, we emphasize on the (invertible) higher-form and higher group symmetries. For the physical applications, we discuss the geometric engineering of QFTs in string theory and the symmetry-protected topological (SPT) phases in condensed matter physics. 

The lecture note is based on a short course on generalized symmetries, jointly given by Yi-Nan Wang and Qing-Rui Wang in Feb. 2023, which took place at School of Physics, Peking University (https://indico.ihep.ac.cn/event/18796/). 
\newpage


\tableofcontents


\newpage

\section{Introduction}

Symmetry is a fundamental guiding principle of physics. Under a symmetry transformation, the physical theory remains invariant, and the physical ingredients (fields, operators) are organized into different representations of the symmetry. There are two different notions of symmetries:
\begin{enumerate}
\item{Global symmetry: the symmetry parameter $g$ is spacetime independent. A global symmetry is usually considered as an ``actual symmetry'' of a physical system.}
\item{Local symmetry (gauge symmetry): the symmetry parameter $g$ is spacetime dependent. Such a local symmetry is also typically referred as the gauge symmetry, which is interpreted as a redundancy of the theory, but not an actual physical symmetry.\footnote{There are also symmetries in between, such as the ``Galileon symmetry''\cite{Nicolis:2008in,Hinterbichler:2015pqa} $\phi\rightarrow\phi+a$ in the case of a free real scalar field $\int \ptl_\mu\phi\ptl^\mu\phi$. The parameter $a$ can be chosen to only satisfy $\ptl_\mu\ptl^\mu a=0$. This is an example of a ``semi-local'' symmetry without a gauge field.}}
\end{enumerate}

Traditionally, the notion of ``ordinary symmetry'' denotes the action of a symmetry group\footnote{A group is a set equipped with a closed, associative binary operation (denoted as $\cdot$), has a unit under $\cdot$ and all elements are invertible under $\cdot$ .} on a local operator\footnote{\emph{id est}, an operator that only depends upon a point on the spacetime manifold.}. This is referred as an invertible, 0-form  symmetry in the modern language.

A simple example of ordinary symmetry (the meaning of ordinary will be explained later) would be

\begin{example}
\[ \mathcal{L} = |\partial_\mu \phi|^2 \]
where $\phi$ is a complex scalar field, a symmetry transformation here is $\phi\mapsto e^{i\alpha}\phi$, $e^{i\alpha}\in U(1)$.
\end{example}

In the past decade, there have been a plethora of efforts to generalize the notion of symmetries~\cite{Gaiotto:2014kfa}, from the perspectives of condensed matter physics, high energy physics and pure mathematics, see the following review articles and lecture notes\cite{Kong_2020,mcgreevy2022generalized,Cordova:2022ruw,Gomes:2023ahz,Schafer-Nameki:2023jdn,Brennan:2023mmt,Bhardwaj:2023kri}. From the definition of the ordinary symmetry, we observe some clues on how to generalize it:
\begin{enumerate}
    \item From 0-form symmetry to higher-form symmetries. We intend to let the symmetry operations act on ``objects that extends $\ge 1$ dimensions'', such as Wilson loops, 't Hooft loops and other higher dimensional generalizations. Furthermore, one may combine symmetries of different forms into an algebraic structure called ``higher group''.
    \item Beyond the group structure. Instead of group, we use a weaker structure to describe the symmetry. We can relax the  requirements such as  invertibility, associativity until we are left with a very complicated structure that can only fit in the general framework of higher category theory. Such symmetries are denoted as non-invertible symmetries or more generally, higher categorical symmetries. They are subject to active research in both physics and mathematics~\cite{Kong_2020,ji2020categorical,Johnson-Freyd:2020usu,Roumpedakis:2022aik,Bhardwaj:2022yxj,Bhardwaj:2022lsg,Bartsch:2022mpm,Bhardwaj:2022kot,Kaidi:2022cpf,Decoppet:2022dnz,Bhardwaj:2022maz,Bartsch:2022ytj,Delcamp:2023kew,Kaidi:2023maf,Bhardwaj:2023wzd,Bartsch:2023pzl,Bhardwaj:2023ayw,Bartsch:2023wvv,Decoppet:2023bay}.
\end{enumerate}


In this lecture note, we provide an introduction of generalized global symmetries from different perspectives. The structure of the note is as follows: in section~\ref{sec:ordinary}, we introduce the framework of topological generators for global symmetries, in the context of ordinary 0-form symmetries. In section~\ref{sec:higher-form}, we give the definition of invertible higher-form symmetries and discuss the examples of pure Maxwell theory, Maxwell theory with charged matter and non-abelian gauge theories. In section~\ref{sec:gauging}, we discuss the gauging of higher-form symmtries, 't Hooft anomalies and the examples of Maxwell theory and 3d Chern-Simons theory. In section~\ref{sec:three}, we give a more comprehensive understanding of symmetries generators using three equivalent languages: topological defect network, flat connection and classifying space. They are particularly important for the discussions of finite group symmetry in the later parts.

In section~\ref{sec:anomaly-SPT}, we introduce the notion of symmetry-protected topological (SPT) phases in condensed matter physics and 't Hooft anomaly using group cohomology language. Examples of 1+1D Haldane chain coupled to 0+1D spin-$\frac{1}{2}$ boundary and 2+1D SPT are presented. In section~\ref{sec:gen-SPT}, we discuss various generalizations of SPT phases, including higher-form symmetries and fermionic systems. We also present the comprehensive cobordism classification framework for the classification.

In section~\ref{sec:applications} we discuss further applications of higher-form symmetries. In section~\ref{sec:string} we talked about the geometric engineering of QFTs in string /M-theory framework, and the computation of higher-form symmetries from the topology of the extra-dimensional space. In particular we present an example of 5d $\mc{N}=1$ SCFT from 11D M-theory on Calabi-Yau 3-fold singularities. In section~\ref{sec:appl-CMP} we discuss some condensed matter physics applications, including 2+1D toric code and higher-form SPT in 3+1D.

Finally in section~\ref{sec:higher-group}, we extend our scope to more general, categorical symmetries with particular emphasis on higher-group symmetries. We present basic knowledge of category theory in section~\ref{App:Cat}. We define strict 2-group in the language of 2-category theory in section~\ref{sec:2-group}, and its relation to weak 2-groups in section~\ref{sec:cross-extension} and section~\ref{sec:Postnikov}. We give physical interpretation of weak 2-group symmetry in section~\ref{sec:2-group-phys} and their gauging in section~\ref{sec:2-group-gauging}. Weak $n$-groups and non-invertible symmetries are briefly described in section~\ref{sec:n-group} and \ref{sec:non-invertible}.

At the beginning of each section, we will provide guidance to help readers navigate and make choices about how to engage with the content of this section.



\section{Higher-Form Symmetry}
\label{sec:higher-form-gen}

In this section, we offer a fundamental introduction to higher-form symmetry. We begin with a review of topological operators associated with ordinary symmetries, setting the stage for deeper exploration. Following this, we present definitions and illustrative examples of higher-form symmetries. The latter part of this section consists of the gauging of higher-form symmetries and three other useful perspectives on higher-form symmetries.

Given that this section lays the groundwork for subsequent discussions, we encourage readers to thoroughly engage with this content. This section is designed to be accessible, aimed at helping readers establish a  foundational knowledge before they progress to the more complex topics in the subsequent sections, which can be chosen based on their interest.

\subsection{Topological Operators in Ordinary Symmetry}
\label{sec:ordinary}

In the modern language, the generators of a symmetry is described as topological operators, which we now  introduce. In the following, we denote the spacetime manifold as $M$ (or $M^d$ if we want to emphasize it's $d$-dimensional). We denote the symmetry operator by $U_g(X)$, where  $g\in G$ is an element of the symmetry group, $X\subseteq M$ is the manifold which supports the topological operator  (which will be explained later).

Now we consider a basic example of quantum mechanics
\begin{example}
    $M^{0+1} = \mb{R}$, $U_g=U_g(t)$ is the symmetry operator defined at a point $t\in \bbR$. Since
    \[ [U_g,H] = 0 \ , \]
    we usually choose to omit the time dependence. Now consider an object (local operator in this case) charged under this symmetry, $O_{\{ i\}}$, and the adjoint action on this object
    \[ U_g O_i(t) U_g^{-1} = \mathcal{R}_i^j(g) O_j(t)  \]
    gives a representation $\mathcal{R}$. A graphical understanding of this formula is shown in Fig.~\ref{Fig:0FormQM}.
\end{example}

\begin{figure}[htbp]
\begin{center}

\begin{tikzpicture}
\draw[thick, ->] (-2,-1) -- (-2,1);
\filldraw[orange] (-2,-0.5) circle (2pt) node[anchor=west]{$U_g(t-\epsilon)$};
\filldraw[orange] (-2,0.5) circle (2pt) node[anchor=west]{$U_g^{-1}(t+\epsilon)$};
\filldraw[purple] (-2,0) circle (2pt) node[anchor=west]{$O_i(t)$};
\draw[thick, ->] (2,-1) -- (2,1);
\filldraw[purple] (2,0) circle (2pt) node[anchor=west]{$O_j(t)$};
\node[] at (0.5,0) {$ =  \ \ \ \ \mathcal{R}_i^j(g)$};
\end{tikzpicture}
\caption{}
\label{Fig:0FormQM}
\end{center}

\end{figure}

So what do we mean by ``topological''? Since $U_g$ commutes with the Hamiltonian, it can move along the time axis, as long as the ``obstruction'' $\{ O_i(t)\}$ does not get in the way. This topological property will become more explicit in higher dimensional examples.

\begin{example}[0-form symmetry in $M^d$]
    To generalize the previous example, we define 0-form symmetry to be generated by a set of unitary topological operators supported on $M^{(d-1)}\subset M^d$, with $U_g(M^{(d-1)})U_{g'}(M^{(d-1)}) = U_{gg'}(M^{(d-1)})$. The action is defined by
    \[ U_g(S^{d-1}) O_i(p) =  \mathcal{R}_i^j(g) O_j(t) \ ,  \]
    where $p$ is a point in the interior of $S^{d-1}$.
\end{example}
This seems odd at your first glimpse, as the previous ``$U_g^{-1}$'' part is gone, but it is actually just a generalization of the previous case. The graphical  understanding is depicted in figure \ref{Fig:1FormQM} for the $1+1$-dim case, as for higher dimensions one can imagine a similar picture.

\begin{figure}[htbp]
\begin{center}

\begin{tikzpicture}
\draw[thick, ->] (-4,-1.5) -- (-4,1.5) node[anchor = south] {t};
\draw[orange,thick] (-5,-0.25) -- (-4.25,-0.25) ;
\draw[orange,thick] (-5,0.25) -- (-4.25,0.25) ;
\draw[orange,thick] (-3.75,-0.25) -- (-2.75,-0.25) node[anchor=north]{$U_g(t-\epsilon)$};
\draw[orange,thick] (-3.75,0.25) -- (-2.75,0.25) node[anchor=south]{$U_g^{-1}(t+\epsilon)$};
\draw[orange,thick] (-3.75,0.25) arc (0:180:0.25);
\draw[orange,thick] (-4.25,-0.25) arc (180:360:0.25);

\draw[thick, ->] (0,-1.5) -- (0,1.5) node[anchor = south] {t};
\filldraw[purple] (0,0) circle (2pt) ;
\draw[orange,thick] (0.25,0) arc (0:360:0.25)node[anchor=west]{$U_g(S^1)$};
\node[] at (-2,0) {$ =$};

\filldraw[purple] (-4,0) circle (2pt) node[anchor=west]{$O_i(t)$};
\draw[thick, ->] (4,-1.5) -- (4,1.5) node[anchor = south] {t};
\filldraw[purple] (4,0) circle (2pt) node[anchor=west]{$O_j(t)$};
\node[] at (3,0) {$ =  \ \ \mathcal{R}_i^j$};
\end{tikzpicture}
\caption{}
\label{Fig:1FormQM}
\end{center}

\end{figure}

\begin{example}[Electromagnetism in $M^d$]
Consider the $U(1)$ conserved current in $d$-dim pure Maxwell theory,
\[ \dd *F = j \]
then 
\[ U_g(M^{(d-1)}) = g^{\int_{M^{(d-1)}} j}  = \exp(i\alpha \int_{M^{(d-1)} } j ) \]
where $g = \exp(i\alpha)\in U(1)$ is a symmetry group element\footnote{In fact, the $U(1)$ mentioned in the previous case actually corresponds to the global part of $U(1)$ gauge symmetry, when $\alpha$ is spacetime independent.}. 
\end{example}

If $M^{(d-1)}$ is taken as a time slice, then $\int_{M^{(d-1)} } j$ is the total electric charge of this system. How is this operator topological? It follows from the Stokes theorem. Consider a small deformation to $M^{(d-1)}$ (small enough so that no point charges went through) and denote it $M^{(d-1)\prime}$, then
\[  \int_{M^{(d-1)}} j -  \int_{M^{(d-1)\prime}} j   = \int_{\partial D} j = \int_{D} \dd j = 0  \ , \]
where $D$ is the volume enclosed by the deformed part. The $U_g$ operator is indeed topological.

\subsection{Higher-Form Symmetry}
\label{sec:higher-form}

In the language of topological operators, the generalization to higher-form symmetry is straight-forward, we simply have to substitute the codim-$1$ submanifold (i.e. dim-($d-1$) submanifold) corresponding to symmetry operators with codim-$(p+1)$ manifold, and substitute the local operators with  objects defined on $p$-dimensional submanifolds ($p$-dimensional extended operators). In this fashion, the topological operator is denoted $U_g(M^{d-p-1})$, the object by $V_(C^{p})$, and then the symmetry action on object is
\be
U_g(M^{d-p-1})V_i(C^p) = R_i {}^{j}(g) V_j(C^p) \ ,
\ee
if the linking number $\langle M^{d-p-1},C^{p} \rangle = 1$. Here the linking number is intuitively understood as how the two submanifolds (with dimension adding up to $d-1$) entangles but not intersecting with each other topologically, a more clear definition is
\be
\langle M^{d-p-1},C^p \rangle = {\rm Int}(X^{d-p}, C^p) = {\rm Int}(M^{d-p-1}, Y^{p+1}) \ ,
\ee
where $\partial X^{d-p} = M^{d-p-1}$, $\partial Y^{p+1} = C^{p}$ are the submanifolds enclosed by the $M$ and $C$ respectively, and ${\rm Int}$ denotes intersection number. The intuition for formula is 

\begin{center}
    \begin{tikzpicture}
    \draw[thick, cyan] (-2,0.1)--(-2,2);
    \draw[thick, cyan] (-2,-1)--(-2,-0.1);
    \draw[thick, orange] (-2.1,1) arc (-265:85:1 and 0.5);
    \node[] at (0,0.5) {=  Int};
    \node[] at (-4,0.5) {Link};

    \filldraw[orange!50] (2.0,0.5) ellipse (1 and 0.5);
    \draw[thick, cyan] (2,0.5)--(2,2);
    \draw[thick, cyan] (2,0)--(2,-1);
    \end{tikzpicture}
\end{center}

Of course, as a representation of the symmetric group on this system, we shall have
\be
U_g(M^{d-p-1})U_{g'}(M^{d-p-1})=U_{gg'}(M^{d-p-1}) \ .
\ee

Next we argue that if $p>0$ and the spacetime topology is trivial, the higher-form symmetry is always \textbf{Abelian}. In the language of topological operators, the argument is almost immediate: suppose you have two operators with codimension $\ge 2$, then one can topologically interchange them in the spacetime without any obstruction. A illustrated case can be seen in Fig.~\ref{Fig:Abelian}. However in the case of $p=0$ we cannot perform such an interchange, as the two topological operators $U_g$ are supported on two different time slices, and they can be generally non-commutative.

\begin{figure}[htbp]
    \centering
    \begin{tikzpicture}
        \draw[->]   (0,0)--(0,2)      node[anchor = south] {t};
        \draw[->]   (0,0)--(2,0)      node[anchor = west] {x};
        \draw[->]   (0,0)--(-0.6,-1.2)      node[anchor = east] {y};
        \draw[thick, orange] (1.1,1.4)--(0,-0.8)  node[anchor=west] {$U_{g_1}$};
        \draw[thick, magenta] (1.1,2.4)--(-0.1,0)  node[anchor=east] {$U_{g_2}$};
        \node[] at (3,1) {=};
        \draw[->]   (4,0)--(4,2)      node[anchor = south] {t};
        \draw[->]   (4,0)--(6,0)      node[anchor = west] {x};
        \draw[->]   (4,0)--(3.4,-1.2)      node[anchor = east] {y};
        \draw[thick, magenta] (5.1,1.4)--(4,-0.8)  node[anchor=west] {$U_{g_2}$};
        \draw[thick, orange] (5.1,2.4)--(3.9,0)  node[anchor=east] {$U_{g_1}$};

    \end{tikzpicture}
    \caption{Interchanging topological operators in $p>0$}
    \label{Fig:Abelian}
\end{figure}

\subsubsection{Example: Maxwell Theory}
\label{subsec:Maxwell}
We consider the pure $U(1)$ Maxwell theory in 4 dimensions, the action is
\be
S = \frac{1}{4e^2} \int {\rm d}^4 x \ F_{\mu\nu}F^{\mu\nu}
\ee
where $F_{\mu\nu} = \partial_\mu A_\nu- \partial_\nu A_\mu$ or equivalently, $F = \frac{1}{2} F_{\mu\nu}{\rm d}x^\mu \wedge {\rm d}x^\nu = {\rm d}A $. It is widely known that this theory has a $U(1)_E \times U(1)_M$ 1-form symmetry. 

We show that $U(1)_E$ is a 1-form symmetry, which is generated by $4-1-1=2$-dimensional topological operators and acts on one-dimensional objects. The electric $U(1)_E$ symmetry is generated by the topological operator
\be
\label{U1-UE}
U_E (S) = \exp(i\alpha Q_E(S))
\ee
where $Q_E(S)$ is the total electric charge enclosed by $S$, namely
\be
Q_E(S) = \frac{1}{e^2}\int_S *F = \frac{1}{e^2}\int_S \frac{\epsilon_{\mu\nu\rho\sigma}}{2!\cdot 2!}F^{\mu\nu} {\rm d}S^{\rho\sigma} \sim  \int_S \vb{E}\cdot \vb{{\rm d}S}\,.
\ee
And this operator is indeed topological because of the equation of motion ${\rm d}*F = 0 $, a deformation not hitting electric charges would not change the operator.

Then what is the object on which the operators act? The answer is the Wilson loop. For a given loop $\gamma:[0,1]\to M, \gamma(1)=\gamma(0)$, the Wilson loop operator is
\be 
W_n(\gamma) = \exp(i n  \oint_\gamma A_\mu {\rm d}x^\mu )  = \exp(i n \int_\gamma A )\,. 
\ee
The topological operators (\ref{U1-UE}) act on Wilson loops by (proof will be given later in Proposition \ref{Prop:1})
\be
\expval{U_E(S) W_n(\gamma)} = \exp(in\alpha \expval{S,\gamma}) \expval{W_n(\gamma)}\,.
\ee

\textbf{Remark}:
\begin{enumerate}
\item In the definition of Wilson loop, $n\in \bbZ$. The reason for the quantization is that under a large gauge transformation $A_\mu\mapsto A_\mu + \partial_\mu \epsilon$,
    \[ \oint_\gamma \partial_\mu \epsilon {\rm d}x^\mu  = 2\pi k \quad \ k\in \bbZ\,, \]
    where $k$ is the winding number.  Hence the Wilson loop is only gauge invariant when $n\in\mb{Z}$.
    
    \item This action of topological operator on Wilson loop can be  equivalently treated as shifting $A$ by a flat connection $\lambda$ (flat refering to ${\rm d}\lambda = 0$), namely 
    \be
    \oint_\gamma \lambda = \alpha \expval{S,\gamma}\,.
    \label{Eq:intlambda}\ee
    Also note that the integration is invariant under small deformation of $\gamma$.

    Note that even if $\lambda$ is flat, it cannot be written as a local gauge transformation $\lambda=d\epsilon$, otherwise the symmetry action is trivial.
    
\end{enumerate}

Now we turn to the magnetic $U(1)_M$ symmetry dual to the electric one. The trick here is simple: add the Hodge star $*$ to everything. Obviously we shall have
\be
U_M (S) = \exp(i\alpha Q_M(S))\,,
\ee
\be Q_M(S) = \frac{1}{2\pi}\int_S F = \frac{1}{4\pi}\int_S F_{\mu\nu} {\rm d}S^{\mu\nu} \sim  \int_S \vb{B}\cdot \vb{{\rm d}S} \,.
\ee
This operator is topological because ${\rm d}F = {\rm d^2}A = 0$, and to create a dual version of Wilson line, we define $\tA$ s.t.  ${\rm d}\tA = *\dd A $. The 't Hooft loop is therefore defined as
\be 
T_n(\gamma)  = \exp(i n \oint_\gamma \tA) = \exp(i n \oint_\gamma \tA_\mu \dd x^\mu)\,.
\ee
And of course, the action of $U_M$ operators on 't Hooft loops is
\be
\expval{U_M(S) T_n(\gamma)} = \exp(in\alpha \expval{S,\gamma}) \expval{T_n(\gamma)}\,.
\ee
This also can be interpreted as shifting $\tA$ by a flat connection.

\begin{prop}
\label{Prop:1}
\[ \expval{U_E(S) W_n(\gamma)} = \exp(in\alpha \expval{S,\gamma}) \expval{W_n(\gamma)}  \]
\end{prop}
\begin{proof} 
\be {\rm L.H.S} = \int \mathscr{D}[A] \exp{iS + \frac{i\alpha}{e^2}\int_S *F + in \oint_\gamma A} \ee
The idea here is to make a shift to $A$, under which the measure $\mathscr{D}[A]$ is invariant and the exponential power terms varies by a total derivative (which does not contribute given our trivial spacetime) as well as a term involving the linking number. Consider
\begin{align*}
    \int_S *F &= \int_{V_S(\partial V_S = S)} d*F = \int_{V_S} \frac{\epsilon_{\nu\rho\sigma\tau}}{3!} \partial_\mu F^{\mu\nu} \dd v^{\rho\sigma\tau}\\
    &= \int \dd^4 x \ \partial_\mu F^{\mu\nu} J_\nu(V_S) = -\int\dd^4 x \ F^{\mu\nu} \partial_\mu J_\nu(V_S)
\end{align*} 
where we define
\[ J_\nu(V_S) = \int_{V_S} \frac{\epsilon_{\nu\rho\sigma\tau}}{3!} \delta^4(x-y) \dd v^{\rho\sigma\tau}(y)\,. \]
We see how those terms expands around $A-\alpha J(V_S)$, and the corresponding results are
\begin{align}
    W_n(\gamma,A) &= W_n(\gamma,A-\alpha J(V_S)) \exp(in\oint_\gamma J(V_S)) \\
    S[A] &= S[A-\alpha J(V_S)] + \frac{\alpha}{e^2} \int \dd^4 x \ \pmu J_\nu(V_S) \Fmn \ (+ {\rm higher\  derivative \ terms\ in\ }J_\mu) \ .
\end{align}

The higher derivative terms in $J_\mu$ has no dependence on $A$, hence it was omitted as a normalization factor. 

Finally after shifting $A\mapsto A-\alpha J(V_S)$, we arrive at
\be
\expval{U_E(S) W_n(\gamma)} = \exp(in\alpha \oint_\gamma J(V_S)) \expval{W_n(\gamma)} \ .
\ee
Through the definition involving intersection number, one arrive at
\[\oint_\gamma J(V_S) = \expval{S,\gamma}\,.\]

\end{proof}

\subsubsection{Example: $U(1)$ Gauge Theory with Charged Matter}
Now we add the Maxwell theory with charged matters $\{ \phi_i\}$, with electric charge $q_i \in \bbZ$ which means $R_{q_i} (e^{i\alpha}) = e^{iq_i\alpha}$, also we define a positive integer $N = {\rm g.c.d}(q_i)$. 

The main difference in this scenario is that the symmetry current $*F$ of $U(1)_E$ is no longer conserved, instead we have
\be \frac{1}{e^2} \dd *F = *j_E \ , \ee
where $j_E$ is the usual electric current 1-form. 
However, if we still require that the operators
\be U_E(e^{i\alpha}, S) = \exp(\frac{i\alpha}{e^2} \int_S *F) \ , \ee
are topological, we would require that any deformation of $S$ should remain invariant, see Fig.\ref{fig:VS}
\begin{figure}[htbp]
    \centering
    \includegraphics[width = 0.6 \textwidth]{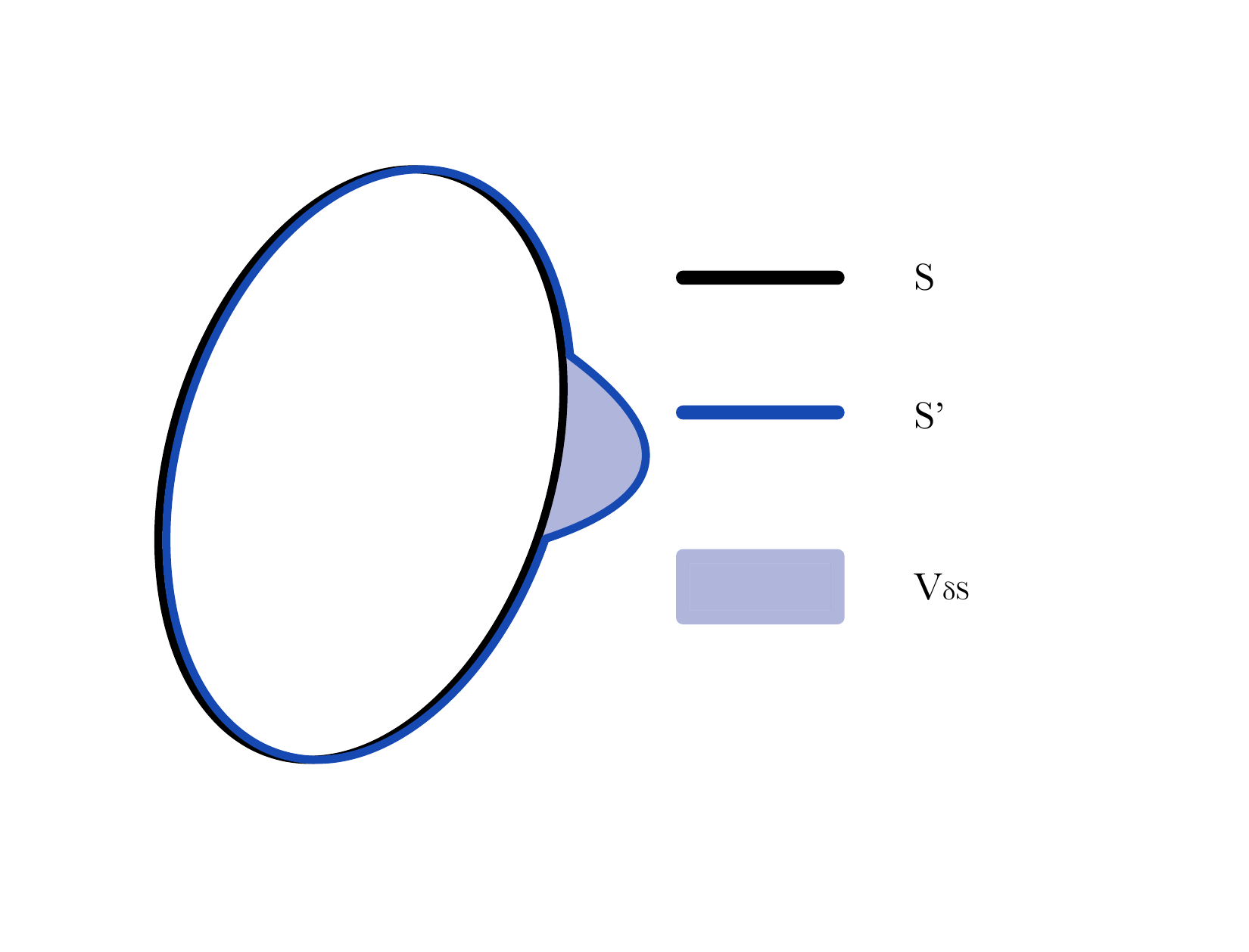}
    \caption{This figure visualizes the deformation from a closed manifold $S$ to $S'$, which are labelled by black and thick blue line. The two surfaces would enclose a volume, upon which the flux integrates to become the difference of the two associated topological operators in Eq.(\ref{eq:VS}).}
    \label{fig:VS}
\end{figure}
\be\ba 
U_E(e^{i\alpha}, S) &= U_E(e^{i\alpha}, S') \\
\Rightarrow \exp(\frac{i\alpha}{e^2}\int_{V_{\delta S}} \dd * F ) &= \exp(i\alpha\int_{V_{\delta S}} \ * j_E ) = 1 \quad \ptl V_{\delta S} = S \cup \overline{S'} \label{eq:VS}
\ea\ee
Charge quantization gives $\int_{V_S} \ *j_E \in N\bbZ $, so in order to satisfy (\ref{eq:VS}), we require that $\alpha \in 2\pi \bbZ_N / N$.

Thus, the $U(1)_E$ electric 1-form symmetry group is broken to a $\bbZ_N$ subgroup. It's also worth noting that $U(1)_M$ is left unchanged when we do not include magnetic monopoles.

This formalism can be generalized into a $U(1)^r$ gauge theory with charged matter fields ${\phi_i}$. The charge of the field $\phi_i$ is  $q_{i,j}\in\bbZ$ under the $j$-th $U(1)$ symmetry. In order to compute the symmetry breaking, we denote the charge matrix by $Q = (q_{i,j})$, where $Q$ is not necessarily a square matrix. The Smith normal decomposition gives us the following result:
\be
\label{SmithDecomp}
U Q V = D = \mqty( \dmat{d_1, d_2, ... ,d_k, 0 & 0 & 0 & ...\\ 0 & 0 & 0 & ...}  )\,.
\ee

Here $U$ and $V$ are invertible square matrices with integer matrix elements. The 1-form symmetry resulted from this is $(\bigoplus_{i=1}^k \bbZ_{d_i})\oplus U(1)^{r-k} $. The magnetic 1-form symmetry is still $U(1)_M=U(1)^r$.

\begin{example}[A theory with $\mb{Z}_3$ electric 1-form symmetry]
As an example, we consider a $U(1)^2$ gauge theory with matter fields whose charges are $\phi_1=(2,-1)$, $\phi_2=(-1,2)$. 
\end{example}

The charge matrix is hence
\be
Q=\bp 2 & -1\\-1 & 2\ep\,.
\ee
After the Smith normal decomposition, we get $D=\rm{diag}(3,1)$, hence the electric one-form symmetry is broken to the $\mb{Z}_3$ subgroup.

In fact, this is exactly the Coulomb phase of a pure $SU(3)$ gauge theory.

\subsubsection{Non-Abelian Gauge Theory in 4D}

Another well known examples are the non-abelian gauge theories, which we point the reader to the gauge theory lecture note by David Tong~\cite{TongGauge}. We consider the pure Yang-Mills theories with gauge group $G$ for simplicity.

For $G = SU(N)$, $\bbZ_N$ is the 1-form electric symmetry generated by the center of $SU(N)$, which acts on Wilson loops  charged under this symmetry. 


Now we discuss the case of finite temperature system on $S^1\times\mb{R}^3$, where the Euclidean time circle satisfies $x^0\sim x^0+\beta$. The $\mb{Z}_N$ 1-form center symmetry is defined in terms of an equivalence class of functions $[h_k]=\{h_k(x^\mu)\in SU(N), h_k(x^\mu)\sim g(x^\mu)h_k(x^\mu)g^{-1}(x^\mu)\}$, which is periodic up to the center of $SU(N)$: $h_k(x^0=\beta)=e^{2\pi k/N}h_k(x^0=0)$. It acts on $A$ as 
\be
A\rightarrow h_kAh_k^{-1}+h_kdh_k^{-1}\,,
\ee
which acts trivially on local operators but non-trivially on the Polyakov loop
\be
P_R=\text{tr}_R(\mathcal{P}\exp\left(\oint_{S^1}A^0\right)
\ee
as
\be
P_R\rightarrow P_R e^{\frac{2\pi i k n_R}{N}}\,.
\ee
Here $R$ is a representation of $SU(N)$ which defines the Wilson loop (Polyakov loop), and $n_R$ is called the ``N-ality'', which equals to the total number of boxes in the Young diagram of $R$ mod $N$.

This action of $\mb{Z}_N$ 1-form symmetry also applies to general Wilson loops:
\be
\ba
U_{k(\in \bbZ_N)}(S^2) W_R(\gamma) &= W_R(\gamma) \Big[A + \frac{2\pi k}{N} I \Big] \quad \expval{S^2 ,\gamma} = 1\cr
&=e^{\frac{2\pi ik n_R}{N}}W_R(\gamma)\,.
\ea
\ee

In the $G = PSU(N) = SU(N)/\bbZ_N$ case, the system exhibits a $\bbZ_N$ magnetic symmetry acting on 't Hooft loops, due to the fact that $\pi_1(PSU(N))=\mb{Z}_N$ \cite{TongGauge}. We can see that in the following description. Consider  magnetic monopole in pure space (a 't Hooft line in spacetime), Wu-Yang monopole construction tells us that by separating the space into northern/southern hemisphere, the gauge potentials are
\be\ba
A_N = &B_0 \frac{1+\cos(\theta)}{2 r^2} \dd \phi \\
A_S = &B_0 \frac{-1+\cos(\theta)}{2 r^2} \dd \phi \\
\ea\ee
and the patching condition at the equator
\be \exp(i B_0\phi) =  \exp(i B_0(\phi+2\pi) )\Rightarrow \exp(2\pi i B_0 ) = I_{PSU(N)} \ee
such choice of $B_0$ is classified by $\pi_1(PSU(N)) = \bbZ_N$.

We note that in the $SU(N)$ theory, there is a center $\bbZ_N^{(1)}$ symmetry but no such magnetic symmetry, since $\pi_1(SU(N)) = 0$. Meanwhile in the $PSU(N)$ theory, there is no center symmetry since its center is trivial, but a magnetic $\bbZ_N^{(1)}$ symmetry. This relation can be generalized, see \cite{Brennan:2023mmt}.

More generally speaking, in a pure Yang-Mills theory with gauge Lie algebra $\mathfrak{su}(N)$ one can classify a loop operator by its charge $(n_E,n_M)\in \bbZ_N\times \bbZ_N$. Two such operators with $(n_E,n_M)$ and $(n_E',n_M')$ are mutually local (and can coexist) if and only if they satisfy the Dirac quantization condition
\be
n_E n_M'-n_M n_E'=0\quad (\text{mod}\ N)\,.
\ee
Such a choice of a set of mutually local operators is called a \textbf{polarization}. For example, if one choose the set of mutually local operators to be all the Wilson loops, then the global form of the gauge theory is an $G=SU(N)$ gauge theory. In constrast, if one choose them to be all the 't Hooft loops, the global form of the gauge theory is an $G=PSU(N)$ gauge theory. Of course, when 
$N$
is not a prime number, we can have cases in between.

After one choose a polarization, the theory is called an \textbf{absolute theory} with a well-defined partition function. If not, the theory is in general a \textbf{relative theory} with only a partition vector. 
The notion of partition vector is first introduced in the context of 6d (2,0) theories and their torus compactifications, where there exists mutually non-local dyonic charged operators, see for example~\cite{Witten:2009at,Freed:2012bs,Tachikawa:2013hya} for more details.

\subsection{Gauging Higher-Form Symmetry}
\label{sec:gauging}

An almost ancient yet refreshing question is: What is gauging? It usually means the procedure of changing a certain global symmetry into a gauge symmetry. But what is a gauge symmetry? In some tales it refers to a symmetry dependent upon the spacetime, but there are certain limitations, for example, how can one make a discrete symmetry dependent upon spacetime and different from global symmetry (if one does not abandon smoothness)? A better understanding is to sum over all possible ``configurations'' of gauge fields so that the gauge transformations are assuredly redundant, and we will be more explicit about these formulations in Section~\ref{sec:three}. In this section, unless specified, the $p$-form symmetry group $G^{(p)}$ is Abelian. We denote the background $(p+1)$-form gauge field of such symmetry to be $A^{(p+1)}$.

To put more precisely, there are two equivalent perspectives of gauging:
\begin{enumerate}
    \item summing over all possible insertions of the operators $U_g(M^{d-p-1})$;
    \item lift $A^{(p+1)}$ to be a dynamical gauge field, then integrate (sum over) all configurations of $A^{(p+1)}$.
\end{enumerate}
In particular, in the cases of gauging a finite abelian group, we have the following proposition:
\begin{prop}
\label{prop:gauge}
    After gauging the finite abelian symmetry $G^{(p)}$ of the theory $Z[A^{(p+1)}]$, the gauged theory $Z'$ would manifest a new symmetry $\hat{G}^{(d-p-2)}$, with $\hat{G}\equiv {\rm Hom}_{\rm Grp}(G,U(1))$ being the Pontryagin dual of $G$. The gauged theory has partition function
    \be Z'[\hat{A}^{(d-p-1)}] = \sum_{[A^{(p+1)}] \in H^{(p+1)}(M^d; G^{(p)})} Z(A^{(p+1)}) \exp(2\pi i \int_{M^d} \expval{\hat{A}^{(d-p-1)},A^{(p+1)}}  ) \ , \ee
    where $\expval{}$ is the inner product at group level. When the group $G^{(p)}$ is continuous, $\expval{}$ corresponds to the wedge product with action of $\Hat{\mathfrak{g}}$-valued element on $\mathfrak{g}$-valued element. If $G^{(p)}$ is discrete, $\expval{}$ should be interpreted as the cup product with action of $\Hat{G}$-valued element on $G$-valued element..
\end{prop}

\begin{proof}
    Before beginning the proof, let's observe the intuition of the theorem. In a naive sense of a theory with 0-form symmetry, gauging should amount to projecting out the Hilbert space into a physical one, but not quite. The projection is non-invertible, and in order to preserve the information of representation, we can insert a Wilson line (remind that a Wilson line is labelled by a chosen representation).

    So here we do the same. Coupling the term
    \[\exp(2\pi i \int_{M^d} \expval{\hat{A}^{(d-p-1)},A^{(p+1)}}  ) \ , \]
    can be seen as coupling a discrete flux. Following our previous method, the symmetry operators in theory $Z$ should be built by the discrete flux (which is the higer dimensional version of Wilson line)\cite{Schafer-Nameki:2023jdn}
    \be D(M^{d-p-1} ) = \exp(i\int_{M^{d-p-1}} \hat{A}^{(d-p-1)}) \ . \ee
    But something is not right, since $\hat{A}^{(d-p-1)}$ takes value in $\Hat{G}$ but not $G$! The correct understanding here is \textbf{the fusion rule is still the multiplication of $G$}, this is realized by writing the operator correctly as
    \be D(\rho, M^{d-p-1}) = \rho\Big( \exp(i\int_{M^{d-p-1}} \hat{A}^{(d-p-1)})\Big) \ , \ \rho\in {\rm Rep}(\Hat{G}) = G \ .  \ee
    The fusion rule for the symmetry operator is thus
    \be D(\rho_1 ,M ) \otimes D(\rho_2,M) = D(\rho_1\otimes \rho_2 ,M)   \ee
    correctly takes the $G$ multiplication as fusion rule.\footnote{We thank Ruizhi Liu for helpful comments on this part.}

    Similarly, for the gauged theory $Z'$, the symmetry operators are constructed as
    \be D'(\rho', M^{p+1}) = \rho'\Big( \exp(i\int_{M^{p+1}} A^{(p+1)}) \Big)\ , \ \rho'\in {\rm Rep}(G) = \Hat{G} \ . \ee
    With the symmetry operators constructed, we can legitimately say that the $Z'$ theory has a $\Hat{G}^{(d-p-2)}$ symmetry.
\end{proof}

As a corollary, we can actually do another gauging for $\hat{G}^{(d-p-2)}$ and obtain the original theory. This is just a simple exercise, done by observing that ${\rm Hom}( {\rm Hom}(G,\U(1)) , U(1) ) = G$ and $d-(d-p-2)-2 = p$.


\subsubsection{'t Hooft Anomaly for $p$-Form Symmetry}
However, in certain cases the gauging cannot be performed perfectly, there could be obstructions which we call 't Hooft anomaly. It is a phenomenon of gauge transformation rendering change in the theory's partition function. Namely, if we take $\delta_\lambda A$ to be a general transformation of the background gauge field\footnote{For an infinitesimal gauge transformation, it is called a \textbf{local anomaly}. For a large gauge transformation, it is a \textbf{global anomaly}.},
\be Z(A + \delta_\lambda A) = \exp(2\pi i \int_{M^d} I_d)\cdot  Z(A) \  \ee
where $I_d$ is a $d$-form in terms of the gauge field. Of course, we say that the system has an 't Hooft anomaly only if such a phase factor $2\pi i \int_{M^d} I_d$ cannot be cancelled by a gauge-invariant local counter term in the $d$-dimensional action. 

A standard procedure to treat this is through the descent formula. The idea is to find a gauge invariant object to describe the anomaly. the descent formula is
\be
\ba
\dd I_d = \delta_\lambda I_{d+1} \ , \ \dd I_{d+1}  = I_{d+2} \cr
Z(A + \delta_\lambda A) = \exp(2\pi i \int_{M^{d+1}} \delta_\lambda I_{d+1})\cdot  Z(A)
\ea
\ee
where $\ptl M^{d+1} = M^d $. $I_{d+1}$ is the $(d+1)$-form 't Hooft anomaly polynomial, $I_{d+1} $ is a gauge invariant $(d+2)$-form.

This gives hint for the ``anomaly inflow'' mechanism, that a QFT in $d$-dimension can be coupled to a QFT in $(d+1)$-dimension. A general construction is a QFT on $M^d$ with 't Hooft anomaly defined as above can be coupled to a TQFT on $M^{d+1}$ with $\mathcal{L} = -2\pi I_{d+1}$. The entire theory is a well-defined,  anomaly free theory. Such notion is directly related to the concept of symmetry protected topological (SPT) phase~\cite{gu2009tensor,chen2012symmetry,chen2013symmetry,chen2011two,levin2012braiding,PhysRevB.89.035147,chenjie2014,chenjie2015,senthil2015symmetry} in condensed matter physics, where the symmetry is exactly the global symmetry which we are gauging. SPT phases are interacting generalizations of noninteracting topological insulators and superconductors~\cite{TI01,TI02,TI1,TI2,TI3,TI4,TI5}, which exhibit robust boundary states and topological properties protected by symmetry.

We give some examples of the gauging of higher-form symmetry and 't Hooft anomaly.

\subsubsection{Example: Maxwell Theory with Background}
Now we consider a four-dimensional $U(1)$ Maxwell theory with $U(1)_E$, $U(1)_M$ 1-form symmetries  coupled to background gauge fields $B_E$, $B_M$ , the coupled action becomes
\be 
\label{Maxwell-gauged}
S = \frac{1}{2e^2} \int (F-B_E)\wedge * (F-B_E) + \frac{1}{2\pi} \int (F-B_E)\wedge B_M\,,  
\ee
the gauge transformations are
\be  A\mapsto A + \lambda_E  \ , \ B_E \mapsto B_E + \dd \lambda_E \ , \ B_M \mapsto B_M + \dd \lambda_M \ee
thus $(F-B_E)$ is gauge invariant, also note that here we allow $\lambda_E$ and $\lambda_M$ to be non-flat. 

Under such gauge transformation,
\be
\delta S = -\frac{1}{2\pi} \int_{M^4} B_E\wedge \dd\lambda_M\,.
\ee
Hence the action is not gauge invariant and the theory has a mixed 't Hooft anomaly, which means that we cannot simultaneously gauge both $U(1)_E \ $and$ \ U(1)_M$. After uplifting to 5d, the 't Hooft polynomial is
\be I_5=-\frac{1}{(2\pi)^2} \dd B_E \wedge B_M \ . \ee 
Note that the gauged action (\ref{Maxwell-gauged}) is chosen such that there are a minimal number of non-gauge invariant terms.

If we take the combined system
\be S_{\rm 5d-4d} = \frac{1}{2e^2} \int_{M^4} (F-B_E)\wedge * (F-B_E) + \frac{1}{2\pi} \int_{M^4} (F-B_E)\wedge B_M+\frac{1}{2\pi}\int_{M^5}dB_E\wedge B_M\,, \ee
the whole system is gauge invariant.

\subsubsection{Example: 3d U(1) Chern-Simons Theory}
The basic setting for 3d U(1) Chern-Simons theory is
\be\ba S &= \frac{\kappa}{4\pi} \int_{M^3} A\wedge \dd A  \ \ \kappa\in\bbZ \cr  {\rm e. o. m. :} \ \ \dd A &= 0 \ea \ee
here $\kappa\in \bbZ$ is the quantized Chern-Simons level.

This system possesses $\bbZ_\kappa^{(1)}$ 1-form symmetry. In 3d, the 1-form symmetry has $3-1-1=1$ dimensional topological line operators (or defects), while the objects are also one-dimensional lines. Only one gauge invariant line operator is natural in this theory: the Wilson loop\footnote{We do not have 't Hooft line previously mentioned in 4d Maxwell theory, because the requirement for electromagnetic duality lies in the Lagrangian. In other words, the magnetic symmetry of the Maxwell theory is broken by the Chern-Simons term.}. In the following, we denote $g_n\equiv \exp(2\pi i n/\kappa)$. The line operators are
\be  U_{g_n} (\gamma)= \exp(i n \oint_\gamma A) \ , \ \gamma:[0,1]\to M, \gamma(0) = \gamma(1) \ .   \ee
Action on the lines is therefore
\be U_{g_m} (\gamma_1) U_{g_n} (\gamma_2) = \exp(\frac{2\pi i m n}{\kappa} \expval{\gamma_1, \gamma_2}) U_{g_n} (\gamma_2) \ .  \ee
Another perspective for this is the topological operator $ U_{g_n} (\gamma)$ shifts $A$ by a flat connection, $A\mapsto A + \frac{n}{\kappa} \epsilon$.

If we gauge the $\bbZ_\kappa^{(1)}$ symmetry, we would encounter the $3-1-2 = 0$-form $\hat{\bbZ}_\kappa = \bbZ_\kappa $ symmetry. The gauging procedure is performed by summing all possible insertions of $U_{g_n}$. After gauging, the theory has become
\be \ba
A' &= \kappa A\cr
S &= \frac{1}{4\pi \kappa} \int_{M^3} A' \wedge \dd A' \quad, \  \kappa\in\bbZ  \ .
\ea\ee
Now $A'$ is the properly quantized gauge field.

However the new theory is no longer gauge invariant under a large gauge transformation of $A'$. For the readers unfamiliar with CS theory, we sketch the reasonings behind the quantization of Chern-Simons level (also see e. g. Section 8.4 of \cite{TongGauge}). Consider the CS action
\be
S=\frac{\kappa}{4\pi}\int_{M^3} A\wedge \dd A
\ee
defined on $M^3=S^1\times S^2$. We denote the Euclidean time direction  $S^1$ to be the $x^0\sim x^0+2\pi R$ direction, and the $S^2$ to be the $x^1$, $x^2$ directions. We introduce a magnetic flux $\int_{S^2}F=2\pi$. Now we consider a large gauge transformation $A_\mu\rightarrow A_\mu+\ptl_\mu \lambda$, where $\lambda=x^0/R$ winds around the $S^1$. Hence under this gauge transformation, we have $A_0\rightarrow A_0+1/R$, and the gauge field $A_0$ has periodicity $1/R$ now. For the gauge transformation of the CS action, we consider a configuration of constant $A_0$, and we simplify
\be
\ba
S&=\frac{\kappa}{4\pi}\int_{S^1\times S^2} A\wedge \dd A\cr
&=\frac{\kappa}{4\pi}\int_{S^1\times S^2}A_0 F_{12}+A_1 F_{20}+A_2 F_{01}\cr
&=\frac{\kappa}{2\pi}\int_{S^1\times S^2}A_0 F_{12}\,.
\ea
\ee
After the large gauge transformation
\be
\ba
\delta S&=\frac{\kappa}{2\pi}\int_{S^1\times S^2}\frac{1}{R} F_{12}\cr
&=2\pi\kappa\,.
\ea
\ee
Hence in order to keep $e^{iS}$ invariant, one requires $\kappa\in\mb{Z}$.

Now we intend to calculate the 't Hooft anomaly polynomial after gauging the $\mb{Z}_\kappa$ 1-form symmetry. We require $M^3$ to be a spin manifold which is a boundary of a 4-dimensional spin manifold $M^4$. We can rewrite the action as
\be S = \frac{\kappa}{4\pi} \int_{M^3} A\wedge \dd A =  \frac{\kappa}{4\pi} \int_{M^4} F\wedge F = 2\pi \kappa \int_{M^4} \frac{c_1^2}{2}  \ee
where $c_1 = \frac{F}{2\pi}$ is the first Chern class of principal bundle with connection $A$. Next we couple the system to a $\bbZ_\kappa$-valued background gauge field $B\in H^2 (M^3;\bbZ_\kappa)$, changing the action into
\be S \mapsto S =  \frac{\kappa}{4\pi} \int_{M^4} (F -\frac{2\pi}{\kappa} B)\wedge (F -\frac{2\pi}{\kappa} B)  \ ,\ee
which is not gauge invariant if we allow non-flat gauge transform $A\mapsto A +\epsilon/\kappa$, for $B$ is $\bbZ_\kappa$-valued and unable to submit to small gauge transform. Thus the anomaly polynomial term is
\be I_4 = \frac{1}{2\kappa} \int_{M^4} B\cup B  \ .  \ee
If $I_4\notin \bbZ$, we would encounter 't Hooft anomaly. The previous request that $M^4$ is a spin manifold gives the fact that
\be \int_{M^4} B\cup B  \in 2\bbZ \Rightarrow I_4\in \bbZ/\kappa \ , \ee
the theory is anomalous for $\kappa > 1 $.

\subsection{Symmetry Trifecta}
\label{sec:three}

We have now roughly understood higher-form symmetry in the language of topological defects, and in some cases we mentioned it could be also seen as shifting a flat connection. With the previous examples in mind, we are now ready to march forward to a more thorough understanding of higher-form symmetry. The assertion is that higher-form symmetry can be comprehended from three interchangeable viewpoints: \textbf{topological defect network} (TDN), \textbf{flat connection} and \textbf{classifying space}.

Let's start by defining these perspectives within 0-form symmetry, and then we'll extend them to higher-form symmetry.

\begin{enumerate}
    \item \textbf{Topological defect network}. For those readers who are already familiar with topological operators, feel free to skip to the next section on flat connections. We will now delve into the topic with a focus on the lattice perspective.

    Consider an array of sites extending through a $d$-dimensional space, with each site associated with a Hilbert space $\CH_i$, such as $\mathbb C^2$ for a qubit. The system evolves over time until it encounters a topological defect that disrupts its continuity. A visual representation of this concept in a $1+1$-dimensional setting is provided in Fig.~\ref{Fig:lattice-TDN}.

\begin{figure}[htbp]
\centering
\begin{tikzpicture}
\draw[thick, ->] (-1,0)--(2.2,0) node[anchor=west]{$x$} ;
\draw[thick, ->] (0,-0.5)--(0,2) node[anchor = south]{$t$};
\draw[fill = cyan] (-0.4,0) circle (2pt)
    [shift = {(10pt,0)}] (-0.4,0) circle (2pt)
    [shift = {(10pt,0)}] (-0.4,0) circle (2pt)
    [shift = {(10pt,0)}] (-0.4,0) circle (2pt)
    [shift = {(10pt,0)}] (-0.4,0) circle (2pt)
    [shift = {(10pt,0)}] (-0.4,0) circle (2pt)
    [shift = {(10pt,0)}] (-0.4,0) circle (2pt);
\draw[color = cyan, thick] (-0.4,0)--(-0.4,0.95)
    [shift = {(10pt,0)}] (-0.4,0)--(-0.4,0.95)
    [shift = {(10pt,0)}] (-0.4,0)--(-0.4,0.95)
    [shift = {(10pt,0)}] (-0.4,0)--(-0.4,0.95)
    [shift = {(10pt,0)}] (-0.4,0)--(-0.4,0.95)
    [shift = {(10pt,0)}] (-0.4,0)--(-0.4,0.95)
    [shift = {(10pt,0)}] (-0.4,0)--(-0.4,0.95);
\draw[thick ,orange] (-1,1)--(2.2,1) node[anchor=west]{$U_g$} ;
\draw[color = blue, thick] (-0.4,1.1)--(-0.4,1.9)
    [shift = {(10pt,0)}] (-0.4,1.1)--(-0.4,1.9)
    [shift = {(10pt,0)}] (-0.4,1.1)--(-0.4,1.9)
    [shift = {(10pt,0)}] (-0.4,1.1)--(-0.4,1.9)
    [shift = {(10pt,0)}] (-0.4,1.1)--(-0.4,1.9)
    [shift = {(10pt,0)}] (-0.4,1.1)--(-0.4,1.9)
    [shift = {(10pt,0)}] (-0.4,1.1)--(-0.4,1.9);

\end{tikzpicture}
\caption{The 1-dimensional lattice with topological defect.}
\label{Fig:lattice-TDN}
\end{figure}

The topological defect operators must adhere to the group representation, which can be expressed as
\be U_{g} U_{h} = U_{gh}. \ee
This operation can be understood as the fusion of topological defects $U_g$ and $U_h$ to $U_{gh}$ as
\be
    \begin{tikzpicture}
    \draw[color = cyan, thick,dashed] (-1,0.05)--(-1,-1.2)
    [shift = {(10pt,0)}] (-1,0.05)--(-1,-1.2)
    [shift = {(10pt,0)}] (-1,0.05)--(-1,-1.2)
    [shift = {(10pt,0)}] (-1,0.05)--(-1,-1.2)
    [shift = {(10pt,0)}] (-1,0.05)--(-1,-1.2)
    [shift = {(10pt,0)}] (-1,0.05)--(-1,-1.2)
    [shift = {(10pt,0)}] (-1,0.05)--(-1,-1.2)
    [shift = {(10pt,0)}] (-1,0.05)--(-1,-1.2)
    [shift = {(10pt,0)}] (-1,0.05)--(-1,-1.2)
    [shift = {(10pt,0)}] (-1,0.05)--(-1,-1.2)
    [shift = {(10pt,0)}] (-1,0.05)--(-1,-1.2);

    \draw[orange, thick] (-1,0)--(1,0) node[anchor = south]{$U_g$};
    \draw[orange, thick] (-1,-1)--(1,-1) node[anchor = north]{$U_h$};
    \draw[orange, thick] (1,0)--(1.5,-0.5);
    \draw[orange, thick] (1.5,-0.5)--(2.5,-0.5) node[anchor = west]{$U_{gh}$};
    \draw[orange, thick] (1,-1)--(1.5,-0.5);

    \end{tikzpicture}
\ee
blue dashed lines represent the time evolution of states at the sites. It's important to note that the group action by defects should have direction; the inverse direction corresponds to the inverse operator $U_{g^{-1}}$. For instance, we have the following consecutive action of defect operators: $U_g$, $U_{g^{-1}}$, and $U_g$. It should be the same as a single action of $U_g$:
\be
    \begin{tikzpicture}
    \draw[orange , thick , ->] (-1,1)--(-0.5,1) ; 
    \draw[orange , thick,->] (0.5,1)--(0.7,1) node[anchor = west]{$U_g$}; 
    \draw[orange , thick,->] (2,1)--(3,1) node[anchor = west]{$U_g$}; 
    \draw[orange,thick,->] (-0.5,1) arc (90:-90:0.25);
    \draw[orange,thick,->] (-0.5,0.5) arc (90:270:0.25);
    \draw[orange,thick,->] (-0.5,0) arc (-90:0:0.5);
    \draw[orange,thick] (0,0.5) arc (180:90:0.5);
    \draw[cyan,thick,->] (-0.5,-0.5) -- (-0.5,1.5);
    \draw[cyan,thick,->] (2.5,-0.5) -- (2.5,1.5);

    \node[] at (1.7,1) {$=$};
    \end{tikzpicture}
\ee

In addition, the inverse relation of $U_{g^{-1}}$ and $U_g$ can be represented as:
\be\label{Prop:bubble}
    \begin{tikzpicture}
    \draw[thick, orange ,->] (-0.5,0) arc (0:360:0.5);
    \node[] at (0,0) {$=$};
    \node[] at (0.5,0) {id};
    \end{tikzpicture}
\ee
Such a bubble can always be added to the TDN without actually changing the system which respects the symmetry.

    \item \label{para:flatconn}\textbf{Flat connection.} In the context of lattice gauge theory, or equivalently in \Cech cohomology, we assign labels to lattice sites using indices like $i$, $j$, $k$, and so on. The 1-form gauge connection, denoted as $A$, is responsible for associating group elements with the links connecting lattice sites. In other words, if there is a link between site $i$ and site $j$, then $A_{ij} \in G$, the group we are working with. However, this assignment is not arbitrary; it must adhere to specific conditions such as 
    \be A_{ij}\cdot A_{jk} \cdot A_{ki} = {\rm id} = A_{ij} \cdot A_{ji} \ . \label{Eq:Flatcondition}\ee
    This condition characterizes what is known as a flat connection. In fact, Eq.~(\ref{Eq:Flatcondition}) can be interpreted as a line integral
    \be  \oint A = 0. \ee
    in the continuous limit.
    The gauge transformation in this case is given by
    \be\label{gaugetransf1} A_{ij} \mapsto A_{ij}'=g_i^{-1} A_{ij} g_j,  \ee
    and thus,
    \be \int_{i\to j} A = A_{ii_1}\cdot A_{i_1 i_2}\dots A_{i_n j} \mapsto g_i^{-1} \cdot (\int_{i\to j} A) \cdot g_j. \ee
    If we choose a closed loop for the integral, we observe that the Wilson loop is gauge invariant, as expected.

    \item \textbf{Classifying space}. This is a more abstract concept, and in this section, we aim to provide insights rather than rigorous proofs. The main idea is to construct a space $BG$ in such a way that any principal bundle $P\to M$ can be generated by taking a pullback through a map $\phi : M\to BG$, as illustrated in Fig.~\ref{Fig:ClassSpace}. In this context, $EG\to BG$ represents the "universal $G$-bundle," which can be pulled back to yield any $G$-bundle. While this may seem complex, if we harken back to our earlier understanding of flat connections, we can discern that what truly matters is the information concerning the assignment of group elements to links connecting sites. The specific details of the sites themselves become less relevant for our purposes. For a more detailed description, please refer to Example \ref{Eg:GT}. Specifying a gauge field configuration is essentially equivalent to defining the holonomy of paths, which, in turn, corresponds to a mapping from spacetime to the classifying space $BG$.
    \begin{figure}[htbp]
        \centering
            \begin{tikzcd}
    G \arrow[d, hook] \arrow[rd, hook] &                                                     \\
    P \arrow[d]\arrow[r]                        & EG \arrow[d] 
    \\
    M \arrow[r, "\phi"]                & BG                                                 
    \end{tikzcd}
    \caption{Classifying space and universal bundle.}
    \label{Fig:ClassSpace}
    \end{figure}

\begin{definition}
The classifying space of a group $G$, denoted as $BG$, is a topological space that possesses the property $\pi_1(BG) \cong G$, while $\pi_n(BG) = 0$ for all $n$ not equal to 1.
\end{definition}

A map $\phi: M\to BG$ assigns points in the spacetime manifold $M$ to the single point in the classifying space $BG$ and paths in $M$ to group elements, as we've discussed previously. It's essential to recognize that homotopy-equivalent $\phi$ functions lead to equivalent $G$-bundles on the spacetime manifold $M$. From the perspective of gauge field connections, this equivalence corresponds to gauge-equivalent gauge field configurations $A$ and $A'$ in Eq.~(\ref{gaugetransf1}). Thus, the concept of homotopy in mathematics is translated into the notion of gauge transformations in physics.



\end{enumerate}

Extending the previous three perspectives from 0-form symmetry to higher-form symmetry is a straightforward process. Here's the overview:

\begin{enumerate}
    \item \textbf{Topological defect network} (Geometric intuition). In this viewpoint, topological defects are $(d-p-1)$-dimensional entities embedded within a $d$-dimensional spacetime. These entities have the capability to combine or interact with each other along lower-dimensional manifolds. A defect network serves as a distinctive and precise guide for understanding how a higher $p$-form symmetry operates within a given physical system. It outlines the rules and patterns governing the action of this symmetry on the system.  This approach provides us with a geometric intuition, as all operators are associated with submanifolds within the spacetime manifold.

    \item \textbf{Flat connection} (Physical computation). In this context, the connections take the form of $(p+1)$-forms that establish an assignment between group elements and $(p+1)$-simplices within the spacetime manifold. For instance, a 1-form symmetry corresponds to 2-form gauge fields. In a lattice framework, these gauge fields associate group elements, denoted as $g_{ijk}$ for an Abelian group $G$, with the triangles, represented as $\langle ijk\rangle$, in the triangulation of the spacetime manifold. These connections are also subject to the flat connection condition, which implies that the sum of group elements associated with all four triangles on the boundary of a tetrahedron adds up to the trivial element in $G$. This approach has the advantage that we can easily write down the action and gauge transformation rules in terms of gauge fields, making it well-suited for physical computations.
    
    \item \textbf{Classifying space} (Topological classification). We now have the classifying map $\phi: M^d\to B^{p+1}G$, where $B^{p+1}G=K(G,p+1)$ is the Eilenberg-MacLane space defined such that $\pi_{p+1}(B^{p+1}G) \cong G$, and all other homotopy groups are trivial. A map from spacetime manifold $M^d$ to the Eilenberg-MacLane space is equivalent to assigning group elements in $G$ to the $(p+1)$-dimensional simplices of the spacetime manifold triangulation. These assignments automatically satisfy the flat connection conditions. This approach translates the physical problem into a homotopy problem, making it useful for classifications, such as understanding anomalies in higher-form symmetry.
    
    \end{enumerate}

\section{Anomalies and SPT}

In this section, we will review 't Hooft anomalies within the context of our recent comprehensive perspectives on symmetries and demonstrate their relationship with SPT phases~\cite{chen2011two,chen2012chiral,elsenayak2014,kapustin2014anomalous,kapustin2014anomalies,kapustin2014symmetry,thorngren2015higher,wang2015bosonic,chenjie2016,bultinck2018global,wan2018higher,wang2018symmetric,TI4,TI5}. In section~\ref{sec:anomaly-SPT}, our discussion will center around the relationship between 't Hooft anomaly and SPT phases in one dimension higher. For the sake of simplicity, our primary focus will be on anomalies in 1D and 2D. We will illustrate this with a well-known example in the condensed matter community: the anomalous spin-1/2 at the boundary of the celebrated Haldane chain. Section~\ref{sec:gen-SPT} is designed for advanced readers and offers an in-depth exploration of the general classifications of 't Hooft anomalies and SPT phases. This section delves into complex topics such as higher-dimensional systems, higher-form symmetries, fermionic systems, and cobordism. While this section provides valuable insights, beginners may choose to skip this section and return to it later, as it addresses more complicated concepts beyond the introductory level.

\subsection{'t Hooft Anomaly on the Boundary of SPT}
\label{sec:anomaly-SPT}

In the context of a theory with a matter field $\phi$ described by the action $S[\phi]$, it is essential to define what we mean by a \emph{global} symmetry, which can be discrete and not necessarily infinitesimally generated. We begin by specifying the symmetry transformation of the matter field, denoted as $\phi\mapsto \phi^g$ for each $g\in G$, where $G$ represents the symmetry group. The requirement for the quantum system of matter field $\phi$ to possess a global symmetry $G$ is that under such transformations, the following conditions are met:
\begin{align}\label{gaugetransf}
S[\phi] = S[\phi^g] \quad \text{and} \quad \mathcal{D}\phi = \mathcal{D}\phi^g, \quad \forall g\in G.
\end{align}
In this case, the theory defined by the partition function $Z = \int \mathcal{D}\phi \exp(iS[\phi])$ is said to exhibit a global symmetry characterized by the group $G$.

Now, let's couple a background gauge field, denoted as $A$, into the theory and explore its implications. In this context, a background gauge field means that we perform the integration only over the matter field $\phi$, while keeping the gauge field $A$ fixed:
\be Z[A] = \int\mathcal{D}\phi e^{iS[A,\phi]} \ . \ee
A 't Hooft anomaly arises when $Z[A]$ is not equal to $Z[A^g]$, indicating that $Z[A]$ does not remain invariant under \emph{local} gauge transformations.
Indeed, because $Z[A]$ lacks gauge invariance, integrating out the gauge field $A$ becomes problematic. This issue arises because we are effectively integrating over the gauge orbit of the gauge field $A$. To account for this, we introduce the concept of an obstruction to gauging $A$ within $Z[A]$. This obstruction is described by the equation:
\be
Z[A^g] = Z[A]\cdot Z_{\text{anomaly}} [A,g],
\ee
where $g$ represents a \emph{local} gauge transformation, rather than a global one in Eq.~(\ref{gaugetransf}). The phase factor $Z_{\text{anomaly}}[A,g] \in U(1)$ characterizes the 't Hooft anomaly of the theory.

From the three perspectives on symmetries in a physical system discussed in Sec.~\ref{sec:three}, 't Hooft anomalies reveal the following implications:
\begin{enumerate}
\item \textbf{TDN perspective}. In this view, the $G$-symmetry defects lack the topological properties. Notably, when we alter the network's topology through local operations like the reconnection of topological defects in the spacetime manifold, additional phase factors may emerge. These additional phase factors are indicative of the symmetry's anomaly.
\item \textbf{Flat connection perspective}. Here, the local transformations of the TDN can be translated into gauge field gauge transformations. Consequently, the non-invariance of the partition function under local TDN changes translates into distinct partition functions for equivalent gauge field connections. In other words, $Z[A^g]\neq Z[A]$.
\item \textbf{Classifying space perspective}. In this perspective, gauge transformations are related to the homotopy of the map $\phi: M^d\rightarrow B^{p+1}G$. Therefore, the partition function $Z: [M^d, B^{p+1}G]\rightarrow \mathbb{C}$ with anomaly is not a well-defined function on $[\phi]$, which represents the homotopy classes of maps from $M^d$ to $B^{p+1}G$. Instead, it depends on the specific choice of the map $\phi$ in the homotopy class $[\phi]$.
\end{enumerate}

\subsubsection{0+1D 't Hooft Anomaly and 1+1D SPT}
\label{subsubsec:0+1D}

Consider a 0+1D system featuring a 0-form global symmetry denoted as $G$. In this one-dimensional spacetime, the effect of an element $g\in G$ is illustrated by the insertion of a point operator $U_g$, which acts on the system's Hilbert space. Notably, the precise location of this insertion of $U_g$ holds no significance due to the global symmetry's nature. Furthermore, the sequential action of $U_h$ and $U_g$ is entirely equivalent to the action of a single operator $U_{gh}$. This equivalence can be conceptually understood as a local move, altering the topology of the topology of TDN. Mathematically, this relationship is expressed as
\be
U_gU_h=U_{gh}.
\ee
This equation signifies that the action of $U_g$ essentially forms a linear representation of the group $G$ within the Hilbert space of the 0+1D system.

On the contrary, if we relax the assumption that the TDN remains invariant under local moves, we introduce a modified equation
\be\label{UgUhUgh}
U_gU_h = \omega(g,h) U_{gh}.
\ee
Here, $\omega: G\times G \rightarrow U(1)$ introduces an additional phase factor. It's essential to note that this phase factor, denoted as $\omega$, is not arbitrary. The $G$-action on the Hilbert space must adhere to associativity, giving rise to the requirement
\be\ba
&\quad (U_g\cdot U_h)\cdot U_k = U_g\cdot (U_h\cdot U_k)\cr
&\Rightarrow \omega(g,h)\omega(gh,k) U_{ghk} = \omega(g,hk)\omega(h,k) U_{ghk}\cr
&\Rightarrow\omega(g,h)\omega(gh,k)  = \omega(g,hk)\omega(h,k). \cr
\ea\ee
This consistency condition can be interpreted as a 2-cocycle condition within group cohomology, precisely:
\be (\dd_2 \omega) (g,h,k) := \frac{\omega(h,k)\omega(g,hk)}{\omega(gh,k)\omega(g,h)} = 1. \ee
For a more comprehensive exploration of group cohomology, please refer to Appendix~\ref{App:GrpChmlg}.

The definition of the symmetry action $U_g$ acting on the Hilbert space of the system can be modified through the introduction of a phase factor $\lambda(g)$:
\be
U_g \mapsto U_g' = \lambda(g) U_g.
\ee
This modification has the effect of transforming the 2-cocycle $\omega$ into a new function
\be
\omega(g,h)\mapsto \omega'(g,h )= \omega(g,h)\lambda(h)   \lambda(g)\lambda(gh)^{-1} =(\omega\cdot \dd_1\lambda)(g,h) \ . 
\ee
Notice that the elements in ${\text{im}(\dd_1)}$ can be interpreted as a global $U(1)$ phase, which is a choice we have when describing the quantum system but has no physical significance. Therefore, these should be quotiented out in our descriptions. Therefore,the phase factor $\omega$ are elements of the second cohomology group of $G$, defined as:
\be H^2_{\rm grp}(G,U(1))  =  \frac{{\rm ker}(\dd_2)} {{\rm im}(\dd_1)} \ . \ee
Specifying a phase factor $\omega \in H^2_{\text{grp}}(G, U(1))$ is tantamount to defining a projective representation of the group $G$.

The analysis we have discussed reveals an intriguing connection between 't Hooft anomalies and the cohomology group $H^2_{\rm grp}(G,U(1))$ in the context of 0+1D quantum field theories. This cohomology group is instrumental in classifying the projective representations of the group $G$. Interestingly, the cohomology group $H^2_{\rm grp}(G,U(1))$ also has another crucial role in physics: it serves as a classification scheme for 1+1D SPT states~\cite{SPT1D1,SPT1D2,SPT1D3,SPT1D4}. SPT phases are gapped states of matter that feature non-trivial protection due to symmetries. Remarkably, a direct correspondence exists between 0+1D quantum field theories exhibiting 't Hooft anomalies characterized by $\omega \in H^2_{\rm grp}(G,U(1))$ and 1+1D $G$-SPT phases that also bear the same $\omega$. This correspondence suggests that a 0+1D quantum field theory with a 't Hooft anomaly can be viewed as the boundary theory of a 1+1D SPT phase. Consequently, this means that a non-trivial SPT phase cannot possess a boundary that is both gapped and symmetric without topological orders, as a boundary should inherently exhibit a non-trivial 't Hooft anomaly.


\begin{center}
\begin{tikzpicture}
    \draw[ultra thick, cyan] (0,0) ellipse (0.5 and 1) (-0.5,0)node[left]{0+1D QFT with $G^{(0)}$-anomaly} ;
    \draw (0,1) .. controls (4,1.2) and (4,-1.2) .. (0,-1) (2,1)node[above]{1+1D SPT with $G^{(0)}$ symmetry};
\end{tikzpicture}
\end{center}

To gain insights into the relationship between 't Hooft anomalies in 0+1D and SPT phases in 1+1D, one can commence by constructing the partition function of a 1+1D SPT system on a lattice explicitly. The step-by-step procedure for defining this partition function is outlined as follows:
\begin{enumerate}
    \item Choose a triangulation for the 1+1D closed spacetime manifold $M$.
    \item Label each edge in the triangulation with a group element from $G$. This should satisfy the flat connection conditions within each triangle, represented as:
    \be\label{eq:tri}
        \begin{tikzpicture}
            \draw[thick,->] (0,0)--(0,1);
            \draw[thick] (0,1)--(0,2) (0,1)node[left]{$gh$};
            \draw[thick,->] (0,0)--(0.5,0.5)  node[right]{$g$};
            \draw[thick] (0.5,0.5)--(1,1);
            \draw[thick,->] (1,1)--(0.5,1.5) node[right]{$h$};
            \draw[thick] (0.5,1.5)--(0,2);
        \end{tikzpicture}
    \ee
    This edge labeling essentially offers a discrete representation of a flat connection $A: M \to BG$, where each distinct path in $M$ corresponds to a group element, obtained by multiplying the respective group elements along the links.
    \item Each basic triangle (or simplex in general dimensions) contributes an element $\omega_2(g, h)^{\pm 1} \in U(1)$ to the partition function. The sign $\pm$ varies with the orientation of the triangle. This function $\omega_2$ can be interpreted as a map $ BG \to B^2 U(1) $, and falls within the second cohomology group $H^2(BG, U(1))$, given the isomorphisms 
    \be [BG, B^2U(1)]\cong H^2(BG,U(1))\cong H_\mathrm{grp}^2(G,U(1)).  \ee
    \item The partition function for any flat connection configuration $A$ is expressed as the product of $\omega_2$ over all triangles, accounting for their orientations~\cite{chen2013symmetry}:
    \be Z[A] =  \exp{i \int_M L} = \prod_{\Delta} e^{iL(\Delta)} = \prod_{\Delta} \omega_2 (A(\Delta)) = \prod_{\Delta}A^* (\omega_2)(\Delta).
    \ee
\end{enumerate}
This construction serves as a universal way for defining the partition function of a 1+1D SPT system. It employs edge labeling within a triangulation and incorporates contributions from individual triangles, thereby establishing a compelling correspondence between the partition function and the cohomology class $ \omega_2 \in H^2_\mathrm{grp}(G,U(1))$.

The intricate relationship between 0+1D 't Hooft anomalies and 1+1D SPT phases becomes particularly transparent when we consider an open 1+1D spacetime manifold. Specifically, let's focus on a triangle near the boundary, as exemplified in Eq.~(ref{eq:tri}). The removal of such a triangle can be understood as a local boundary operation, altering the number of boundary edges by one unit. In this context, each group element label on an edge corresponds to the insertion of the operator $U_g$ in the dual TDN representation. When a triangle near the boundary is removed, the symmetry action at the boundary transitions from $U_{gh}$ to $U_gU_h$. This change induces a discrepancy in the phase factor, denoted as $\omega(g,h)$ in Eq.~(\ref{UgUhUgh}), which serves as a signature of the 't Hooft anomaly. Remarkably, this phase factor precisely counterbalances the contribution that the removed triangle made to the 1+1D bulk SPT partition function. This unveils a generic anomaly inflow mechanism: the combined system, which includes both the boundary and the bulk, remains physical consistency and gauge invariance. This insight confirms that a non-trivial SPT state cannot have a trivially gapped symmetric boundary, as the boundary theory manifests a non-trivial 't Hooft anomaly.

\subsubsection{Example: Anomalous spin-1/2 on the Boundary of a Haldane Chain}

To delve into the intricate phenomenon of bulk-boundary correspondence, let us examine a seminal example: the relationship between the 1+1D Haldane chain~\cite{haldane1983nonlinear,haldane1983continuum} and its anomalous 0+1D spin-$\frac{1}{2}$ boundary. This example illustrates how 't Hooft anomalies are related to the SPT phase in one higher dimension.

The Haldane chain is a 1+1D spin-$s$ chain with a Heisenberg interaction given by
\be\label{eq:Haldane}  H = \sum_{\expval{ij}} \vb{S}_i\cdot \vb{S}_j. \ee
In the low-energy limit, Haldane demonstrated that this model can be effectively described by a continuous quantum field theory characterized by the partition function 
\be\ba 
Z &= \int \mc{D} \vb{n} \exp(i\int_{M^2} (\ptl_\mu \vb{n} )^2 \dd^2 x + i  \theta \mathcal{N}[\vb{n}]  ),
\ea\ee
where $\mathcal{N}$ represents the winding number or mapping degree of the map of the field $\mathbf{n}: M^2 \rightarrow S^2$. The parameter $\theta=2\pi s$ in the second term is related to the spin $s$ of the chain.

For integer spins $s$, the second term $e^{i\theta \mathcal N}$ simplifies to $e^{2\pi i} = 1$ when integrated over closed manifolds $M^2$ as $\mathcal{N}$ is an integer. Consequently, with only the first non-linear sigma model, the theory will flow to a gapped symmetric disorder state under renormalization group flow. In contrast, for half-odd-integer spins $s$, this term can take values of $\pm 1$ depending on whether $\mathcal{N}$ is even or odd, leading to destructive interference in the path integral. This is believed to leads to a gapless state. The gapped/gapless dependence of $s$ for the antiferromagnetic 1+1D Heisenberg model is known as the Haldane conjecture~\cite{haldane1983nonlinear,haldane1983continuum}.

The non-trivial edge state of the spin-1 Haldane chain can be understood through the partition function of an open 1+1D chain. Along the 0+1D boundary, the $\theta$ term becomes a Wess-Zumino-Witten model of a spin-1/2 system~\cite{ng1994edge}. The key distinction between the spin-1 degrees of freedom (d.o.f.) in the 1+1D bulk and the spin-1/2 d.o.f. on the boundary lies in their symmetry groups. The spin-1 d.o.f. forms a linear representation of the spin rotational symmetry group $SO(3)$, while the spin-1/2 d.o.f. constitutes a linear representation of $SU(2)$ and a projective representation of $SO(3)$. From this perspective, we recognize that the Haldane chain is an SPT state protected by the $SO(3)$ spin rotational symmetry~\cite{gu2009tensor}. Its boundary corresponds to a spin-1/2 system, which exhibits fractionalized charges and cannot be gapped out symmetrically. Consequently, the $SO(3)$ 't Hooft anomaly of 0+1D spin-1/2 is canceled by the 1+1D SPT bulk protected by the same $SO(3)$ group. This intriguing relationship between the 't Hooft anomaly at the boundary and the SPT phase in the bulk elucidates the role of symmetry and topology in the unique properties of the Haldane chain.

A more explicit way to directly observe that the boundary of the spin-1 Haldane chain is a spin-1/2 with a projective representation of $SO(3)$ is through the AKLT model~\cite{affleck2004rigorous}. This model can be visualized as follows:
\begin{center}
    \begin{tikzpicture}
        \filldraw[cyan] (0,0) circle (0.15);
        \filldraw[shift = {(1,0)},cyan] (0,0) circle (0.15);
        \filldraw[shift = {(2,0)},cyan] (0,0) circle (0.15);
        \filldraw[shift = {(3,0)},cyan] (0,0) circle (0.15);
        \filldraw[shift = {(4,0)},cyan] (0,0) circle (0.15);
        \filldraw[shift = {(5,0)},cyan] (0,0) circle (0.15);
        \filldraw[shift = {(6,0)},cyan] (0,0) circle (0.15);
        \filldraw[shift = {(7,0)},cyan] (0,0) circle (0.15);
        \draw[thick] (0.5,0) ellipse (0.8 and 0.4 );
        \draw[shift = {(2,0)}, thick] (0.5,0) ellipse (0.8 and 0.4 );
        \draw[shift = {(4,0)}, thick] (0.5,0) ellipse (0.8 and 0.4 );
        \draw[shift = {(6,0)}, thick] (0.5,0) ellipse (0.8 and 0.4 );
        \draw[thick,blue] (1,0)--(2,0);
        \draw[shift = {(2,0)},thick,blue] (1,0)--(2,0);
        \draw[shift = {(4,0)},thick,blue] (1,0)--(2,0);
        
    \end{tikzpicture}
\end{center}
Here, the blue dots represent spin-1/2 particles. The ellipses enclosing two spin-1/2 particles represent the projection onto the spin-1 triplet states on each site. A blue line connecting two spin-1/2 particles signifies a singlet state between them. It can be shown that the Hamiltonian~\cite{affleck2004rigorous}
\be
H_\mathrm{AKLT} = \sum_{\expval{ij}} \left[ \vb{S}_i\cdot \vb{S}_j + \frac{1}{3} (\vb{S}_i\cdot \vb{S}_j)^2 \right]
\ee
has exactly the ground state depicted in the above picture. Numerical calculations confirm that the AKLT model and the antiferromagnetic Heisenberg model in Eq.~(\ref{eq:Haldane}) belong to the same phase without phase transitions. On a closed 1D spatial ring, the ground state of the AKLT model is unique. On an open chain, however, the ground state is four-fold degenerate due to the presence of dangling spin-1/2 particles at the two ends of the chain. Therefore, we can explicitly see that the boundary of the spin-1 Haldane chain is a spin-1/2 system with a $SO(3)$ symmetry anomaly manifested as a projective representation rather than a linear representation.

From the perspective of SPT phases, the 1+1D Haldane phase is protected by the symmetry group $SO(3)$ or its subgroup $\mathbb{Z}_2 \times \mathbb{Z}_2$. Both the 1+1D SPT phase and the 0+1D 't Hooft anomaly can be classified using the projective representation of the symmetry group. Specifically, we can examine the second group cohomology of $SO(3)$ and $\mathbb{Z}_2 \times \mathbb{Z}_2$ with $U(1)$ coefficient:
\be  H_\mathrm{grp}^2 (SO(3),U(1)) = H^2_\mathrm{grp} (\mathbb Z_2\times\mathbb Z_2,U(1)) = \bbZ_2. \ee
The non-trivial element in $H_{\text{grp}}^2 (SO(3),U(1)) = H^2(BSO(3),U(1)) $ is the second Stiefel-Whitney class $\frac{1}{2}w_2$. Therefore, the action of the 1+1D $SO(3)$-SPT is given by
\be
S=\frac{1}{2}\int \dd^2 x\ w_2.
\ee
Additionally, the presence of a boundary spin-$\frac{1}{2}$ also corresponds precisely to the non-trivial projective representation of $SO(3)$ or $\mathbb{Z}_2 \times \mathbb{Z}_2$. Its 't Hooft anomaly is also captured by the second Stiefel-Whitney class.

This example effectively illustrates the concept of bulk-boundary correspondence where a quantum field theory residing on the boundary, exhibiting a 't Hooft anomaly, corresponds to the boundary of the corresponding SPT state.

\subsubsection{Example: 1+1D Anomaly and 2+1D SPT}

Now, as we lift the spacetime dimension by one, the local move for the $G$ connection in 1+1D is associated with the 2D Pachner move. This move introduces a phase factor $\omega_3$ that depends on three group elements:
\begin{center}
    \begin{tikzpicture}
        \draw[thick , ->] (0,0) -- (0.5,0) node[below] {$ghk$};
        \draw[shift={(4,0)}, thick , ->] (0,0) -- (0.5,0) node[below] {$ghk$};
        \draw[thick ] (0.5,0) -- (1,0);
        \draw[shift={(4,0)}, ,thick ] (0.5,0) -- (1,0);
        \draw[thick,->] (0,0)--(0,0.5) node[left]{$g$};
        \draw[shift={(4,0)},thick,->] (0,0)--(0,0.5) node[left]{$g$};
        \draw[thick ] (0,0.5) -- (0,1);
        \draw[shift={(4,0)},thick ] (0,0.5) -- (0,1);
        \draw[thick, ->] (0,1) -- (0.5,1) node[above] {$h$};
        \draw[shift={(4,0)},thick, ->] (0,1) -- (0.5,1) node[above] {$h$};
        \draw[thick] (0.5,1)--(1,1) ;
        \draw[shift={(4,0)},thick] (0.5,1)--(1,1) ;
        \draw[thick , ->] (1,1)--(1,0.5) node[right]{$k$} node[right]{$\  \ \ = \omega_3(g,h,k)$};
        \draw[shift={(4,0)},thick , ->] (1,1)--(1,0.5) node[right]{$k$};
        \draw[thick] (1,0.5)--(1,0);
        \draw[thick,shift={(4,0)}] (1,0.5)--(1,0);
        
        \draw[thick, ->, blue] (0,0)--(0.5,0.5);
        \draw[thick, blue] (0.5,0.5)--(1,1) node[above]{$gh$};
        \draw[shift={(4,0)},thick,-> ,blue] (0,1)--(0.5,0.5);
        \draw[thick, blue] (0.5,0.5)--(1,1) node[above]{$gh$};
        \draw[shift={(4,0)},thick,blue] (0.5,0.5)--(1,0) node[right]{$hk$} ;
    \end{tikzpicture}.
\end{center}
From the 2+1D SPT bulk perspective, the phase factor $\omega_3(g,h,k)$ is assigned to a tetrahedron of a flat connection $A: M^3 \to BG$ as
\begin{center}
    \begin{tikzpicture}
        \draw[->,thick] (0,0)--(0.5,-0.5) node[left]{$g$} ;
        \draw[thick] (0,0)--(1,-1)  ;
        \draw[->,thick] (1,-1)--(1,0.3) node[left]{$h$} ;
        \draw[thick] (1,0)--(1,1) ;
        \draw[->, thick] (1,1)--(1.5,0.5) node[right]{$k$};
        \draw[thick] (1.5,0.5)--(2,0) ;
        \draw[->, ultra thick,cyan,densely dotted] (1,-1)--(1.5,-0.5) node[right]{$hk$} ;
        \draw[ultra thick,cyan, densely dotted] (1.5,-0.5)--(2,0);
        \draw[ultra thick, orange,densely dotted] (0,0)--(2,0) node[right]{$ghk$};
        \draw[ultra thick, orange,densely dotted,->] (0,0)--(1,0);
        \draw[ultra thick, magenta,densely dotted,->] (0,0)--(0.5,0.5) node[left]{$gh \ $};
        \draw[ultra thick, magenta,densely dotted] (0,0)--(1,1);
        \node (a) at (5,0) {$ \mapsto  \ \ \omega_3(g,h,k)$.};
    \end{tikzpicture}
\end{center}
Here, $M^3$ represents a three-dimensional spacetime manifold, and $BG$ denotes the classifying space of the symmetry group $G$.
The partition function of the 2+1D SPT state can be represented as the product of phase factors associated with each tetrahedron in the spacetime. By accounting for the contributions from all these phase factors, we observe that the 't Hooft anomaly originally present in the 1+1D theory is cancelled by the influence of the 2+1D bulk SPT. This cancellation mechanism ensures the consistency of both symmetry and topological properties in the combined system.

There exists a consistency condition governing the phase factors $\omega_3$. Much like the 0+1D anomalies where $\omega_2$ needs to be a 2-cocycle, the phase factor $\omega_3$ must adhere to a similar requirement, which can be visually represented through the following commutative diagram:
\begin{center}
\begin{tikzcd}
&  & (g(hk))l \arrow[rr, "{\omega_3(g,hk,l)}"] & & g((hk)l) \arrow[rrd, "{\omega_3(h,k,l)}"] &  & \\
((gh)k)l \arrow[rrrd, "{\omega_3(gh,k,l)}"] \arrow[rru, "{\omega_3(g,h,k)}"] &  & & &   &  & g(h(kl)) \\
&  & & (gh)(kl) \arrow[rrru, "{\omega_3(g,h,kl)}"] & &  &         
\end{tikzcd}
\end{center}
The condition above is precisely the 3-cocycle condition for $\omega_3$:
\be
(\dd \omega_3)(g,h,k,l) = \frac{\omega_3(h,k,l) \omega_3(g,hk,l) \omega_3(g,h,k)}{\omega_3(gh,k,l) \omega_3(g,h,kl)} = 1.
\ee
The 3-cocycle condition, as a special case of the pentagon equation, is a fundamental requirement for maintaining the consistency and coherence of the phase factors $\omega_3$ across different tetrahedra. By satisfying this condition, the phase factors ensure a well-defined and topologically invariant SPT state. This condition is crucial for establishing the robustness and stability of the SPT phase, as it guarantees that the physical properties and symmetries of the system are preserved throughout different triangulations or partitions of the spacetime.

Similar to the 0+1D anomaly and 1+1D SPT, the 2+1D SPT and 1+1D anomaly also exhibit a ``gauge transformation'' for the cocycle $\omega_3$. This transformation arises from the redefinition of the symmetry action through local counterterms. These counterterms introduce a subgroup referred to as 3-coboundaries, which needs to be factored out in the classification of both the anomaly and SPT phases. Consequently, the classification is determined by the third cohomology group $H_{\mathrm{grp}}^3(G,U(1))$~\cite{chen2012symmetry,chen2013symmetry}.


\subsection{General Classifications of 't Hooft Anomalies and SPT}
\label{sec:gen-SPT}

Previously, we discussed the 't Hooft anomaly of a physical system with 0-form symmetry $G$. The 't Hooft anomaly is captured by a phase factor associated with the gauge transformation of gauge connections or the local move of the TDN of the group. We also demonstrated that this phase factor corresponds to the construction data of the $G$-SPT partition function in one higher dimension. In the following, we will generalize previous discussions to higher dimensions, higher-form symmetries, fermionic systems, and cobordism classifications. This content is more advanced, and beginners might find it beneficial to revisit this section later.

\subsubsection{Higher Dimensions}

This construction can be easily generalized to 't Hooft anomalies in $d$ spacetime dimensions. In this case, we begin by triangulating the $d$-dimensional spacetime manifold. A TDN of $G$ connections is then constructed by assigning a group element $g_{ij}$ to each link $\langle ij \rangle$ of the triangulation. Under local retriangulations, such as Pachner moves, the connections should be equivalent. However, the system undergoes a change by a $U(1)$ phase factor $\omega_{d+1}$. This phase factor characterizes the 't Hooft anomaly of the $G$ symmetry in the system.

Furthermore, the Pachner move of triangulation in $d$ dimensions can be glued to a $(d+1)$-dimensional simplex. This $(d+1)$-dimensional simplex can be understood as the building block for the $G$-SPT partition function in one higher dimension. As a result, a direct relationship emerges between the $d$-dimensional 't Hooft anomaly and the $(d+1)$-dimensional $G$-SPT phases in one higher dimension. The consistency equation, similar to associativity in 0+1D, is captured by the cocycle condition $\dd \omega_{d+1}=1$ for the phase factor $\omega_{d+1}$. Additionally, gauge transformations can be applied to these phase factors, which is equivalent to modding out all the $(d+1)$-coboundaries. 
Therefore, the 't Hooft anomaly in $d$ dimensions and the $G$-SPT phases in $d+1$ dimensions are both classified (in low dimensions) by the $(d+1)$-dimensional group cohomology~\cite{chen2012symmetry,chen2013symmetry}
\be\label{eq:HSPT}
H_{\mathrm{grp}}^{d+1}(G,U(1)) \cong H^{d+1}(BG,U(1)) \cong [BG,B^{d+1}U(1)].
\ee
Here, $H_{\mathrm{grp}}^{d+1}$ represents the group cohomology, $H^{d+1}$ denotes the usual (singular) cohomology of topological spaces, and $[X,Y]$ represents the homotopy classes of continuous maps between the topological spaces $X$ and $Y$.

\subsubsection{Higher-Form Symmetries}

A direct generalization applies to $p$-form symmetries in spacetime dimensions $d$. In this scenario, we label the $(p+1)$-dimensional simplices of the triangulated spacetime using elements from the Abelian group $G$. This labeling is equivalent to assigning a mapping $A: M^d \rightarrow B^{p+1}G$, where $A$ can be regarded as a $(p+1)$-form gauge field. When the system undergoes retriangulation or a Pachner move, representing a gauge transformation between equivalent connections, the change is characterized by a phase factor $\omega_{d+1}$.

The phase factor $\omega_{d+1}$, associated with the retriangulation, can be understood as the phase factor of a $(d+1)$-dimensional simplex in the construction of the $p$-form $G$-SPT partition function. The consistency condition implies the cocycle condition for $\omega_{d+1}$: $\mathrm{d} \omega_{d+1} = 1$. By excluding the counterterm redefinitions or coboundaries, both the $p$-form 't Hooft anomaly and $p$-form SPT phases in one higher dimension are classified by
\be
H^{d+1}(B^{p+1}G,U(1)) \cong [B^{p+1}G,B^{d+1} U(1)].
\ee
This equation represents a straightforward generalization of Eq.~(\ref{eq:HSPT}) from 0-form symmetry to $p$-form symmetry~\cite{thorngren2015higher,kapustin2017higher}.

\subsubsection{Fermionic Systems}

The relationship between 't Hooft anomalies in $d$-dimensional spacetime and SPT phases in $(d+1)$-dimensional spacetime is not limited to bosonic systems; it extends to fermionic systems as well.

In fermionic systems, the symmetry structure is more complicated compared to bosonic systems. Unlike bosonic systems, fermionic systems always possess an unbreakable symmetry $\mathbb{Z}_2^f$ generated by fermion parity $(-1)^F$. A generic fermionic symmetry $G_f$ is an extension of fermion parity and a bosonic symmetry $G_b$, and it can be classified by the short exact sequence~\cite{wang2020construction}
\be
1\rightarrow \mathbb Z_2^f \rightarrow G_f \rightarrow G_b \rightarrow 1,
\ee
which is governed by a 2-cocycle $\lambda_2\in H^2_\mathrm{grp}(G_b,\mathbb Z_2^f)$.

Given the relationship between 't Hooft anomalies and SPT phases, in order to understand the 't Hooft anomaly in a fermionic system, we can explore the fermionic SPT phases in one higher dimension. Similar to the classification of bosonic systems using group cohomology, there exists a classification of fermionic SPT phases utilizing general group super-cohomology as a generalized cohomology theory~\cite{gu2014symmetry,wang2018towards,Xiong_2018,gaiotto2019symmetry,wang2020construction}.

In this classification scheme, multiple layers of cohomology groups come into play. For instance, a bosonic SPT phase can be regarded as a fermionic SPT phase, although it might become trivial in a fermionic system. Thus, the top layer corresponds to the bosonic layer represented by $\omega_{d+1} \in H^{d+1}(G,U(1))$. Furthermore, there exists a distinctive super-cohomology layer represented by $n_d \in H^d(G,\mathbb{Z}_2)$. This layer essentially involves the decoration of complex fermions into the interaction lines of the TDN in spacetime. The relationship between $n_d$ and $\omega_{d+1}$ is regulated by consistency condition. In the case of 1+1D anomaly or 2+1D fermionic SPT phases, the condition take a special form known as the super-pentagon equation~\cite{gu2014lattice,gu2015classification}. More generally, the consistency condition for fermionic SPT phases is expressed as~\cite{gu2014symmetry,wang2020construction}
\be
\dd \omega_{d+1} = (-1)^{\lambda_2 \smile n_d + Sq^2(n_d)},
\ee
Here, $Sq^2(n_d) = n_d \smile_{d-2} n_d$ denotes the cohomology operation of the Steenrod square~\cite{steenrod1947products}, which is related to the higher cup product. Notably, this equation reveals that $\omega_{d+1}$ is no longer a cocycle but only a cochain.

The third layer, often referred to as the Majorana or Kitaev chain layer denoted as $n_{d-1} \in H^{d-1}_\mathrm{grp}(G_b,\mathbb{Z}_2)$, plays a crucial role. It involves decorating the Kitaev chain onto the interaction surface of the TDN in spacetime, thereby introducing additional twists to the equations governing $n_d$ and $\omega_{d+1}$. While obstruction functions for this layer are well-known in lower dimensions~\cite{wang2018towards,wang2020construction,kapustin2017fermionic,brumfiel2016pontrjagin,brumfiel2018pontrjagin}, they become notably more intricate as the dimension increases.

Furthermore, apart from the necessity of satisfying the consistency equations, we must also account for ``gauge transformations'' of coboundaries stemming from local counterterms. It's noteworthy that fermionic systems exhibit a more intricate web of gauge redundancies. In principle, whenever there's an anomalous SPT state residing on the boundary of another SPT state, this boundary SPT state should be considered trivial~\cite{wang2019anomalous}. This is because there exists a symmetric, fermionic, local unitary transformation that connect this state to a product state.

By carefully accounting for these obstructions and gauge redundancies or trivializations, we can systematically construct and classify fermionic SPT phases and the associated 't Hooft anomalies in systems that are one dimension lower.

One can also extend the study of higher-form symmetry to fermionic systems, delving into the classification of higher-form symmetry-protected fermionic topological phases. This exploration closely relates to our understanding of 't Hooft anomalies of higher-form symmetry in quantum field theories involving fermions in dimensions one lower. Fortunately, we can adapt existing mathematical tools, including supercohomology theory and spin cobordism theory, to handle higher-form symmetries in these fermionic systems. The extension involves replacing $BG_b$ with $B^{p+1}G_b$ to accommodate higher-form gauge fields in all decoration data and obstruction functions. Importantly, the obstruction functions remain unchanged, with the understanding that their arguments now refer to the corresponding simplices labeled by higher-form gauge fields. Thus all the construction and classification schemes remain applicable.

\subsubsection{Cobordism Classification}

Previously, we discussed the correspondence between a $d$-dimensional 't Hooft anomaly and a $(d+1)$-dimensional SPT phase, both of which are classified by the cohomology group of the symmetry. However, this correspondence, while accurate in many cases, encounters limitations in higher dimensions. To address this and provide a more comprehensive framework, we introduce a refined mathematical tool known as bordism.

Let's consider the case of a 1+1D SPT state on the spacetime manifold $S^2$. As we discussed earlier, the partition function for such a theory is obtained by multiplying the phase factors $\omega_2$ associated with each triangle of the triangulation. To illustrate this, we can start by triangulating the 2-sphere $S^2$ as a single tetrahedron:
\begin{center}
    \begin{tikzpicture}
        \draw[->,thick] (0,0)--(0.5,-0.5) node[left]{$g$} ;
        \draw[thick] (0,0)--(1,-1)  ;
        \draw[->,thick] (1,-1)--(1,0.3) node[left]{$h$} ;
        \draw[thick] (1,0)--(1,1) ;
        \draw[->, thick] (1,1)--(1.5,0.5) node[right]{$k$};
        \draw[thick] (1.5,0.5)--(2,0) ;
        \draw[->, thick] (1,-1)--(1.5,-0.5) node[right]{$hk$} ;
        \draw[ thick] (1.5,-0.5)--(2,0);
        \draw[ thick,densely dotted] (0,0)--(2,0) node[right]{$ghk$};
        \draw[ thick,densely dotted,->] (0,0)--(1,0);
        \draw[ thick,->] (0,0)--(0.5,0.5) node[left]{$gh \ $};
        \draw[ thick] (0,0)--(1,1);
    \end{tikzpicture},
\end{center}
In this case, the partition function $Z[S^2]$ is given by the product of the phase factors associated with the tetrahedron. Mathematically, we can express it as:
\be Z[S^2] = \omega(g,h) \omega(gh,k)\omega^{-1} (g,hk) \omega^{-1} (h,k)=1. \ee
The key observation in the last step is that the cocycle condition of $\omega_2$ is satisfied, which ensures that the partition function is equal to 1. This means that the 1+1D SPT state on the 2-sphere is trivial. By applying this procedure recursively, we can generalize the result to the partition function of the boundary $Z[\partial M_3]$, which is also equal to 1. This result holds true not only for 1+1D, but also for higher dimensions.

Considering now the boundary decomposition $\partial M_3 = M_2 \sqcup \overline{M_2'}$, we can utilize two counter-oriented surfaces $M_2$ and $\overline{M_2'}$ to cancel out the partition function. This leads to the equation:
\be\ba  1 = Z[{M_2}]\cdot Z[\overline{M_2'}]  &= Z[M_2]\cdot Z[M_2']^{-1}   \\  
\Rightarrow Z[M_2]  &=  Z[M_2'] \ea\ee
This observation suggests the introduction of a new equivalence class.
\begin{definition}
Two $d$-manifolds $M_d$ and $M_d'$ are said to be \textbf{bordism equivalent} if there exists a $(d+1)$-dimensional manifold $M_{d+1}$ such that its boundary can be decomposed as $\partial M_{d+1} = M_d \sqcup \overline{M_d'}$.
\end{definition}

Two manifolds belonging to the same bordism class will have the same partition function under this topological construction. In addition, we define
\begin{definition}
    The bordism ring of manifolds is the triple $(\Omega_n^{SO},\sqcup,\times)$.    $\Omega_n^{SO}$ represents the set of $n$-dimensional oriented manifolds modulo the bordism equivalence relation. The operation $\sqcup$ denotes the disjoint union, which acts as the ``addition'' operation in the bordism ring. The operation $\times$ corresponds to the Cartesian product and serves as the ``multiplication'' operation in the bordism ring.
\end{definition}

In the construction of invertible topological phases, we not only consider manifolds but also manifolds equipped with flat $G$-connections. In this case, the theory assigns a value from the group $G$ to each edge, subject to the flatness condition. Thus, we have the maps:
\be  A: M_n/{\rm \{ bordism \ equivalence\} } \to BG.   \ee
Similarly, we can define the bordism class of manifolds with flat $G$-connections. The corresponding bordism group is denoted as $\Omega_n^{SO}(BG)$, which requires the connection to be consistent in both the bulk $M_{n+1}$ and the boundary.

The partition function $Z$ of invertible topological phases with symmetry $G$ is defined as a map from the bordism group to $U(1)$:
\be\label{eq:cobor}  Z: \Omega_n^{SO}(BG)\rightarrow U(1).   \ee
This partition function allows us to classify invertible topological orders by the isomorphism classes of the above maps. The classification is represented by the cobordism group $\Omega^n_{SO}(BG)$, which is the Pontryagin dual of the bordism group. If we disregard the specific gauge group $G$ and consider the case where $BG$ reduces to a point $pt$, we have $\Omega^n_{SO}(pt) = \Omega^n_{SO}$, as there is only one unique map from any space to a point.

The 't Hooft anomaly and SPT phases in one higher dimension are described by the cobordism~\cite{Kapustin:2014tfa,Kapustin:2014dxa,kapustin2014symmetry,freed2014shortrange,brumfiel2016pontrjagin,campbell2017homotopy,brumfiel2018pontrjagin,Wan:2018bns,Wan:2019soo,yonekura2019cobordism,freed2021reflection}. This new classification takes into account not only changes in the $G$-connection but also allows for changes in the topology of the spacetime manifold itself. It captures the intricate interplay between the symmetry $G$ and the underlying spacetime structures, providing a comprehensive framework for understanding 't Hooft anomaly and SPT phases.

On the other hand, the cohomology classification serves as an approximation that is often easier to calculate. in general, we have a group homomorphism from the group cohomology classification to the cobordism classification. However, this group homomorphism is neither injective nor surjective. the good news is that in low dimensions, this map is a isomorphism, so the two classifications coincide.

\begin{example}
The simplest examples of bordism groups for a point is list as follows.
\begin{table}[htbp]
    \centering
    \begin{tabular}{|c|c|c|c|c|c|c|c|c|c|}
    \hline
        $n $ & $0$ & $1$ & $2$ & $3$ & $4$ & $5$ & $6$ & $7$ & $8$ \\
        \hline
        $\Omega_n^{SO}(pt)  $ & $\bbZ$ & $0$ & $0$ & $0$ & $\bbZ$ & $\bbZ_2$ & $0$ & $0$ & $\bbZ^2$ \\
        \hline
    \end{tabular}
    \caption{List of bordism classification of a point.}
    \label{tab:bordpts}
\end{table}

In this example, the zero dimension is classified by difference of numbers of points with $\pm$ directions (a pair of $+$/$-$ direction points cancels each other); in dimension 4, it is classified by the first Pontryagin class $p_1$; in dimension 5, the product of Stiefel-Whitney classes $w_2w_3$ classifies the manifold.
\end{example}

The general classification of $\Omega^n(BG)$ can be quite intricate. However, there is a practical method known as the Atiyah-Hirzebruch spectral sequence that allows us to calculate $\Omega^n(BG)$ using the knowledge of $\Omega^n(pt)$, which we have discussed earlier for lower dimensions. The Atiyah-Hirzebruch spectral sequence has a physical interpretation as decorating bosonic invertible orders onto the domain walls of the group $G$~\cite{chen2014symmetry, wang2021domain}. This approach provides a useful tool for understanding and classifying symmetries and their associated topological orders.

The above bordism classification of 't Hooft anomalies and SPT phases can be easily extended to higher-form symmetries, at least conceptually. The key modification is to replace the 1-form gauge field $A$ with a $(p+1)$-form gauge field. As a result, the bordism group is now denoted as $\Omega_n(B^{p+1}G)$. In practice, the Atiyah-Hirzebruch spectral sequence can still be employed as a useful tool for calculating these bordism groups and understanding the higher-form symmetries involved. This generalization allows us to explore a wide range of topological phenomena and classify exotic phases of matter with higher-form symmetries.

In the context of fermionic systems, the notion of bordism equivalence extends to manifolds with spin (or other tangential) structures. Interestingly, in lower dimensions, the bordism classification and the general group super-cohomology classification coincide with each other. This correspondence allows us to leverage the domain wall decoration picture provided by the Atiyah-Hirzebruch spectral sequence, offering a more accessible method for calculating 't Hooft anomalies in fermionic systems.




\section{Applications of Higher-Form Symmetry}
\label{sec:applications}

In this section, we will illustrate some applications of higher-form symmetry to string theory and condensed matter physics.

We understand that our readers come from diverse academic backgrounds. Therefore, we encourage readers to engage with the topics that align with their current understanding and interests. Feel free to focus on familiar applications and skip or return to others as needed.

\subsection{Applications in String Theory}
\label{sec:string}

In this section we discuss the interplay between higher-form symmetry and string theory, which is a popular research direction in the recent years~\cite{DelZotto:2015isa,Bergman:2020ifi,Morrison:2020ool,Albertini:2020mdx,Bah:2020uev,Closset:2020scj,DelZotto:2020esg,Apruzzi:2020zot,Cvetic:2020kuw,BenettiGenolini:2020doj,Cordova:2020tij,DelZotto:2020sop,Gukov:2020btk,Heidenreich:2020pkc,Closset:2020afy,Apruzzi:2021phx,Apruzzi:2021vcu,Hosseini:2021ged,Cvetic:2021sxm,Buican:2021xhs,Braun:2021sex,Cvetic:2021maf,Cvetic:2021vsw,Apruzzi:2021mlh,Closset:2021lwy,Apruzzi:2021nmk,DelZotto:2022fnw,Genolini:2022mpi,Benini:2022hzx,Cvetic:2022imb,DelZotto:2022joo,Apruzzi:2022dlm,Hubner:2022kxr,Heckman:2022xgu,Heckman:2022suy,Damia:2022bcd,GarciaEtxebarria:2022vzq,Apruzzi:2022rei,vanBeest:2022fss,Heckman:2022muc,Antinucci:2022vyk,Grimm:2022xmj,Etheredge:2023ler,Amariti:2023hev,DelZotto:2023ahf,Acharya:2023bth,Cvetic:2023plv,Dierigl:2023jdp,Lawrie:2023tdz,Bah:2023ymy,Apruzzi:2023uma,Cvetic:2023pgm}. The general procedure here is to first geometrically construct QFTs from string theory, and then compute higher-form symmetries using geometric methods.

Roughly speaking, there are two ways to construct QFTs from string theory,
\begin{enumerate}
    \item AdS/CFT scenario. In this case we consider superstring/M-theory on ${\rm AdS}_{d+1} \times  M_{10-d / 9-d}$ ((10-d) for the 11D M-theory, (9-d) for the 10D superstring theory), which is dual to a certain QFT$_d$, often interpreted as the world-volume theory of brane objects.
    \item Geometric engineering scenario. In this case we consider superstring/M-theory on $\bbR^{d-1,1}\times X$ with $X$ being a non-compact (and possibly singular) space s.t.  the gravity sector is decoupled.

    But what do we mean by ``decoupled''? Consider the (oversimplified) model of Einstein-Hilbert action, we start from
    \be \int_{\mb{R}^{d-1,1}\times X}\frac{1}{G_{11}} \sqrt{-g_{11}} R_{11}  \mapsto \int_{\mb{R}^{d-1,1}}\frac{{\rm vol}(X)}{G_{11}} (\sqrt{-g_{d}} R_{d} + \cdots)  \ee
    where the subscript $_{11}/ _{d}$ means the quantities in the 11/d-dimensional theory. Then obviously the effective $d$-dimensional Einstein-Hilbert coupling constant shall have
    \be G_d = \frac{G_{11}}{{\rm vol}(X)} \ee
    as $X$ goes non-compact, ${\rm vol}(X)\to \infty$, the gravity coupling in $d$-dimension goes to $0$, and that is the picture of ``decoupled''\footnote{Note that in the limit of vol$(X)\rightarrow\infty$, there would be additional massless KK modes. Nonetheless in such geometric engineering setups, one only consider the localized modes near the origin (singular point).}.
\end{enumerate}

Now we introduce more specific setting of such theories, and start to observe the higher-form symmetries and 't Hooft anomaly polynomials in these theories.

Consider the IR limit of a 11D M-theory, which turns out to be a 11D supergravity (SUGRA) theory with $C_3$ gauge field (and to that extent, $G_4 = \dd C_3$ the field strength), and M2, M5 branes are the fundamental objects in M-theory.

\subsubsection{Example: AdS$_{d+1}$/QFT$_{d}$}

In this setting, we put the 11-dimensional M-theory on AdS$_{d+1}\times M_{11-d}$, where $M_{11-d}$ is a compact space with positive curvature. Now the general procedure is to expand $C_3/G_4$ along the cohomology elements of $M_{11-d}$, which give rise to (a general $(p+1)$-form) gauge field and field strength of the supergravity theory in AdS$_{d+1}$. Then, observe that with proper boundary conditions, the dynamical gauge fields in the bulk would also correspond to the background gauge fields of $p$-form global symmetries in QFT$_{d}$. The intuition for this part is that the global symmetries on the ``boundary'' corresponds to gauge symmetries in the ``bulk''.

The t' Hooft anomaly polynomial can be computed by the reduction of the 11d topological coupling
\be\int C_3\wedge G_4 \wedge G_4\ee
onto $M_{11-d}$. More generally, one can also compute the more general  SymTFT (symmetry topological field theory) action from the differential cohomology version, which includes discrete symmetries as well~\cite{Apruzzi:2021nmk,vanBeest:2022fss,Kaidi:2022cpf,Kaidi:2023maf,Chen:2023qnv}.

\subsubsection{M-Theory on $\bbR^{d-1,1}\times X_{11-d}$}
\label{sssect:M-th}
Here we require that $X_{11-d}$ has a cone structure, with a singularity at the tip of the cone,
\begin{figure}[htbp]
    \centering
    \begin{tikzpicture}
        \draw[thick] (0,0) -- (-1.92,-3.23);
        \draw[thick] (0,0) --  (1.92,-3.23) (1.5,-1.5) node[above]{$X$};
        \draw[thick] (0,-3.5) ellipse (2 and 1)  (2,-3.5) node[right]{$\partial X$};
        \filldraw[cyan] (0,0) circle (.1) (0,0) node[left]{singularity} ;
    \end{tikzpicture}
\end{figure}
for example, if $X$ is locally an intersection of equations
\be  f_1(x_1,\dots ,x_n) = f_2(x_1,\dots ,x_n) = \dots =f_k(x_1,\dots ,x_n) =0\subset \bbR[x_1,\dots , x_n] \ee
then the ``boundary'' $\ptl X$ is given by the subspace
\be  f_1 = f_2 =\dots =f_k = \sum_{i=1}^n x_i^2 -\varepsilon^2 = 0  \ee
which forms a $S^{n-1}$ with radius $\varepsilon$.

Now suppose that $\ptl X$ is smooth, and there are non-trivial torsional topological cycles $[ \alpha ]\in H_i(\ptl X,\bbZ)$ with torsion degree $l$ s.t.  $l[\alpha] = 0$. We can construct a subcone $\Sigma_{i+1}[\alpha]\subset X$ as the cone over $[\alpha]$, shown in Fig.~\ref{fig:cones}. Then we can construct the charged object under higher-form symmetries as M2 or M5 branes wrapping $\Sigma_{i+1}[\alpha]$\footnote{Precisely speaking, $\Sigma_{i+1}[\alpha]\in \frac{H_{i+1}(X,\ptl X,\mb{Z})}{H_{i+1}(X,\mb{Z})}$, where $H_{i+1}(X,\ptl X,\mb{Z})$ is the relative homology group of $(X,\ptl X)$.}, explicitly speaking,

\begin{enumerate}
    \item M2-brane over $\Sigma_{i+1}[\alpha]$ is charged under $\bbZ_l^{(2-i)}$, the electric $(2-i)$-form symmetry.
    \item M5-brane over $\Sigma_{i+1}[\alpha]$ is charged under $\bbZ_l^{(5-i)}$, the magnetic $(5-i)$-form symmetry.
\end{enumerate}

Note that the two symmetries listed above are generated by non-commutative torsional flux, and they are not mutually local~\cite{DelZotto:2015isa,Albertini:2020mdx}. In a well defined (absolute) theory, one need to choose a polarization, analogous to the case of $SU(N)$ gauge theory. Namely, one can either choose to keep the M2-branes, leading to the electric $(2-i)$-form symmetry, or the M5-branes, which leads to the magnetic $(5-i)$-form symmetry.

The mechanism can be also applied to  IIA/IIB string theory, where the charged objects are D$p$-branes instead of M2/M5-branes.

\begin{figure}[htbp]
    \centering
    \begin{tikzpicture}
        \draw[thick] (0,0) -- (-1.92,-3.23);
        \draw[thick] (0,0) --  (1.92,-3.23) (1.5,-1.5) node[above]{$X$};
        \draw[thick] (0,-3.5) ellipse (2 and 1)  (2,-3.5) node[right]{$\partial X$};
        \draw[thick,orange] (0,-3.5) ellipse (1 and 0.5)  (1,-3.5) node[right]{$\alpha$};
        \draw[thick,orange] (0,0) -- (1,-3.5);
        \draw[thick,orange] (0,0) -- (-1,-3.5) (-0.5,-1.75)node[left]{$\Sigma$} ;
        \filldraw[cyan] (0,0) circle (.1) (0,0) node[left]{singularity} ;
    \end{tikzpicture}
    \caption{}
    \label{fig:cones}
\end{figure}

\subsubsection{5D $\mathcal{N}=1$ SCFT}
In five dimensions, an important feature is that all the (SUSY-)gauge theories are strongly coupled in the UV, for example, the Yang-Mills action,
\be  S_{YM} = \int \dd^5 x \frac{1}{g_{YM}^2} (F_{\mu\nu}F^{\mu\nu}) \ \Rightarrow \ [g_{YM} ] = -\frac{1}{2}  \ee
as energy scale $\to \infty$, the unit of $g_{YM}\to 0$ and the gauge theory becomes infinitely strongly coupled. Under certain conditions, a SUSY gauge theory can be UV-completed into an $\cN=1$ SCFT with 8 supercharges, see for example~\cite{Seiberg:1996bd,Morrison:1996xf,Intriligator:1997pq,Aharony:1997bh,Benini:2009gi,Kim:2012gu,Bergman:2013aca,Zafrir:2014ywa,Hayashi:2015zka,Xie:2017pfl,Ferlito:2017xdq,Hayashi:2018lyv,Jefferson:2018irk,Bhardwaj:2018vuu,Closset:2018bjz,Cabrera:2018jxt,Apruzzi:2018nre,Bhardwaj:2018yhy,Apruzzi:2019vpe,Apruzzi:2019opn,Apruzzi:2019enx,Apruzzi:2019kgb,Bhardwaj:2019fzv,Bhardwaj:2020gyu,Eckhard:2020jyr,Bhardwaj:2020ruf,BenettiGenolini:2020doj,Bhardwaj:2020avz,Closset:2020scj,Closset:2021lwy,Tian:2021cif,Genolini:2022mpi,DelZotto:2022fnw,Collinucci:2022rii,DeMarco:2022dgh,Bourget:2023wlb}.

If in the vector multiplet\footnote{A supermultiplet is a representation of the supersymmetry algebra, which consists of particles with different spins. In the case of a massless vector multiplet, there is the vector field $A_\mu$, a Majorana spinor $\chi$ and a real scalar $\phi$.} $(A_\mu,\chi,\phi)$, the real scalar $\phi$ is given a non-zero vacuum expectation value (VEV, i.e. $\expval{\phi}\ne 0$). This non-zero VEV breaks the original possibly non-Abelian gauge theory/SCFT to a $U(1)^r$ ($r$ is the rank of the original gauge group) Abelian gauge theory with the charged matter being hypermultiplet\footnote{In the hypermultiplet, there is a majorana spinor $\psi$ and two complex scalars $\phi$ and $\Tilde{\phi}$.} $(\psi,\phi,\Tilde{\phi})$.

This admits a M-theory realization, the tool we developed in Sec. \ref{sssect:M-th} comes in handy. A 11D M-theory compactified on Calabi-Yau 3-fold $X$
\footnote{The 3 here denotes the complex dimension, in terms of real dimension is 6. Being Calabi-Yau means that it has a trivial canonical bundle.} 
with singularity. The process from UV SCFT to Coulomb branch is incorporated to the crepant resolution\footnote{Formally, a crepant resolution is a birational map  which keeps the canonical divisor invariant.}, which is a smoothing of the singularity while preserving the supersymmetry of the field theory.
\begin{center}
    \begin{tikzpicture}[scale = 0.8]
        \draw[thick] (0,0) ellipse (2 and 1);
        \draw[thick] (1.95,0.2) -- (0,3.5)  (-1.95,0.2) -- (0,3.5) ;
        \filldraw[cyan]  (0,3.5) circle (0.12);
        \node at (3,2) { crepant \ resolution};
        \node at (-1.1,1.75) [left] {$X$: the UV complete};
        \node at (7.1,1.75) [right] {$\Tilde{X}$: the Coulomb branch};

        \draw[->] (2,1.7)--(4,1.7);
        \draw[shift = {(6,0)},thick] (0,0) ellipse (2 and 1);
        \draw[shift = {(6,0)},thick] (1.95,0.2)--(0.236,3.1) (-1.95,0.2)--(-0.235,3.1);
        \draw[shift = {(6,0)}, thick] (0,3.1) ellipse (0.236 and 0.117);
    \end{tikzpicture}
\end{center}

In the resolved $\Tilde{X}$, there are compact 4-cycles $S_i \ (i\in {1,\cdots ,r})$, where each corresponds to a gauge group $U(1)_i$ on the Coulomb branch. Using the expansion of M-theory $C_3$ gauge field we have
\be C_3 = \sum_{r=1}^r A_i \wedge \omega_i  \ee
where $\omega_i$ are Poincar\'{e} dual 2-form of $S_i$, and $A_i$ are the $U(1)_i$ gauge fields. There are also compact 2-cycles $C_j$, the M2-branes wrapping around $C_j$ gives BPS particles in $5$D, whose electric charges under $U(1)_i$ is defined as
\be q_{i,j} = {\rm Int}(C_j,S_i)|_{\Tilde{X}}\,,  \ee
the intersection number.

With certain conditions, we can take the geometric limit from $\Tilde{X}$. Suppose that each $s_i$ admits a $S^2$-fibration structure, and take the limit of ${\rm vol}(S^2)\to 0$. The picture is that this is the limit of a non-Abelian gauge theory, since M2-branes over $S^2$ are W-bosons, so by taking the volume to zero, we make the corresponding (w.r.t. index $i$) W-bosons massless. This is consistent with the IR physical description of a 5D non-Abelian gauge theory.

As an detailed example, we consider the case of a toric CY3 singularity $X$, which is a cone over a Sasaki-Einstein fivefold\footnote{A Sasaki-Einstein manifold $X$ is defined as a manifold whose metric cone $\mb{R}_{>0}\times X$ is Kahler and Ricci-flat~\cite{Sparks:2010sn}. For $Y^{N,k}$, it can be described as an $S^1$-bundle over $S^2\times S^2$, whose first Chern class $c_1$ is parametrized by $N$ and $k$.} $Y^{N,k}$. We draw the crepant resolution of $X$ as a toric diagram (for the notations of toric CY3, see for example  \cite{Xie:2017pfl,Eckhard:2020jyr})
\begin{center}
    \begin{tikzpicture}[scale = 0.8]
        \draw[thick] (0,1) -- (0,2) (0,1)-- (1,1) (0,2) -- (1,1) (0,2) -- (-1,-2) (0,1) -- (-1,-2) (0,-1) -- (0,-2) (0,-1) -- (1,1) (0,-2) -- (1,1) (0,-1) -- (-1,-2) (0,-2) -- (-1,-2);
        
        \node at (0,2.5) { $(0,N)$};
        \node at (0,-2.5) { $(0,0)$};
        \node at (-1.7,-2.5) { $(-1,0)$};
        \node at (2,1) { $(1,N-k)$};
        \node at (0,0.3) { $\cdot$};
        \node at (0,0) { $\cdot$};
        \node at (0,-0.3) { $\cdot$};
    \end{tikzpicture}
\end{center}
The compact divisors $S_i$ $(i=1,\dots,N-1)$ are all $S^2$ fibrations over $S^2$. The theory has an IR non-Abelian gauge theory description $SU(N)_k$, where $k$ is the Chern-Simons level.

One can compute the 1-form symmetry using three different methods, which give the identical result. Here we choose the polarization such that we have the maximal 1-form symmetry for the given geometric setup.

\begin{enumerate}
\item{One can compute the Smith normal form of the charge matrix $q_{i,j}$, see \ref{SmithDecomp}. From the Smith normal form, one can read off the 1-form symmetry $\Gamma^{(1)}=\mb{Z}_{\text{gcd}(N,k)}$.}
\item{From the homology of the link fivefold $Y^{N,k}$, we have
\be
\ba
&H_0(Y^{N,k},\mb{Z})=\mb{Z}\ ,\ H_1(Y^{N,k},\mb{Z})=\mb{Z}_{\text{gcd}(N,k)}\ ,\ H_2(Y^{N,k},\mb{Z})=0\ ,\ \cr
&H_3(Y^{N,k},\mb{Z})=\mb{Z}_{\text{gcd}(N,k)}\ ,\ H_4(Y^{N,k},\mb{Z})=0\ ,\ 
H_5(Y^{N,k},\mb{Z})=\mb{Z}\,.
\ea
\ee
}
\item{From field theory arguments, the 1-form center symmetry $\mb{Z}_N$ of $SU(N)$ is broken to 
the subgroup $\mb{Z}_{\text{gcd}(N,k)}$ by the 5d Chern-Simons term.}
\end{enumerate}

\subsection{Applications in Condensed Matter Physics}
\label{sec:appl-CMP}

A natural framework for studying generalized symmetries is provided by 2+1D topological orders with anyonic excitations, which have been extensively studied in the field of condensed matter physics. These systems are mathematically described by unitary modular tensor categories, which provide a powerful formalism for capturing the fusion and braiding structures of the anyons~\cite{KITAEV20062}. The worldlines traced out by the anyons in these systems can be interpreted as the topological defect lines (TDL) of the system. The crucial point of these operators is that, unlike global symmetry actions that act on the entire system, they only act on a subset of the degrees of freedom in the system~\cite{PhysRevB.72.045137,pnas.0803726105,NUSSINOV2009977}. In this sense, it is a generalization of gauge symmetry transformation that acts on each individual point of the system.

When the anyons are Abelian, all TDLs in the system are invertible. In this case, the symmetries generated by the worldlines of the anyons correspond to 1-form symmetries. However, if the anyons are non-Abelian, implying that the TDL is non-invertible, the symmetries associated with the anyon worldlines are non-invertible symmetries.

To illustrate these concepts, we will consider the famous toric code model as an explicit example. We will demonstrate that the 1-form symmetry arising from the anyon worldlines in this 2+1D system is anomalous. Moreover, we will show that this system with an anomalous 1-form symmetry can be realized as the boundary of a 3+1D 1-form SPT state.

\subsubsection{Toric Code Model}

The toric code model~\cite{KITAEV20032} is a well-known example of a 2+1D lattice gauge theory with a $\mathbb{Z}_2$ gauge group. In this model, we consider a 2D space with a square lattice, where each lattice link is associated with a spin-$1/2$ degree of freedom. The Hilbert space of the system is given by the tensor product of complex two-dimensional vector spaces $\mathbb{C}^2$ for each link:
\be
\mathcal{H} = \bigotimes_{\text{each link}} \mathbb{C}^2.
\ee

We define the following two kinds of operators in the toric code model:
\begin{center}
\begin{tikzpicture}
    \draw[thick] (-0.5,0)--(0.5,0) (-1,0)node[left]{$A_s = $};
    \draw[thick] (0,-0.5)--(0,0.5) (1,0) node[right]{$ = \bigotimes_{l, \ptl l\supset {x}} \sigma^{x}_{l}$};
    \filldraw[orange] (0,0) circle (2pt) (-0.2,0)node[above]{$s$};
    \draw[thick] (-0.5,-2) rectangle (0.5,-1) (-1,-1.5)node[left]{$B_p = $} (0,-1.2)node[below]{$p$} (1,-1.5)node[right]{$ = \bigotimes_{ l\in \ptl p} \sigma^{z}_{l}$};
\end{tikzpicture}
\end{center}
with each vertex labeled by $s$ and plaquette by $p$. It is easy to verify the following properties of the operators
\be\ba 
A_s^\dag &= A_s,  \ \ A_s^2 = 1 \Rightarrow \ {\rm Spec}(A_s) = {\pm 1} ,\\ 
B_p^\dag &= B_p, \ \ B_p^2 = 1 \Rightarrow \ {\rm Spec}(B_p) = {\pm 1} ,\\
[A_s,A_{s'}] &= [B_p,B_{p'}]= [A_s,B_p]  = 0 , \ \forall s,s',p,p' .
\ea\ee

The Hamiltonian of the toric code model is a summation of all commuting projector terms associated with $A_s$ and $B_p$ operators, given by:
\be
H = -\left(\sum_{s} A_s + \sum_{p} B_p\right).
\ee
Since all terms in the Hamiltonian commute with each other, we can easily determine the entire spectrum of the model:
\be\ba
{\rm Ground \ states} &:  A_s\ket{0} = B_p\ket{0} = \ket{0} \\
{\rm Excited \ states} &: A_s \ket{\Psi} = -\ket{\Psi}, \text{ or } B_p \ket{\Psi} = -\ket{\Psi},  \ {\rm for \ some \ }s,p
\ea\ee
This simplifies the process of finding the ground state and excited states, making it feasible to analyze the behavior of the toric code model and study its excitations.

There is an intuitive correspondence between the toric code model and electromagnetism, which is described by a $U(1)$ (lattice) gauge theory.
\begin{table}[htbp]
    \centering
    \begin{tabular}{ccc}
        Discrete $\bbZ_2$ gauge theory & $\longrightarrow$  &  $U(1)$-gauge theory  \\
        $\sigma^z_l$ &   $\longrightarrow$  &  Wilson line $\exp(i\int A_{ij})$  \\
        $B_p$ &   $\longrightarrow$  &  $\exp(i\oint A)=  \exp(i\int F) \sim \exp(ia^2 B) $  \\
        $\sigma^x_{l}$ &   $\longrightarrow$  &  $\exp(i\int_l E)$  \\
        $A_s$ &   $\longrightarrow$  & $\exp(i \nabla\cdot E)$   \\
        $H = -(\sum_s A + \sum_p B)$ &   $\longrightarrow$  &  $\vb{B}^2 + (\nabla\cdot \vb{E})^2$  \\
        $B_p = \pm 1$ &  $\longrightarrow$   &  $ \int B =0 \ / \ \ne 0$  \\
        $A_s = \pm 1$ &   $\longrightarrow$  &  $\nabla E = 0\ / \ \ne 0$  \\

    \end{tabular}
    \label{tab:Z2U1}
\end{table}
This correspondence also helps us in constructing excitations in the toric code model. The 1-form symmetry can be understood through the string operators associated with these excitations, as shown in the following figure:
\begin{center}
\begin{tikzpicture}
    \draw[] (-3.3,3.2)--(3.3,3.2);
    \draw[shift = {(0,-0.8)}] (-3.3,3.2)--(3.3,3.2);
    \draw[shift = {(0,-1.6)}] (-3.3,3.2)--(3.3,3.2);
    \draw[shift = {(0,-2.4)}] (-3.3,3.2)--(3.3,3.2);
    \draw[shift = {(0,-3.2)}] (-3.3,3.2)--(3.3,3.2);
    \draw[shift = {(0,-4.0)}](-3.3,3.2)--(3.3,3.2);
    \draw[shift = {(0,-4.8)}](-3.3,3.2)--(3.3,3.2);
    \draw[shift = {(0,-5.6)}](-3.3,3.2)--(3.3,3.2);
    \draw[shift = {(0,-6.4)}](-3.3,3.2)--(3.3,3.2);
    \draw[] (3.3,-3.2)--(3.3,3.2);
    \draw[shift = {(-0.8,0)}] (3.2,-3.3)--(3.2,3.3);
    \draw[shift = {(-1.6,0)}] (3.2,-3.3)--(3.2,3.3);
    \draw[shift = {(-2.4,0)}] (3.2,-3.3)--(3.2,3.3);
    \draw[shift = {(-3.2,0)}] (3.2,-3.3)--(3.2,3.3);
    \draw[shift = {(-4.0,0)}] (3.2,-3.3)--(3.2,3.3);
    \draw[shift = {(-4.8,0)}] (3.2,-3.3)--(3.2,3.3);
    \draw[shift = {(-5.6,0)}] (3.2,-3.3)--(3.2,3.3);
    \draw[shift = {(-6.4,0)}] (3.2,-3.3)--(3.2,3.3);
    \draw[orange, ultra thick]  (1.6,3.2)--(1.6,2.4)--(1.6,1.6)--(0.8,1.6)--(0.8,0.8)--(1.6,0.8)--(2.4,0.8)--(2.4,-1.6) (5,1.5)node[right]{Wilson line operator, $W_E(L) = \prod_{l\in L} \sigma^z_l  $};
    \draw[blue,ultra thick, densely dotted] (0.4,1.2)--(-0.4,1.2)--(-0.4,0.4)--(-0.4,-0.4)--(-0.4,-1.2)--(-0.4,-2.0)--(0.4,-2.0)--(1.2,-2.0)--(2.0,-2.0)--(2.0,-2.8)--(2.8,-2.8);
    \draw[blue,ultra thick] (0,0.8)--(0,1.6)     (5,-1.5)node[right]{'t Hooft line operator, $W_M(L^*) = \prod_{l\perp L^*} \sigma^x_l$};
    \draw[blue,ultra thick] (0,0.8)--(-0.8,0.8);
    \draw[shift = {(0,-0.8)},blue,ultra thick] (0,0.8)--(-0.8,0.8);
    \draw[shift = {(0,-1.6)},blue,ultra thick] (0,0.8)--(-0.8,0.8);
    \draw[shift = {(0,-2.4)},blue,ultra thick] (0,0.8)--(-0.8,0.8);
    \draw[shift = {(2.4,-3.2)},blue,ultra thick] (0,0.8)--(-0.8,0.8);

    \draw[shift = {(0,-3.2)},blue,ultra thick] (0,0.8)--(0,1.6);
    \draw[shift = {(0.8,-3.2)},blue,ultra thick] (0,0.8)--(0,1.6);
    \draw[shift = {(1.6,-3.2)},blue,ultra thick] (0,0.8)--(0,1.6);
    \draw[shift = {(2.4,-4)},blue,ultra thick] (0,0.8)--(0,1.6) ;
    
\end{tikzpicture}
\end{center}
With the above definition of Wilson and 't Hooft operators, one can show the commutators of then are (when there are non-zero intersections)
\be\ba
[ W_E(L), B_p ] &= 0,  \ \ [W_M(L^*),B_p]\ne 0 \text{ iff } p\in\partial L^\ast \\
[W_E(L),A_s]&\ne 0 \text{ iff }s\in \partial L, \ \ [W_M(L^*), A_s] = 0  \\
[W_E(L) ,H] &= [W_M(L^*),H] = 0 \ {\rm if} \ \ptl L = \partial L^\ast=  0\\
[W_E(L), W_M(L^*) ] &= 0  \ , \ {\rm if \ } L \cap L^\ast = 0 
\ea\ee
Now we can state the following fact that \textbf{excitations of toric code system are endpoints of open line operators.} More specifically, we have
\be
 \ba
s &\notin \ptl L \  : \ A_s W_E(L) \ket{0} = W_E(L) A_s \ket{0} =  W_E(L) \ket{0} \\
s &\in \ptl L \ : \ A_s W_E(L) \ket{0} = -W_E(L) A_s \ket{0} =  -W_E(L) \ket{0}
\ea  \ee
So electric charge $e$ are ends of the Wilson line. Similarly, for the magnetic part, we have
\be
 \ba
p &\notin \ptl L^* \  : \ B_p W_M(L) \ket{0} = W_M(L^*) B_p \ket{0} =  W_M(L^*) \ket{0} \\
p &\in \ptl L^* \ : \ B_p W_M(L) \ket{0} = -W_M(L^*) B_p \ket{0} =  -W_M(L^*) \ket{0}
\ea   \ee
so magnetic charges $m$ are excitations at the ends of 't Hooft lines.

In contrast to open lines give excitations, the closed loops are conserved charges classifying degenerate states. We are interested in ground states. Let's assume the global structure of this lattice is a torus $T^2 = S^1\times S^1$, then there are two types of non-trivial loops. Considering $W_E(L_{1,2}) = \pm 1 $ or $W_M(L_{1,2})=\pm 1$, the degeneracy is at least 4.

Through a simple calculation of Pauli matrices acting on Hilbert space of 2D spatial manifold, one can see that
\be\label{EM} W_E(L) W_M(L^*) = (-1)^{{\rm Int}(L,L^*)}   W_M(L^*)W_E(L), \ee
which means the 't Hooft lines are the objects that Wilson lines act upon, and vice versa. Here, $\mathrm{Int}(L,L^\ast)$ is the intersection number of the two lines $L$ and $L^\ast$ on 2D spatial manifold. With that, we are now ready to talk about the statistics of the excitations.

Suppose we are also considering the time dimension, and the lines are now drawn within 2+1 dimensional spacetime. 
Reordering the lines results in a negative sign due to the anticommutation of $\sigma_l^x$ and $\sigma_l^z$ along the same link. We can illustrate this concept visually as
\begin{center}
\begin{tikzpicture}
    \draw[orange , thick ] (0,0)--(0.45,0.45)  (0.55,0.55)--(1,1);
    \draw[blue, thick]  (1,0)--(0,1);
    \draw[thick , -> ] (-0.2,-0.2)--(-0.2,1.2);
    \draw[thick] (-0.2,0)--(-0.2,1) (1.1,0.5)node[ right]{$= \ (-1)$};
    \draw[shift = {(3,0)},thick , -> ] (-0.2,-0.2)--(-0.2,1.2);
    \draw[shift = {(3,0)},orange ,thick] (0,0)--(1,1);
    \draw[shift = {(3,0)},blue, thick] (0,1)--(0.45,0.55) (0.55,0.45)--(1,0);
\end{tikzpicture}
\end{center}
This implies the ``Aharonov-Bohm effect'' of this system
\begin{center}
    \begin{tikzpicture}
        \draw[thick,->] (-0.20,-0.1)--(-0.20,4.1) node[left]{$t$};
        \draw[shift = {(3.7,0)},thick,->] (-0.20,-0.1)--(-0.20,4.1) node[left]{$t$};

        \draw[thick] (-0.20,0)--(-0.20,1) (2,2)node[right]{$= (-1)$};
        \draw[orange,thick] (0,0) arc (180:90:1) (1,1) arc (-90:87:1) (0,4) arc (180:268:1);
        \draw[thick, blue] (1,0)--(1,0.9) (1,1.1)--(1,4);
        \draw[shift = {(3.7,0)},thick, blue] (1,0)--(1,4) ;
        \draw[shift = {(3.7,0)},thick, orange] (0.2,0)--(0.2,4) ;

    \end{tikzpicture}.
\end{center}
and the value of a Hopf link if we close up the end points of the lines:
\begin{center}
    \begin{tikzpicture}
        \draw[thick , orange] (1,0) arc (0:295:1) (1,0) arc (0:-55:1) ;
        \draw[thick , blue] (2,0) arc (0:-235:1) (2,0) arc (0:115:1);
        \draw[thick,->] (-1.2,-1)--(-1.2,1) node[left]{$t$};
        \draw[thick] (-1.2,0)--(-1.2,0.1) (2.0,0) node[right]{$= \ (-1)$} (6,0) node[right]{$= \ (-1)$} ;
        \draw[shift = {(1.6,0)},thick , orange] (3,0) arc (0:360:0.4) ;
        \draw[shift = {(2,0)},thick , blue] (3.5,0) arc (0:360:0.4) ;
      
    \end{tikzpicture}
\end{center}

From the perspective of higher-form symmetries, it is evident that the 2+1D toric code model exhibits a $\mathbb{Z}_{2,E} \times \mathbb{Z}_{2,M}$ 1-form symmetry, characterized as follows:
\be\ba
\bbZ_{2,E}^{(1)}: W_E(L) W_M(L^*) &= (-1)^{{\rm Link}(L,L^*)}  W_M(L^*), \\
\bbZ_{2,M}^{(1)}: W_M(L^\ast) W_E(L) &= (-1)^{{\rm Link}(L^\ast,L)}  W_E(L) .
\ea\ee
These equations should be understood in the context of 2+1D spacetime, rather than 2D space as in Eq.~(\ref{EM}). It's noteworthy that the Wilson line operator and the 't Hooft line operator can be interchanged as the 1-form symmetry operator and the charged operator for these two distinct 1-form symmetries.

This leads us to the question of anomalies of this 1-form symmetry. In fact, the two $\Z_2$ 1-form symmetries have a mixed 't Hooft anomaly. It can be interpreted from three different perspectives:
\begin{enumerate}
    \item The configurations of Wilson loops and 't Hooft loops in 2+1D spacetime correspond to a defect network of the 1-form symmetry. As the two $\Z_2$ 1-form symmetries are distinct, their defect lines should remain independent of each other. However, as discussed earlier, the defect network is not topological in nature:
    \be
    \begin{tikzpicture}
    \label{grf:link}
        \draw[thick , orange] (1,0) arc (0:295:1) (1,0) arc (0:-55:1) ;
        \draw[thick , blue] (2,0) arc (0:-235:1) (2,0) arc (0:115:1);
        \draw[thick,->] (-1.2,-1)--(-1.2,1) node[left]{$t$};
        \draw[thick] (-1.2,0)--(-1.2,0.1) (2.0,0) node[right]{$= \ (-1)$} (6,0) node[right]{$= \ (-1)$} ;
        \draw[shift = {(1.6,0)},thick , orange] (3,0) arc (0:360:0.4) ;
        \draw[shift = {(2,0)},thick , blue] (3.5,0) arc (0:360:0.4) ;
      
    \end{tikzpicture}
\ee
This phenomenon can be understood as a mixed anomaly, akin to the previously mentioned breaking of associativity in the 1+1D case.
    \item Another perspective to consider this anomaly is through the lens of gauge transformations. The action of toric code can be written as a Chern-Simons theory
    \be S = \int_{M_3}  \frac{2}{2\pi} b^{(1)}\wedge \dd a^{(1)}  \ee
    with the Wilson and 't Hooft line operators
    \be\left\{ \ba
    W_E(L) &= \exp(i\int_L a),  \ \ \ \ \bbZ_{2,M}^{(1)} {\rm \ charged\ operator} \\
    W_M(L) &= \exp(i\int_L b), \ \  \ \ \bbZ_{2,E}^{(1)} {\rm \ charged\ operator} 
    \ea \right. \ee
    Now, if we couple this theory to background 2-form gauge fields, denoted as $B_{E,M}^{(2)}$, in an attempt to gauge the $\bbZ_{2,E}^{(1)}\times \bbZ_{2,M}^{(1)}$ 1-form symmetry, we obtain the following action
    \be S = \int_{M_3}  \frac{2}{2\pi} b^{(1)}\wedge \dd a^{(1)}  + \int_{M_3} \frac{1}{2\pi} (a^{(1)}\wedge B_M^{(2)}+ b^{(1)}\wedge B_E^{(2)}). \ee
    Under a $\bbZ_{2,E}^{(1)}$ gauge transformation
    \be\left \{ \ba
    B_E &\mapsto B_E - \dd \chi \\
    a &\mapsto a+ \frac{1}{2} \chi 
    \ea \right. \ee
    it's evident that the action is not gauge-invariant due to the second term, even though the other two terms remain invariant. Similarly, for the magnetic gauge transformation, the action also lacks gauge invariance. However, if we couple the theory to a 3+1D bulk with the action
    \be\label{eq:Sbulk} S_{\rm bulk} = \frac{1}{4\pi^2}\int_{M_4} B_E \wedge B_M, \ee
    we find that the total theory becomes gauge-invariant, indicating that the mixed anomaly of the two 1-form symmetries is governed by the above anomaly polynomial in one higher dimension.
    \item From a classifying space perspective, let's consider two gauge fields as maps from spacetime to the classifying space:
    \be B,B': M_3 \to B^2(\bbZ_{2,E}\times \bbZ_{2,M} ) \ .\ee
    Suppose these two gauge field maps correspond to the left and right-hand side defect network configurations of nontrivial and trivial links in Eq.~\ref{grf:link}. Then, the maps of the two configurations are homotopic since the two 1-form symmetries are independent of each other. In a non-anomalous theory, the partition function should be the same for both configurations. However, in this case, the partition functions differ by a $(-1)$ factor in Eq.~\ref{grf:link}, signifying the presence of a mixed anomaly.
\end{enumerate}

\subsubsection{$\bbZ_{2,E}^{(1)}\times \bbZ_{2,M}^{(1)}$ SPT in 3+1D}

Expanding our comprehension of previous $\bbZ_{2,E}\times \bbZ_{2,M}$ 1-form symmetry anomaly, we can try to understand it from a 3+1D 1-form SPT point of view. The 3+1D $\bbZ_{2,E}^{(1)}\times \bbZ_{2,M}^{(1)}$ SPT is classified by
\be\ba  
H^4( B^2 \bbZ_{2,E}\times  B^2\bbZ_{2,M} ,U(1)) &\supset H^2(B^2 \bbZ_{2,E} , H^2( B^2\bbZ_{2,M} , U(1))) \\
&\cong  H^2(B^2 \bbZ_{2,E} ,\bbZ_2) \cong \bbZ_2 .
\ea\ee
It is one of the terms in the spectral sequence calculation of the cohomology groups. Remarkably, this $\bbZ_2$ classification is generated by the term $B_E^{(2)}\wedge B_M^{(2)}$. The resulting action for the 3+1D SPT is expressed as:
\be S_\mathrm{SPT} = \frac{1}{4\pi^2}\int_{M_4} B_E \wedge B_M.  \ee
This outcome neatly corroborates the anomaly polynomial approach detailed in Eq.~(\ref{eq:Sbulk}). Essentially, it reaffirms that the mixed anomaly of the two 1-form symmetries in 2+1D can be comprehended through a 3+1D SPT characterized by the wedge product of two 2-form gauge fields.

Now, let's consider why the anomalies from the bulk and boundary actions are cancelled from the point of view of topological defect networks. The line operator in 2+1D can be naturally extended to a surface operator in the 3+1D bulk. Given that the total spacetime dimension is four, the intersection of two such surface operators occurs at discrete points within 3+1D spacetime. Applying Poincar\'{e} duality, this intersection can be equivalently understood as the integral of the wedge product of $B_E$ and $B_M$:
\be  B_E\wedge B_M \stack{\text{Poincar\'{e} dual}}{\longleftrightarrow} S_E\cap S_M. \ee
Consequently, the action of the 3+1D 1-form SPT is essentially a count of the parity of the intersection number between the two surfaces that are dual to the 2-form gauge fields.

This geometrical interpretation beautifully elucidates the cancellation of anomalies between the bulk and boundary actions. In fact, when we try to extend the configuration of a Hopf link in 2+1D spacetime [shown on the left-hand side of Eq.~(\ref{grf:link})] to one higher dimension, the most straightforward extension to a surface in 4+1D would involve a single intersection point. Following this geometric reasoning, the combined bulk and boundary theories as a whole becomes non-anomalous:
\be  (-1)^{\rm Intersection \ number\ in\ 3+1D} \cdot (-1)^{\rm Linking \ number\ in\ 2+1D} = 1.  \  \ee
From this perspective, we gain a understanding of the anomaly cancellation purely from a geometric perspective.

\begin{center}
    \begin{tikzpicture}[scale = 0.8]
        \draw (0.5,0) ellipse (3 and 2);
        \draw[shift = {(0.5,0)}] (-3,0) .. controls (-3.3, 5.5) and (3.3,5.5) .. (3,0);
        \draw[ultra thick, orange] (-1,0) .. controls (-1.1,2.5) and (1.1,2.5) .. (1,0);
        \draw[shift = {(1,0)},ultra thick, blue] (-1,0) .. controls (-1.1,2.5) and (1.1,2.5) .. (1,0);

        \draw[thick , orange] (1,0) arc (0:295:1) (1,0) arc (0:-55:1) ;
        \draw[thick , blue] (2,0) arc (0:-235:1) (2,0) arc (0:115:1);

    \end{tikzpicture}
\end{center}

One can further enhance the understanding by examining explicit lattice models of 1-form SPT and its boundary model. This allows for a closer investigation of how bulk symmetry operators interact with boundary operators and how anomaly cancellation occurs. An illustrative example involving $\mathbb{Z}_N$ 1-form symmetry is provided in Ref.~\citeonline{PhysRevB.101.035101}.

\section{Higher Group Symmetry}
\label{sec:higher-group}

In this chapter we present another way to generalize symmetry using the language of category theory. We will observe how to describe a group as a 1-category, and how to generalize it into a higher category language.

For readers more focused on practical applications than mathematical details, skipping the initial sections and starting with section~\ref{sec:2-group-phys} is recommended. This section presents the physical interpretations in a user-friendly manner. Subsequently, section~\ref{sec:2-group-gauging} includes practical physical examples, making the concepts approachable and easier to grasp.

\subsection{A Touch of Category Theory}
\label{App:Cat}
Here we present some basic knowledge on category theory, 
\begin{definition}
A {\rm\textbf{category}} C has the following structure:
\begin{enumerate}
    \item A collection of \textbf{objects}, denoted ${\rm Obj}(C)$
    \item A collection of \textbf{morphisms}, denoted ${\rm Hom}_C(\cdot , \cdot )$. Morphisms can be thought of as relations connecting two objects. For example, a morphism from $X \in {\rm Obj}(C) $ to $Y \in {\rm Obj}(C)$ can be $f:X\to Y$ as a usual map (we will see examples where morphisms are not maps). The set of morphisms from $X$ to $Y$ is denoted ${\rm Hom}_C(X,Y)$.
    \item $\forall X\in {\rm Obj}(C), \exists {\rm id}_X \in {\rm Hom}_C(X,X)$ as the identity map from X to itself, namely, $\forall f\in {\rm Hom}_C(X,Y), f\cdot id_X = id_Y\cdot f = f$.
    \item Composition. $\cdot :{\rm Hom}_C(Y,Z)\times {\rm Hom}_C(X,Y)\to {\rm Hom}_C(X,Z),\quad (f,g)\mapsto f\cdot g$
    \item Associativity law, $f(gh) = (fg)h$.
\end{enumerate}
\end{definition}

We usually use \textbf{commutative diagrams} to express these relations. For example, the composition is written as: The following diagram commute

\begin{center}
\begin{tikzcd}
X \arrow[rd, "g"'] \arrow[rr, "f\circ g"] &                    & Z \\
                                          & Y \arrow[ru, "f"'] &  
\end{tikzcd}
\end{center}

We can see that this category structure can describe many things we have known by now.

\begin{example}
The collection of all groups, {\rm \textbf{Grp}}, is a category. Obviously, its objects are groups, and we assign the morphisms to be group homomorphisms. One can check that it satisfies all requirement we need for a category.
\end{example}

\begin{example}
The collection of all topological spaces, {\rm \textbf{Top}}, is a category, with objects being the collection of all topological spaces and morphisms being continuous maps.
\end{example}

A frequently asked question is: do we require morphisms to be maps? The answer is \textbf{NO}! If so, we would not have to assert the associative law, $(fg)h = f(gh)$. In the next example we will see a category with its morphism not being maps.

\begin{example}
Consider a category {\rm \textbf{C}}, with{\rm
\begin{enumerate}
    \item Obj(\textbf{C}) = \{Finite dimensional vector spaces\}
    \item \be\ba  {\rm Hom}_C (V,W) &= \{ 1\}, \ {\rm if} \ {\rm dim}(V) = {\rm dim}(W) \\ {\rm Hom}_C (V,W) &=\{ 0 \} , \ {\rm if} \ {\rm dim}(V) \ne {\rm dim}(W) \ea\ee
    Composition of morphisms is multiplication of numbers.
\end{enumerate}}
We can observe that the morphism is \textbf{definitely not} a map from one vector space(an object) into another vector space (another object), but rather signifies an equivalence of dimension.
\end{example}

From the example above, we can see that morphism cannot be accurately interpreted as maps among objects, but rather, ``connections" between objects that could be rather casually assigned.\footnote{A even more peculiar case is when the objects are all the sets and the morphism being Cartisian product. Why? Because it does not even have to satisfy the composition law! Instead, it has a weakened composition law with $(X\times Y)\times Z$ isomorphic to $X\times (Y\times Z)$.}

Many other examples can be raised. To get to our point, we want to see how to add structures so that the category will look more like a group. Some intuition would be to make our ``connection" between objects ``invertible", so that we could talk about the ``isomorphism" or ``equivalence" between objects.

\begin{definition}
A {\rm\textbf{groupoid}} $\CG$ is a category with inverse structure in morphisms. Namely, $\forall f\in {\rm Hom}_{\CG}(X,Y),\exists f^{-1}\in {\rm Hom}_{\CG}(Y,X),st.  \ f^{-1}f = id_X,\  ff^{-1} = id_Y $
\end{definition}

The name groupoid already suggests that it is closely linked to group, but somehow a primitive version. We will introduce a rather important groupoid, and from it we shall see why it is closely linked to group.

\begin{example}
Given a manifold M, the {\rm\textbf{path groupoid}} $\mathcal{P}_1(M)$ is defined with:
\begin{enumerate}
    \item objects are points in manifold M, ${\rm Obj}(\mathcal{P}_1(M)) = M$
    \item morphisms between points $x,y\in M$ are ``thin homotopy classes" of (smooth) curves connecting $x$ and $y$. Namely, such curve can be written as $\gamma : [0,1]\to M, 0\mapsto x, 1 \mapsto y$, and the corresponding morphism is the class of $\gamma$ with equivalence up to a reparameterization, denoted $[\gamma]$.
    \item composition and inverse are trivially defined.
\end{enumerate}
\end{example}

The path groupoid is explicit when viewing the objects seriously as points in manifold, namely,

\begin{center}
\begin{tikzcd}
x & y \arrow[l, "\gamma_2"', bend right] & z \arrow[l, "\gamma_1"', bend right] \arrow[ll, "\gamma_3", bend left]
\end{tikzcd}
\end{center}

This groupoid describes the paths of a manifold, exactly what we need for gauge theory. But how is groupoid related to group? We may remember that it is not the homotopy classes of curves in the manifold that forms a group, but that starting and ending at a specific point! This gives us the clue to obtain a group from groupoid. The narration is rather elegant(we intend not to lose the elegance).
\begin{definition}
A {\rm\textbf{group}} G is a groupoid with only one object.
\end{definition}

One can check that this indeed align with our usual definition of group, with the morphisms being group elements and composition as multiplication. Another interesting fact is that this definition of group $G$ looks like the classifying space $BG$ (introduced in \ref{sec:three})  visually. We will see in the following that these two are actually equivalent.

The next thing we care about is how to describe gauge theory using $\mathcal{P}_1(M)$ and $G$. Obviously, we need to build some sort of map that translates paths of manifolds into group elements. Such ··homomorphism" between categories, that preserve the categorical structure is called a \textbf{functor}.

\begin{definition}
A {\rm\textbf{(covariant) functor}} between category $C$ and $D$ has the following structure:
\begin{enumerate}
    \item Maps of objects. $F:{\rm Obj}(C)\to {\rm Obj}(D), x\mapsto F(x)$.
    \item Maps of morphisms. For $f\in {\rm Hom}_C(x,y),F(f)\in {\rm Hom}_D(F(x),F(y))$.
    \item Preservation of identity. $F({\rm id}_x) = {\rm id}_{F(x)}$.
    \item Composition. $F(fg) = F(f)F(g)$.
\end{enumerate}
\end{definition}

We can see that functor just has everything we intuitively desire for a ``homomorphism" (or maps that preserves much properties of the category) between categories, and explains how much one system can internally be associated to another.

Now after so much math, at the end of the first chapter, we can finally spoil ourselves a little bit with some physics candy!

\begin{example}{(\textbf{Gauge Theory \cite{Baez_2010}})}
\label{Eg:GT}
{\rm  In gauge theory, holonomy can be seen as a functor:
\[ {\rm hol}: \mathcal{P}_1(M)\to G \]
Let's check what we have got. Maps of objects are trivial, since group $G$ is defined with only one object.Given a path $\gamma$ in the manifold M, we get
\[ {\rm hol}: \gamma \mapsto {\rm hol}(\gamma)\in G \]
Composition is preserved, so
\[ {\rm hol}(\gamma \delta) =  {\rm hol}(\gamma)  {\rm hol}(\delta)  \]
or in explicit terms,
\[ {\rm hol}({\gamma\delta}) = ({\rm hol}{\gamma})\cdot ({\rm hol}{\delta}) \]
And preservation of identity
\[ {\rm hol}({\rm id}_x) = 1_G  \]

All this information is neatly captured by saying ``{\rm hol} is a functor"\footnote{More mathematically rigorously speaking, ${\rm hol}$ should be an anafunctor, which means that the map of $\CP_1(M)$ to $G$ shall be mediated by another category, which in this case is the category with objects as points on open covers of the manifold $M$, morphisms are sewed paths passing through various patches. In the proof of Thm.\ref{Thm:GaugeTh}, we evade this weakness by stipulating that we consider a trivial $G$-bundle.}. One may also add other restraints, for example, considering homotopy classes of paths would give a gauge theory with flat connection.
Hereby we state the main theorem for this introductory section.

\begin{theorem}
\label{Thm:GaugeTh}
Given a manifold M and a Lie group G, there is a one to one correspondence between:
\begin{enumerate}
    \item connections on the trivial G-bundle over M
    \item smooth functors
    \[ {\rm hol}: \mathcal{P}_1(M)\to G \ ,  \]
     where $\mathcal{P}_1(M)$ is the path groupoid of M.
\end{enumerate}
\end{theorem}
\begin{proof}
{\rm 
Of course in our common sense, given a connection $A$( a $\mathfrak{g}$-valued 1-form on manifold $M$), we can construct
\[ {\rm hol}: \gamma \mapsto \mathcal{P}\exp(\int_\gamma A) \]
This is the form that we are familiar with, and look good enough unless we search for problems like path ordering and smoothness.

The other side follows a similar construction. Given a holonomy functor {\rm hol}, the $\mathfrak{g}$-valued 1-form $A$ is
\[ A(v) = \frac{\rm d}{{\rm d}s} {\rm hol}(\gamma_s)\big|_{s=0}  \]
where $v$ is any vector at point x,
\[ \gamma (0) = x\,, \quad \gamma'(0) = v\,. \]
}
\end{proof}

}\end{example}

This view of gauge theory aligns with our classifying space perspective in \ref{sec:three}.

\subsection{Towards Defining a 2-Group}
\label{sec:2-group}

The above showed how to generalize the symmetry from acting on local operators to non-local ones. Now we seek to expand our definition to a higher categorical setting, which is a deviation away from the group structure. Previously we have described paths on a manifold by assigning morphisms (or equivalently, group elements) to be paths, and now we want surfaces. An intuition would be to build a surface by filling two paths starting and ending at the same points, which is to build a ``morphism" between morphisms. This leads us to the following definition of 2-category.

Another perspective from category theory is to see the following relation. In sets, we can only claim if two elements are equal or not; in categories, we can say two different (not equal) objects are ``connected" or not, in some occasions (like groupoids), we can even say if two objects are isomorphic or not! With the philosophy that \textbf{every equality is up to some isomorphism} in mind, we naturally ask: why don't we extend the same courtesy to morphisms? Indeed, why not?

\begin{definition}
A {\rm\textbf{2-category}} $\mathcal{C}$ is constituted by the following data:
\begin{enumerate}
    \item Objects.
    \item Morphisms with ordinary compositions.
    \item 2-Morphisms, i.e. binary relations between morphisms. A 2-morphism can be $\alpha : {\rm Hom}_C(x,y)\to {\rm Hom}_C(x,y), \ f\mapsto g$.
    \item Identity. Identity 2-morphism for each morphism exists.
    \item Compositions. There are two kinds of compositions, horizontal one and vertical one, see the diagrams below.

    Horizontal composition (we denote by $\circ$):
\be
\begin{tikzcd}
      & {} \arrow[dd, "\alpha", Rightarrow] &                                                                      & {} \arrow[dd, "\beta", Rightarrow] &                                                                      &   &       & {} \arrow[dd, "\alpha\circ\beta", Rightarrow] &                                                                                     \\
\cdot &                                     & \cdot \arrow[ll, "f"', bend right=71] \arrow[ll, "f'", bend left=71] &                                    & \cdot \arrow[ll, "g"', bend right=71] \arrow[ll, "g'", bend left=71] & = & \cdot &                                               & \cdot \arrow[ll, "f\cdot g"', bend right=71] \arrow[ll, "f'\cdot g'", bend left=71] \\
      & {}                                  &                                                                      & {}                                 &                                                                      &   &       & {}                                            &                                                                                    
\end{tikzcd}
\ee

Vertical composition (we denote by $\cdot$):
\be
    \centering
\begin{tikzcd}
      & {} \arrow[d, "\alpha", Rightarrow] &                                                                                     &   &       & {} \arrow[dd, "\beta\cdot \alpha", Rightarrow] &                                                                                \\
\cdot & {} \arrow[d, "\beta", Rightarrow]  & \cdot \arrow[ll, "f"', bend right=71] \arrow[ll, "h", bend left=71] \arrow[ll, "g"] & = & \cdot &                                                & \cdot \arrow[ll, "f" description, bend right=67] \arrow[ll, "h", bend left=67] \\
      & {}                                 &                                                                                     &   &       & {}                                             &                                                                               
\end{tikzcd}
\ee
We demand that both compositions satisfy the associativity law.
    \item Compatibility of two 2-morphism compositions shall be compatible, meaning for
    \be
\begin{tikzcd}
      & {} \arrow[d, "\alpha_1", Rightarrow] &                                                                     & {} \arrow[d, "\alpha_2", Rightarrow] &                                                                     \\
\cdot & {} \arrow[d, "\beta_1", Rightarrow]  & \cdot \arrow[ll, bend left=67] \arrow[ll, bend right=67] \arrow[ll] & {} \arrow[d, "\beta_2", Rightarrow]  & \cdot \arrow[ll, bend right=67] \arrow[ll, bend left=67] \arrow[ll] \\
      & {}                                   &                                                                     & {}                                   &                                                                    
\end{tikzcd}
    \ee
    two ways to composite should be strictly equal, i.e., 
    \be (\beta_1\circ \beta_2)\cdot(\alpha_1 \circ \alpha_2) = (\beta_1\cdot \alpha_1)\circ (\beta_2\cdot \alpha_2) \ee
\end{enumerate}

\end{definition}

And likewise, we can build 3-morphisms which connects 2-morphisms, $n$-morphisms which connects $(n-1)$-morphisms, the ladder goes on\footnote{How about an infinite one?}. Like what we did with 1-categories, we would like to get an inverse or isomorphic structure between morphisms, more formally, 2-groupoid.

\begin{definition}
A {\rm\textbf{2-groupoid}} is a 2-category with:
\begin{enumerate}
    \item 1-Morphisms invertible (like in groupoid)
    \item 2-Morphisms invertible (in both horizontal and vertical ways)
\end{enumerate}
The invertibility for 2-morphisms is explicitly written as:

And written as($\alpha: {\rm Hom}_C(x,y)\to {\rm Hom}_C(x,y), f\mapsto g$)
\[ \alpha\circ\alpha_h^{-1} = {\rm id}_{{\rm id}_y};\quad \alpha_h^{-1}\circ\alpha = {\rm id}_{{\rm id}_x}\]
\[ \alpha\cdot\alpha_v^{-1} = {\rm id}_g; \quad \alpha_v^{-1}\cdot \alpha = {\rm id}_f \]
\end{definition}

So, likewise we can define 2-group in an elegant and natural way.

\begin{definition}
A {\rm \textbf{(strict) 2-group}} $\mathcal{G}$ is a 2-groupoid with one object.
\end{definition}

Also, we can build this structure up to n-groupoid and n-group. And we can interpret this structure physically. For normal gauge theory, we describe the \textbf{worldline} a point particle moving in spacetime manifold (base manifold) by assigning it with an group element. Now we attempt to do similar thing with \textbf{worldsheet} of a string.

The structure of a strict 2-group can be equivalently described by a more explicit structure: crossed module.

\begin{definition}
    A {\rm \textbf{crossed module}} is defined by a tuple $(G,H,\partial ,\rhd)$, where
    \begin{enumerate}
        \item $G, \ H$ are groups.
        \item $\ptl: H\to G$ is a group homomorphism, an equivalent expression is $\ptl\in{\rm Hom}_{\rm Grp}(H,G)$
        \item $\rhd : G\to {\rm Aut}(H)$.
        \item Constraint that $\ptl(g\rhd h) = g(\ptl h) g^{-1}$ and $\ptl h\rhd h' = hh'h^{-1}$.
    \end{enumerate}
\end{definition}

Now we explain how all these sums up to become equivalent to a strict 2-group.

\begin{enumerate}
    \item The group $G$ corresponds to the 1-morphisms,

    \begin{center}
    \begin{tikzcd}
    \cdot & \cdot \arrow[l, "g_1"] & \cdot \arrow[l, "g_2"] & = & \cdot & \cdot \arrow[l, "g_1g_2"]
    \end{tikzcd}
    \end{center}
    
    \item The group $H$ corresponds to the 2-morphisms whose source is ${\rm id}_\cdot$, and the result of this action is $\ptl(h)\in G$,
    \begin{center}\begin{tikzcd}
      & {} \arrow[dd, "h", Rightarrow] &                                                                                           \\
\cdot &                                & \cdot \arrow[ll, "\partial h", bend left=67] \arrow[ll, "{\rm id}_\cdot"', bend right=67] \\
      & {}                             &                                                                                          
    \end{tikzcd}\end{center}
    This also give the horizontal composition 
\begin{center}\begin{tikzcd}
      & {} \arrow[dd, "h", Rightarrow] &                                                                                           & {} \arrow[dd, "h'", Rightarrow] &                                                                                            &   &       & {} \arrow[dd, "h\circ h'", Rightarrow] &                                                                                            \\
\cdot &                                & \cdot \arrow[ll, "\partial h", bend left=67] \arrow[ll, "{\rm id}_\cdot"', bend right=67] &                                 & \cdot \arrow[ll, "\partial h'", bend left=67] \arrow[ll, "{\rm id}_\cdot"', bend right=67] & = & \cdot &                                        & \cdot \arrow[ll, "\partial(hh')", bend left=67] \arrow[ll, "\rm id_\cdot"', bend right=67] \\
      & {}                             &                                                                                           & {}                              &                                                                                            &   &       & {}                                     &                                                                                           
\end{tikzcd}\end{center}

\item  Action of $G$ on $H$ by $\rhd$ is categorically explained as
\begin{center}
\begin{tikzcd}
      & {} \arrow[dd, "{\rm id}_{g}", Rightarrow] &                                                                     & {} \arrow[dd, "h", Rightarrow] &                                                                                           & {} \arrow[dd, "{\rm id}_{g^{-1}}", Rightarrow] &                                                                               &   &       & {} \arrow[dd, "g\rhd h", Rightarrow] &                                                                                                                      \\
\cdot &                                           & \cdot \arrow[ll, "g"', bend right=67] \arrow[ll, "g", bend left=67] &                                & \cdot \arrow[ll, "\partial h", bend left=67] \arrow[ll, "{\rm id}_\cdot"', bend right=67] &                                                & \cdot \arrow[ll, "g^{-1}", bend left=67] \arrow[ll, "g^{-1}"', bend right=67] & = & \cdot &                                      & \cdot \arrow[ll, "\partial(g\rhd h) = g(\partial h)g^{-1}", bend left=67] \arrow[ll, "\rm id_\cdot"', bend right=67] \\
      & {}                                        &                                                                     & {}                             &                                                                                           & {}                                             &                                                                               &   &       & {}                                   &                                                                                                                     
\end{tikzcd}
\end{center}
The upper row composes into ${\rm id}_\cdot$ still, but the bottom line is changed, this what a $G$ action does to 2-morphisms, and it give us an explanation to one of the constraints we imposed  earlier.

\item (Pieffer identity)We will obtain the second constraint in the definition of 2-group using the interchange law.

\begin{center}
    \begin{tikzcd}
        h h' h^{-1} & = & \cdot & & \cdot  \arrow[ll, bend right=75,"{\rm id}_{\cdot}"' ,""{name = L1}] \arrow[ll, bend left=75,"\partial(h)" ,""{name = L2} ] &  & \cdot  \arrow[ll, bend right=75,"{\rm id}_{\cdot}"' ,""{name = M1}]  \arrow[ll, bend left=75,"\partial (h')" ,""{name = M2}] 
        & & \cdot  \arrow[ll, bend right=75,"{\rm id}_{\cdot}"' ,""{name=R1}]  \arrow[ll, bend left=75,"\partial (h)^{-1}" ,""{name=R2}]
        \arrow[Rightarrow, from=R1, to=R2,"h^{-1}"]
        \arrow[Rightarrow, from=M1, to=M2,"h'"]
        \arrow[Rightarrow, from=L1, to=L2,"h"]
        \\
        \\
        \\
        \\
         & = & \cdot & &  \cdot  \arrow[ll, bend right=75,"{\rm id}_{\cdot}"' ,""{name = L1}] \arrow[ll, bend left=75,"\partial(h)" ,""{name = L2} ]
         \arrow[ll, ""{name = L3} ]& & \cdot  \arrow[ll, bend right=75,"{\rm id}_{\cdot}"' ,""{name = M1}]  \arrow[ll, bend left=75,"\partial (h')" ,""{name = M2}] \arrow[ll, ""{name = M3} ]
        & & \cdot  \arrow[ll, bend right=75,"{\rm id}_{\cdot}"' ,""{name=R1}]  \arrow[ll, bend left=75 ,""{name=R2},"\ptl(h^{-1})"]\arrow[ll, ""{name = R3} ]
        \arrow[Rightarrow, from=L1, to=L3, "h"]
        \arrow[Rightarrow, from=M1, to=M3, "{\rm id}_{\rm id_\cdot}"]
        \arrow[Rightarrow, from=R1, to=R3, "h^{-1}"]
        \arrow[Rightarrow, from=R3, to=R2, "{\rm id}_{\ptl h^{-1}}"]
        \arrow[Rightarrow, from=M3, to=M2, "h'"]
        \arrow[Rightarrow, from=L3, to=L2 , "{\rm id}_{\ptl h}"]
        \\ \\ \\ \\
        & = & \cdot & & \cdot \arrow[ll, bend right = 75, "{\rm id}_\cdot" , ""{name = 1}]\arrow[ll, bend right = 0, "{\rm id}_\cdot" , ""{name = 2}]\arrow[ll, bend left = 75, "\ptl (hh'h^{-1})" , ""{name = 3}]
        \arrow[Rightarrow, from = 1, to = 2, "{\rm id}_{\rm id_\cdot}"]
        \arrow[Rightarrow, from = 2, to = 3, "\ptl h \rhd h'"] & = & \ptl h \rhd h'

    \end{tikzcd}
\end{center}
    
\end{enumerate}

It can be proved that there is a one to one correspondence between strict 2-groups and crossed modules. Now we attempt to extract more information from the crossed module.

\subsection{Crossed Extension}
\label{sec:cross-extension}

In this section we seek to derive a certain equivalence relation within the crossed modules, and derive the ``weak" 2-group\footnote{Most literature introduces weak 2-group first because the weak 2-groups carry physical information, but in our formalism, strict 2-groups are more natural to define}.

Notice that the existence of group homomorphism $\ptl$ in a crossed module $(G,H,\ptl , \rhd)$ can lead to an exact sequence
\be 1\to {\rm ker}\ptl \stack{i}{\longrightarrow} H \stack{\ptl}{\longrightarrow} G  \stack{p}{\longrightarrow} {\rm coker}\ptl \to 1 \ee
We denote $\Pi_2 \equiv {\rm ker}\ptl, \  \Pi_1\equiv {\rm coker}\ptl = G/ {\rm im}\ptl  $, which is legitimate since  $g(\ptl h)g^{-1} = \ptl(g\rhd h)\in {\rm im}\ptl$(${\rm im}\ptl$ is normal subgroup). This exact sequence is called the \textbf{crossed extension}.

Now we can define an equivalence relation as follows, two crossed modules $(G_1,H_1,\ptl_1 , \rhd_2)$ and $(G_2,H_2,\ptl_2 , \rhd_2)$ are weakly equivalent if the diagram
\be
\label{2-group-equiv}
\begin{tikzcd}
	&& {H_1} & {G_1} \\
	1 & {\Pi_2} & {} && {\Pi_1} & 1 \\
	&& {H_2} & {G_2}
	\arrow[from=2-1, to=2-2]
	\arrow["{i_1}", from=2-2, to=1-3]
	\arrow["{i_2}"', from=2-2, to=3-3]
	\arrow["{\partial_1}",from=1-3, to=1-4]
	\arrow["{\partial_2}",from=3-3, to=3-4]
	\arrow["{t_H}"', from=1-3, to=3-3]
	\arrow["{t_G}"', from=1-4, to=3-4]
	\arrow["p_1",from=1-4, to=2-5]
	\arrow["p_2",from=3-4, to=2-5]
	\arrow[from=2-5, to=2-6]
\end{tikzcd}
\ee
is commutative, and is compatible with the two actions $\rhd_1, \ \rhd_2 $, meaning
\be t_H(g\rhd h) = t_G(g)\rhd' t_H(h) , \ \forall g\in G_1,\ H\in H_1  \ee

This equivalence class can be nicely classified by $(\Pi_1,\Pi_2,\alpha,\beta)$, where $\Pi_{1,2}$ are previously mentioned,  $\alpha$ is the action $\alpha : \Pi_1\to {\rm Aut}(\Pi_2)$ induced by $\rhd$, and most notably, $\beta\in H^3_{\rm grp}(\Pi_1,\Pi_2)$ is a group-cohomological information named the \textbf{Postnikov class}. 

\subsection{Postnikov Class}
\label{sec:Postnikov}

How do we understand and compute the Postnikov class information? Notice we have previously mentioned group cohomology in Section \ref{subsubsec:0+1D}. There the group cohomology serves as an element to measure the breaking of associavity in the projective representation. Here, as we will soon see, the Postnikov class will play a similar role classifying the ``twist" of group action in the crossed module.

Given a crossed extension
\be 1\to \Pi_2 \stack{i}{\longrightarrow} H \stack{\ptl}{\longrightarrow} G  \stack{p}{\longrightarrow} \Pi_1 \to 1\,, \ee
we compute the Postnikov class in the following procedures:
\begin{enumerate}
    \item Choose a lifting map $s: \Pi_1\to G$ st. $p\circ s = {\rm id}_{\Pi_1}$. Note that it does not have to be a homomorphism, for we are about to define $f: \Pi_1\times \Pi_1 \to G$ as
    \be  s(g)s(h) = s(g,h)f(g,h) \ . \ee
    Also, the function $f(g,h)$ satisfies the 3-cocycle condition
    \be
    s(g)f(h,k)s(g)^{-1}f(g,hk)=f(g,h)f(gh,k)\,.
    \ee
    \item Now we uplift $f$ to a function $F:\Pi_1\times\Pi_1\rightarrow H$ by requiring that $\ptl (F(g))\equiv f(g)$ for all $g\in\Pi_1$. 
    \item The breaking of cocycle condition for $F(g,h)$ can be explicitly written as
    \be
    \label{cocycle-F}
    (s(g)\rhd F(h,k))F(g,hk)=i(c(g,h,k))F(g,h)F(gh,k)\,,
    \ee
    where $c:\Pi_1\times\Pi_1\times\Pi_1\rightarrow\Pi_2\in H^3(B\Pi_1,\Pi_2) $ is the Postnikov we so desire.
\end{enumerate}

\subsubsection{Example: $G = H = \bbZ_4$}
This is one of the most simple non-trivial example, the crossed extension is explicitly written as
\be 1\to \bbZ_2 \stack{\times 2}{\longrightarrow} \bbZ_4 \stack{\times 2}{\longrightarrow} \bbZ_4  \stack{\text{mod}\ 2}{\longrightarrow} \bbZ_2 \to 1 \ee
Note that we use $\{0,1,\dots,N-1\}$ to denote the elements in the additive group $\mb{Z}_N$.

Now since ${\rm Aut}(\bbZ_4) \cong \bbZ_2$, where the non-trivial element $\alpha$ stands for the action $(1,3)$(permuting 1 and 3 while leaving other elements unchanged). The group cohomology here is \be H^3_{\rm grp} (B\bbZ_2,\bbZ_2)=\bbZ_2 \ee 
If we choose the action $\rhd = {\rm id}\in {\rm Aut}(\bbZ_4)$, then the Postnikov class is $0\in H^3_{\rm grp} (B\bbZ_2,\bbZ_2)=\bbZ_2$; If we choose the non-trivial action $\rhd = \alpha $, then the Postnikov class is the non-trivial element $1\in H^3_{\rm grp} (B\bbZ_2,\bbZ_2)$.

\subsubsection{Example: $G = SU(2), \ H = \bbZ_{2l}$}
We shall observe a much richer structure when we turn our attention on Lie groups. This time the crossed extension is
\be 1\to \bbZ_{l} \stack{\times 2}{\longrightarrow} \bbZ_{2l} \stack{ {\rm mod} 2 \cdot I }{\longrightarrow} SU(2)  \stack{p}{\longrightarrow} SO(3) \to 1 \ee
We want to break this long and complicated sequence up to some short sequence that we can under stand, and we can do it!
\be
    \begin{tikzcd}
	{1}\arrow{r} & {\mb{Z}_{l}}\arrow{r}{i}\arrow{d} & {\mb{Z}_{2l}}\arrow{r}\arrow{d}{\ptl = {\rm mod} 2 \cdot I} & {\mb{Z}_2 = {\rm im}\ptl}\arrow{r}\arrow{d} & {1} \\
	{1}\arrow{r} & {\mb{Z}_2 ={\rm im}\ptl}\arrow{r} & {SU(2)}\arrow{r}{p} & {SO(3)}\arrow{r} & {1}
\end{tikzcd}
\ee
For any crossed extension, this procedure is well-defined. The important information is:
\begin{itemize}
    \item The first row
    \be\begin{tikzcd}
    \label{eq:1strow}
        {1}\arrow{r} & {\mb{Z}_{l}}\arrow{r}{i} & {\mb{Z}_{2l}}\arrow{r} & {\mb{Z}_2 = {\rm im}\ptl} \arrow{r} & {1}
    \end{tikzcd} \ee
    is a central extension of ${\rm im}\ptl$ by $\Pi_2$, and is classified by $\omega_1 \in H^2_{\rm grp}( \bbZ_2, \bbZ_{2l} )$.
    \item Also, the second row
    \be\begin{tikzcd}
        {1}\arrow{r} & {\mb{Z}_2 ={\rm im}\ptl}\arrow{r} & {SU(2)}\arrow{r}{p} & {SO(3)}\arrow{r} & {1}
    \end{tikzcd}\ee
    is also a central extension, classified by $\omega_2 \in H^2_{\rm grp}( \Pi_1, {\rm im}\ptl ) = H^2_{\rm grp}(SO(3),\bbZ_2)$, which describes the obstruction of lifting a $SO(3)$-bunlde into a $SU(2)$-bundle, which is exactly the second Stiefel-Whitney class. The bundle-lifting interpretation also stands for other continuous $\Pi_1$ and $G$.
    \item What about the Postnikov class we talked about? What is the relation between the 3rd order Postnikov class and the 2nd order Stiefel-Whitney class? This have to do with the first row. First we introduce a famous tool in homology theory.
    \begin{theorem}[Bockstein]
        Given an exact sequence of Abelian groups
        \be 0\to P\to Q\to R\to 0  \ee
        there is a homomorphism of cohomology group $B_i : H^i(C,R)\to H^{i+1}(C,P)$, and thus a long exact sequence
        \be  \cdots \to H^{n}(C,P) \to H^{n}(C,Q) \to H^{n}(C,R) \to
 H^{n+1}(C,P)  \to H^{n+1}(C,Q)\to \cdots  \ee
    \end{theorem}
    With this theorem, it is explicit that since Eq.~(\ref{eq:1strow}) is an exact sequence, we can induce a Bockstein homomorphism
    \be  B: H^2_{\rm grp}(\Pi_1,{\rm im}\ptl) \to H^3_{\rm grp}(\Pi_1,\Pi_2) , \ \ \omega_2\mapsto \beta \ee
    which is exactly the Postnikov class.
\end{itemize}

\subsection{Physical Interpretation}
\label{sec:2-group-phys}

In this section we will see how 2-group describes the higher-symmetry structure in physical system. We will derive how $\Pi_1$ describes 0-form symmetry and how $\Pi_2$ describes 1-form symmetry, and $\alpha,\beta$ describes how these two incorporate.

Now we use the topological defect network perspective to see the above. The $\alpha$ is understood as the 0-form defect ($\Pi_1$-defect) twisting the intersecting 1-form defect ($\Pi_2$-defect, one dimensional lower),
\begin{figure}[htbp]
\centering
\begin{tikzpicture}
    \draw [yslant = -0.5, fill = cyan] (0,0) rectangle (1,2) (1,2.5)node[above , cyan] {$g\in\Pi_1 $};
    \draw[densely dotted, thick, orange] (-1,0.8) node[anchor = south] {$\alpha_{g}(h)$} -- (0.5,0.8)  ;
    \draw[ultra thick, orange] (0.5,0.8) -- (1.7,0.8) node[anchor = south]{$h\in \Pi_2$};
\end{tikzpicture}
\end{figure}
The 1-form symmetry element changes $h\mapsto \alpha_g(h)=g\rhd h$ after crossing the generator $U_g$ of 0-form symmetry.

The information of $\beta$ is encoded in the fusion of 1-form defects.
\begin{figure}[htbp]
    \centering
    \begin{tikzpicture}[scale = 0.5]
    \draw[thick] (0,0) -- (2,-4) (0,0) node[above] {$g$};
    \draw[shift = {(8,0)} ,thick] (0,0) -- (2,-4) (0,0) node[above] {$g$};
    \draw[thick] (4,0) -- (2,-4) (4,0) node[above]{$k$}; 
    \draw[shift = {(8,0)}, thick] (4,0) -- (2,-4) (4,0) node[above]{$k$}; 
    \draw[thick] (2,-4) -- (2,-6) (2,-6) node[below]{$ghk$};
    \draw[shift = {(8,0)}, thick] (2,-4) -- (2,-6) (2,-6) node[below]{$ghk$};
    \draw[thick] (2,0) -- (1,-2) (2,0) node[above]{$h$};
    \draw[thick] (10,0) -- (11,-2) (10,0) node[above]{$h$} (6,-3) node[above]{$= \ \ \ \ \beta(g,h,k)$} ;
    
    \end{tikzpicture}
\end{figure}

The above procedure from the left to the right is sometimes called F-move in accordance with the fusion rule in CFT. After a F-move, a new generator for 1-form symmetry $\beta(g,h,k)$ emerges.

The weak 2-group background gauge fields are $A$ for the 0-form symmetry and $B$ for the 1-form symmetry. In the simplicial set (flat connection) language,
\begin{enumerate}
    \item For the 0-form part, the transition functions $A_{ij}\in C^1(M,\Pi_1)$, and the 1-cocycle condition would suggest for each 2-simplex $\{ ijk\}$, we have
    \be A_{ij}A_{jk}A_{ki} = {\rm id}_{\Pi_1}  \ee
    
    \item For the 1-form part, the gauge field $B\in C^2(M,\Pi_2)$, satisfies a twisted-shifted cocycle condition
    \be\ba
    \delta_A B &= A^* \beta \\
    \Leftrightarrow \alpha_{A_{ij}} (B_{jkl}) B_{ikl}^{-1} B_{ijl} B_{ijk}^{-1} &= \beta(A_{ij},A_{jk},A_{kl})
    \ea\ee
    this is twisted because of the derivative $\delta_A$ has an additional $\Pi_1$-action than the usual \Cech derivative $\delta$; shifted because it does not equal 0, instead, it matches the Postnikov invariant.
\end{enumerate}

The gauge transformations shall preserve the flatness condition, there are two types of gauge transformations,
\begin{enumerate}
    \item the first one is
    \be A \mapsto A , \ \ B\mapsto B + \delta_A \Lambda  \ee
    $A$ field remains intact, and $B$ shifted by a twisted-exact term.
    \item the second one is more complicated
    \be\ba 
    A_{ij} &\mapsto A^f_{ij} = f_i A_{ij} f_j^{-1} \\   
    B&\mapsto \alpha_f(B) + \xi(A,f) , \\ 
    \xi(A,f) \in C^2(M,\Pi_2) \ \ &{\rm and} \ \delta_{A^f}\xi(A,f) = A^{f*} B-\alpha_f (A^*\beta) \ .
    \ea\ee
\end{enumerate}

One of the naturally occurring 2-group structures in condensed matter systems is the symmetry-enriched gauge theories, which can be viewed as special cases of symmetry-enriched topological (SET) states~\cite{etingof2010fusion,PhysRevB.87.155115,PhysRevB.87.165107,Chang_2015,PhysRevB.93.155121,PhysRevB.94.235136,PhysRevB.96.115107,PhysRevB.100.115147,10.21468/SciPostPhys.8.2.028,PhysRevB.106.115104}.

For example, let's consider a (2+1)D system with an Abelian gauge group $A$ and a symmetry group $\Pi_1$. The anyons or excitations of the $A$ gauge theory are described by the quantum double of the gauge group $A$. In the Abelian setting, this corresponds to $A\times\hat{A}$, where $A$ and $\hat{A}=\text{Hom}(A,U(1))$ describe the gauge flux and gauge charge in the $A$ gauge theory, respectively. The world lines of such anyons can be understood as the TDL of the 1-form symmetry $\Pi_2=A\times\hat{A}$. On the other hand, the 0-form symmetry $\Pi_1$ acts on the anyons precisely through the action $\alpha$.

All the data and consistency conditions for the 2-group theory can be traced back to the classification data and obstructions for SET states~\cite{etingof2010fusion,PhysRevB.100.115147,wang2021domain}. More generally, if we consider non-Abelian $A$, then the quantum double $\Pi_2=DA$ as 1-form symmetry will be non-invertible. Consequently, the SET states will correspond to systems with non-invertible 1-form symmetry and invertible 0-form symmetry. The interaction of these two types of symmetries will make the classification more involved~\cite{etingof2010fusion,PhysRevB.100.115147}. The lattice model for a 2-gauge theory can be expressed as a discrete nonlinear $\sigma$-model. This representation arises from the triangulation of spacetime into a classifying space associated with a 2-group, as discussed in detail in Ref.~\citeonline{PhysRevB.100.045105}.

\subsection{Example: Gauging with 't Hooft Anomaly }
\label{sec:2-group-gauging}

Here we present a physical application of the weak 2-group symmetry.

Consider a 3d QFT with 0-form global symmetry $\hat{\Pi}_2\times \Pi_1$, where for now we just take $\hat{\Pi}_2$ to be an Abelian group. This system has a mixed 't Hooft anomaly described by
\be S = 2\pi \int_{M_4} C\cup A^*\beta  \ee
where $C$ is the background 1-form gauge field of $\hat{\Pi}_2$, $A$ is the background 1-form gauge field of $\Pi_1$ and $\beta$ the Postnikov class. With the presence of the 't Hooft anomaly, we cannot gauge the entire $\hat{\Pi}_2\times \Pi_1$ symmetry, but after gauging $\hat{\Pi}_2$, we can obtain a new 1-form symmetry $\Pi_2\cong {\rm Hom}(\hat{\Pi}_2,U(1))$ and the corresponding 2-form background gauge field B.

The new topological action after the procedure is 
\be S = 2\pi \int_{M_3} C\cup B  + 2\pi \int_{M_4} C\cup A^*\beta \ee
which should be gauge invariant under the $\Pi_2$ gauge transformation, thus requiring
\be  \delta B + A^* \beta = 0 \ee
which aligns with our previous definition of a weak 2-group (up to a sign notation, if you watch carefully). So, in this theory, $\Pi_1$, $\Pi_2$ constitutes a 2-group symmetry $(\Pi_1,\Pi_2,1,\beta)$ where $1$ denotes trivial action.

\subsection{Generalization: Weak $n$-Group}
\label{sec:n-group}
As we previously said, the condition of ``equal" can always be loosen to ``equivalence", and the language chosen to describe such equivalence relation varies.

We can definitely imagine the sketch of a weak $n$-group: there are $n$-morphisms filling the gap of two $(n-1)$-morphisms with the same source and target,... and the ladder goes down until 1-morphism. But here what we would like to describe is a relatively simple case, which is a mere generalization of the previous example.

Remind that at Prop.\ref{prop:gauge}, when we gauge a $p$-form symmetry, we shall have a $(d-p-2)$-form symmetry. Now consider gauging a 0-form symmetry $\hat{\Pi}_2$ in the complete $\hat{\Pi}_2\times \Pi_1$ in $(n+1)$-dimensional spacetime with 't Hooft anomaly
\be  S = 2\pi \int_{M_{n+2}} C\cup A^*\beta   \ee
where $\beta\in H^{n+1}(\Pi_1,\Pi_2)$ is the Postnikov class. This special version of weak n-group have $\Pi_1$ the 0-form symmetry describing the 1-morphisms, $\Pi_2$ the $(n-1)-$form symmetry describing the n-morphisms.

\subsection{A Touch of Non-Invertible Symmetries}
\label{sec:non-invertible}
In this lecture, the generalized symmetries we encountered are all invertible symmetries, which can be described with groups or higher groups. A broader class of symmetry is the non-invertible symmetries, including the more general higher categorical symmetries. For such symmetries, the topological generators would obey the general fusion rule
\be
U_g(M)U_h(M)=\sum_i U_{g_i}(M_i)
\ee
instead of the group associativity law. Note that the r.h.s. can even consist of topological operators defined on manifolds $M_i$ with different dimensions. The realm of non-invertible symmetries is fastly evolving, and we will point the interested readers to the recent lecture note~\cite{Schafer-Nameki:2023jdn}.

From the perspective of condensed matter physics, the creation operator of non-Abelian anyons within a 2+1D topological order can be interpreted as the generators of a non-invertible symmetry. This symmetry has the ability to operate on the 1+1D boundary, imposing non-invertible symmetry constraints on potential boundary theories~\cite{PhysRevResearch.2.043086,kong_mathematical_2020,kong_mathematical_2021,PhysRevResearch.2.033417,ji2022unified,PhysRevB.107.155136,PhysRevB.108.075105}.

We give one example of the $Spin(4N)$ Yang-Mills theory in $d$-dimensions~\cite{Bhardwaj:2022yxj}. The theory has a $\mb{Z}_2^{(1)}\times\mb{Z}_2^{(1)}$ 1-form center symmetry, whose background 2-form gauge fields are denoted by $B_2$ and $C_2$. There is also an $\mb{Z}_2^{(0)}$ outer automorphism of $D_{2N}$ roots, whose background gauge field is $A_1$. The $\mb{Z}_2^{(0)}$ symmetry acts on the Dynkin diagram of $D_{2N}$ as

\begin{center}
\begin{tikzpicture}
    \filldraw[] (0,0) circle (0.15) (1,0) circle (0.15) (3,0) circle (0.15) (3.85,0.5) circle (0.15) (3.85,-0.5) circle (0.15);
    \draw[thick] (0,0)--(1,0) (2.5,0)--(3,0) (3,0)--(3.85,0.5) (3,0)--(3.85,-0.5);
    \draw[thick] (1,0)--(1.5,0) node[anchor = west]{\textbf{$ \ \cdots$}};
    \draw[thick,->] (4.1,-0.35) arc (-20:20:1);
    \draw[thick,->] (4.1,0.35) arc (20:-20:1);
    \node[] at (4.7,0){$\bbZ_2^{(0)}$};
\end{tikzpicture}
\end{center}
which maps the representation $\mbf{R}$ of $D_{2N}$ into the conjugate representation $\overline{\mbf{R}}$. Now after gauging the $\mb{Z}^{(1)}$ symmetry with $B_2$ background gauge field, we obtain the $SO(4N)$ Yang-Mills theory, with a mixed 't Hooft anomaly 
\be
\mc{A}=\pi\int_{M_{d+1}} A_1\cup C_2\cup \hat{B}_{d-2}
\ee
Here $\hat{B}_{d-2}$ is the background gauge field for the $(d-3)$-form dual symmetry.

Finally, if we further gauge $A_1$ and $\hat{B}_{d-2}$, we will get the $Pin^+(4N)$ theory with non-invertible symmetry defects.

In general, if a non-invertible symmetry arises from such kind of partial gauging, it is called a \textbf{non-intrinsically non-invertible symmetry}. Otherwise, it is an \textbf{intrinsically non-invertible symmetry}.

\section*{Acknowledgements}

We thank Lakshya Bhardwaj, Chi-Ming Chang, Max Hubner, Ruizhi Liu, Sakura Schafer-Nameki, Fajie Wang, Jingxiang Wu and Yi Zhang for discussions on the related subjects. We would also like to thank the audiences at Peking University for the discussions and questions. QRW is supported by
National Natural Science Foundation of China under Grant No. 12274250. YNW is  supported by National Natural Science Foundation of China under Grant No. 12175004, by Peking University under startup Grant No. 7100603667, and by Young Elite Scientists Sponsorship Program by CAST (2022QNRC001).

\newpage
\newpage

\appendix

\section{Čech Cohomology}
In this appendix we intend to explain how to formulate a method for properly assigning group elements to links in manifolds as described in the perspective of flat connection(\ref{para:flatconn})

In the construction of manifolds, we describe a manifold $M$ through the union of locally $\bbR^n$ patches, and glue them together with a set of transition function $f_{ij} : \phi_i(U_i\cap U_j)\to \phi_j(U_i\cap U_j) $ as we all very well know. In the following, we choose a good cover, within which all the intersections $U_{ij}$ are simply connected and isomorphic to $\bbR^n$.

The \Cech (co)homology is basically a way to define a (co)homology and takes idea from intersection of patches. The construction is as follows:
\begin{enumerate}
    \item Each $U_i$ represents a 0-simplex;
    \item Each $U_{ij} \equiv U_i\cap U_j$ represents a 1-simplex;
    \item Each $k$-simplex is represented by 
    \be 
    U_{i_0,\cdots , i_k} \equiv \bigcap_{m = 0}^k U_{i_m} \ ,
    \ee
    graphically drawn as
    \begin{center}
        \begin{tikzpicture}
            \filldraw[ cyan!50,fill opacity=0.5 ] (0,0) circle (1.5) ; 
            \filldraw[] (0,0) circle (0.08);
            \filldraw[ red!50,fill opacity=0.5 ] (1.9,0) circle (1.5);
            \filldraw[] (1.9,0) circle (0.08);
            \filldraw[ orange!50,fill opacity=0.5 ] (0.95,-1.64) circle (1.5);
            \filldraw[] (0.95,-1.64) circle (0.08);

            \draw[thick]  (0,0)--(1.9,0)  ;
            \draw[thick]  (0.95,-1.64)--(1.9,0)   (1.425,-0.82)node[right]{$U_{jk}$};
            \draw[thick]  (0,0)--(0.95,-1.64)  (0.475, -0.82)node[left]{$U_{ik}$};
            \node[above] at (0,0) {$U_i$};
            \node[above] at (01.9,0) {$U_j$};
            \node[below] at (0.95,-1.64) {$U_k$};
            \node[above] at (0.95,0) {$U_{ij}$};
            \node[] at (0.95, -0.55) {$U_{ijk}$};
        \end{tikzpicture}
    \end{center}
    \item The boundary operator is defined to be\footnote{A mathematician may wonder how the addition structure comes up, we claim that we are now working in an Abelian group freely generated by $U_{i_1,\cdots,i_k}$s with coefficient in $\bbZ$ (other Abelian groups are also okay).} 
    \be
    \ptl U_{i_0,\cdots ,i_k} = \sum_{m=0}^k (-1)^m U_{i_0,\cdots , \Hat{i}_m ,\cdots , i_k }
    \ee
    For example,
    \be  \ptl U_{ijk} = U_{jk} -U_{ik} + U_{ij}  \ee
    graphically,
    \begin{center}
        \begin{tikzpicture}
            \filldraw[] (0,0) circle (0.08);
            \filldraw[] (1.9,0) circle (0.08);
            \filldraw[] (0.95,-1.64) circle (0.08);

            \draw[thick] (0,0)--(1.9,0);
            \draw[thick,rotate = -60] (0,0)--(1.9,0);
            \draw[thick,rotate around = {60:(1.9,0)}] (1.9,0)--(0,0);
            \node[above] at (0,0) {$U_i$};
            \node[above] at (01.9,0) {$U_j$};
            \node[below] at (0.95,-1.64) {$U_k$};
            \node[] at (0.95, -0.55) {$U_{ijk}$};
            \node[] at (-0.2, -0.82){$\ptl$};
            \node[right] at (1.93, -0.82){$=$};

            \draw[shift = {(1.5,0)},thick,rotate around = {60:(1.9,0)}] (1.9,0)--(0,0);
            \draw[shift = {(4,0)},thick,rotate = -60] (0,0)--(1.9,0);
            \draw[shift = {(6,0)},thick] (0,0)--(1.9,0);
            \node[shift = {(4,0)},above] at (0,0) {$U_i$};
            \node[shift = {(6,0)},above] at (0,0) {$U_i$};

            \node[shift = {(1.5,0)},above] at (01.9,0) {$U_j$};
            \node[shift = {(6,0)},above] at (01.9,0) {$U_j$};

            \node[shift = {(4,0)},below] at (0.95,-1.64) {$U_k$};
            \node[shift = {(1.5,0)},below] at (0.95,-1.64) {$U_k$};
            \node[] at (3.7,-0.82) {$-$};
            \node[] at (5.7,-0.82) {$+$};

            \filldraw[shift = {(4,0)}] (0,0) circle (0.08);
            \filldraw[shift = {(1.5,0)}] (1.9,0) circle (0.08);
            \filldraw[shift = {(1.5,0)}] (0.95,-1.64) circle (0.08);
            \filldraw[shift = {(6,0)}] (0,0) circle (0.08);
            \filldraw[shift = {(6,0)}] (1.9,0) circle (0.08);
            \filldraw[shift = {(4,0)}] (0.95,-1.64) circle (0.08);

        \end{tikzpicture}
    \end{center}
    \item One can easily verify that with the $\ptl$ operator, we have a $\bbZ$-graded chain complex. We denote 
    \be C_k(\mathcal{U}) = \bigoplus_{i_0 < i_1 < \cdots < i_k} \mathbb{Z}[U_{i_0} \cap \cdots \cap U_{i_k}] \ .  \ee
\end{enumerate}

Now we have the homology part, what we are more interested in is the cohomology part where the manifold $M$ map to various coefficients to become fields we adore. What properties do we hope the field configurations to satisfy?
\begin{enumerate}
    \item Fields are (at least local) sections $\Gamma (U,F)$, $F$ is some parameter space of our choosing.
    \item To make the section maps abstract while align with our intuition, for $V\subseteq U$, there should be a restriction map ${\rm Res}_{U,V} : \Gamma(U,F)\to \Gamma(V,F)$ st.
    \be  {\rm Res}_{U,U} = {\rm id}_{F} \ , \ \ \ {\rm Res}_{V,W} \circ {\rm Res}_{U,V} = {\rm Res}_{U,W} \ . \ee
    \item For some fields, we hope they can add each other at least on the same local patch, meaning for each $U_i$, $\Gamma(U_i,F)$ is an Abelian group. Then the restriction map should at least respect the addition structure, making each ${\rm Res}_{U,V}$ a homomorphism of Abelian groups,
    \be  {\rm Res}_{U,V} : {\rm Hom}_{\rm Ab} (\Gamma(U,F),\Gamma(V,F))  \ee
\end{enumerate}
The above language sums up to be\textbf{ what we want is a presheaf of Abelian groups on $M$}, such presheaf is denoted $\mathcal{F}$. Thus we define a \Cech p-cochain $C^p(M,\mathcal{U})$ ($\CU$ denotes the set of atlas) by assigning each $p-$simplex to an element in the presheaf of the simplex, i.e.
\be
f : U_{i_1,\cdots, i_p} \mapsto G\in \CF ({U_{i_1,\cdots, i_p}})  
\ee
And the differential operator $\delta_p$ is realized by ${\rm Res}$ and $\ptl$, for a $p$-cochain,
\be
(\delta_p f)(\sigma ) = \sum_{k=0,p+1} (-1)^{k}  {\rm Res}_{\ptl_k \sigma, \sigma } (f ) (\ptl_k \sigma)   \ \
\ee
where $\sigma\in C^p(M,\CU)$, and $\ptl_k (U_{i_1,\cdots,i_p}) \equiv U_{i_1,\cdots, \Hat{i}_k,\cdots,i_p} $. The \Cech cohomology group is thus defined as
\be  
H^k(\CU,\CF) \equiv \frac{{\rm ker}(\delta_k)}{{\rm im}(\delta_{k-1}) }  \ .
\ee

\section{Group Cohomology}
\label{App:GrpChmlg}
Group cohomology is a mathematical tool used to study the properties of groups and their actions. It is a branch of algebraic topology that deals with the cohomology of spaces with group actions. The cohomology groups of a group with coefficients in a module are defined as follows:

Let $G$ be a group and $M$ be a $G$-module. The $n$th cohomology group of $G$ with coefficients in $M$ is defined as the quotient group:

$$H^n(G, M) = \frac{\text{ker}(\delta_n)}{\text{im}(\delta_{n-1})}$$

where $\delta_n:H^n(G, M) \rightarrow H^{n+1}(G, M)$ is the coboundary map defined by:
\be\delta_n(f)(g_0, g_1, \ldots, g_n) = g_0f(g_1, \ldots, g_n) + \sum_{i=1}^{n}(-1)^if(g_0, \ldots, g_{i-1}g_i, \ldots, g_n) + (-1)^{n+1}f(g_0, \ldots, g_{n-1})\ee
for $f \in C^n(G, M)$, where $C^n(G, M)$ is the group of $n$-cochains of $G$ with coefficients in $M$.

Notably, there is a theorem stating the isomorphism between group cohomology and singular cohomology
\be
H^n_{{\rm grp}}(G,A) \cong  H^n_{{\rm sing}}({\rm B}G,A) \ ,
\ee
one may refer to \cite{weibel_1994} for a rigorous statement and proof.



\bibliographystyle{JHEP}
\bibliography{main}

\end{document}